\pdfoutput=1
\documentclass[11pt]{article}
\usepackage[T1]{fontenc}
\usepackage[letterpaper, margin=1in]{geometry}
\usepackage{graphicx} 
\usepackage{amsmath}
\usepackage{amssymb}
\usepackage{amsthm}
\usepackage{mathabx,mathrsfs}
\usepackage{mathtools}
\usepackage{bbm}
\usepackage[colorlinks=true, citecolor=blue, urlcolor=blue, pagebackref=true, linktocpage, plainpages=false, pdfpagelabels]{hyperref}

\usepackage{xcolor}
\usepackage{stmaryrd}
\usepackage{subcaption}
\usepackage{paralist}
\usepackage{booktabs}
\usepackage{enumitem}
\usepackage{undertilde}
\usepackage{aliascnt}
\usepackage{cleveref}
\usepackage{tikz-cd}
\usepackage{tikz}

\usepackage[textwidth=2.5cm]{todonotes}

\newtheorem{theorem}{Theorem}[section]
\newaliascnt{lemma}{theorem}\newtheorem{lemma}[lemma]{Lemma}\aliascntresetthe{lemma}
\newaliascnt{fact}{theorem}\newtheorem{fact}[fact]{Fact}\aliascntresetthe{fact}
\newaliascnt{corollary}{theorem}\newtheorem{corollary}[corollary]{Corollary}\aliascntresetthe{corollary}
\newaliascnt{proposition}{theorem}\newtheorem{proposition}[proposition]{Proposition}\aliascntresetthe{proposition}

\theoremstyle{definition}
\newaliascnt{definition}{theorem}\newtheorem{definition}[definition]{Definition}\aliascntresetthe{definition}
\newaliascnt{problem}{theorem}\aliascntresetthe{problem}
\newaliascnt{remark}{theorem}\newtheorem{remark}[remark]{Remark}\aliascntresetthe{remark}
\newaliascnt{conjecture}{theorem}\newtheorem{conjecture}[conjecture]{Conjecture}\aliascntresetthe{conjecture}
\newaliascnt{question}{theorem}\newtheorem{question}[question]{Question}\aliascntresetthe{question}
\newaliascnt{example}{theorem}\newtheorem{example}[example]{Example}\aliascntresetthe{example}

\crefname{theorem}{Theorem}{Theorems}
\crefname{lemma}{Lemma}{Lemmas}
\crefname{fact}{Fact}{Facts}
\crefname{corollary}{Corollary}{Corollaries}
\crefname{proposition}{Proposition}{Propositions}
\crefname{definition}{Definition}{Definitions}
\crefname{problem}{Problem}{Problems}
\crefname{remark}{Remark}{Remarks}
\crefname{conjecture}{Conjecture}{Conjectures}
\crefname{question}{Question}{Questions}
\crefname{example}{Example}{Examples}

\renewcommand{\setminus}{\smallsetminus}

\newcommand{\dif}{\mathop{}\!\mathrm{d}}

\newcommand{\uF}{\utilde{F}}
\newcommand{\uf}{\utilde{f}}

\newcommand{\ov}{\overline}

\newcommand{\bbC}{\mathbb{C}}
\newcommand{\bbF}{\mathbb{F}}

\newcommand{\bbK}{\mathbb{K}}

\newcommand{\bbN}{\mathbb{N}}

\newcommand{\bbQ}{\mathbb{Q}}
\newcommand{\bbR}{\mathbb{R}}
\newcommand{\bbZ}{\mathbb{Z}}

\newcommand{\sfZ}{\mathsf{Z}}

\newcommand{\Par}{\mathscr{P}}
\newcommand{\Pars}{\mathscr{P}_{\mathrm{s}}}

\newcommand{\calP}{\mathcal{P}}
\newcommand{\calQ}{\mathcal{Q}}
\newcommand{\calR}{\mathcal{R}}
\newcommand{\calV}{\mathcal{V}}
\newcommand{\calW}{\mathcal{W}}
\newcommand{\calX}{\mathcal{X}}
\newcommand{\calZ}{\mathcal{Z}}

\newcommand{\specMM}{\Delta_{\mathrm{MM}}}

\newcommand{\degen}{\trianglelefteq}
\newcommand{\neged}{\trianglerighteq}

\newcommand{\id}{\mathrm{id}}

\newcommand{\ang}[1]{\langle #1 \rangle}
\newcommand{\Ang}[1]{\left\langle #1 \right\rangle}

\newcommand{\floor}[1]{\lfloor #1 \rfloor}
\newcommand{\Floor}[1]{\left\lfloor #1 \right\rfloor}
\newcommand{\ceil}[1]{\lceil #1 \rceil}
\newcommand{\Ceil}[1]{\left\lceil #1 \right\rceil}

\newcommand{\Mat}{\mathrm{Mat}}

\newcommand{\fourcycle}[5][1.25]{%
  \tikz[baseline=-0.7ex, x=1.35ex, y=1.35ex, scale=#1]{
    \node[inner sep=0pt] (L) at (-2.0,0) {$1$};
    \node[inner sep=0pt] (R) at ( 2.0,0) {$2$};
    \node[inner sep=0pt] (U) at (0, 1.6) {$3$};
    \node[inner sep=0pt] (D) at (0,-1.6) {$4$};

    \draw[line width=0.45pt] (L) -- node[pos=.5, above left,  inner sep=0.4pt] {$\scriptstyle #2$} (U);
    \draw[line width=0.45pt] (U) -- node[pos=.5, above right, inner sep=0.4pt] {$\scriptstyle #3$} (R);
    \draw[line width=0.45pt] (R) -- node[pos=.5, below right, inner sep=0.4pt] {$\scriptstyle #4$} (D);
    \draw[line width=0.45pt] (D) -- node[pos=.5, below left,  inner sep=0.4pt] {$\scriptstyle #5$} (L);
  }%
}

\DeclareMathOperator{\Rk}{R}

\DeclareMathOperator{\AR}{\utilde{\mathrm{R}}}
\DeclareMathOperator{\AQ}{\utilde{\mathrm{Q}}}
\DeclareMathOperator{\SR}{SR}
\DeclareMathOperator{\ASR}{\utilde{\mathrm{SR}}}
\DeclareMathOperator{\Q}{Q}

\DeclareMathOperator{\PR}{PR}
\DeclareMathOperator{\CR}{CR}
\DeclareMathOperator{\ACR}{\utilde{\mathrm{CR}}}
\DeclareMathOperator{\NCR}{NCR}

\DeclareMathOperator{\GL}{GL}

\DeclareMathOperator{\eH}{H}

\DeclareMathOperator*{\bbE}{\mathbb E} 
\DeclareMathOperator*{\argmax}{arg\,max} 
\DeclareMathOperator{\supp}{supp}
\DeclareMathOperator{\Rep}{Rep}
\DeclareMathOperator{\Span}{span}

\newcommand{\email}[1]{\href{mailto:#1}{\texttt{#1}}}

\title{The edge of the asymptotic spectrum of tensors}
\author{Josh Alman\thanks{Columbia University. \email{josh@cs.columbia.edu}. Supported in part by NSF Grant CCF-2238221 and a Packard Foundation Fellowship.} \and
    Baitian Li\thanks{Columbia University. \email{bl3052@columbia.edu}. Supported in part by NSF Grant CCF-2238221, a Packard Foundation Fellowship, and a Columbia SEAS Presidential Fellowship.} \and
    Kevin Pratt\thanks{Columbia University. \email{ktp2116@columbia.edu}. Supported in part by NSF Grant CCF-2238221 and a Packard Foundation Fellowship.}}
\date{}

\begin{document}

\maketitle

\begin{abstract}
Strassen founded the theory of the asymptotic spectrum of tensors to study the complexity of matrix multiplication. Since then, this framework has found broader applications in theoretical computer science and combinatorics.

A central challenge in this theory is to explicitly construct new spectral points. In Crelle~1991, Strassen proposed the upper support functionals $\zeta^\theta$ as candidate spectral points, where $\theta$ ranges over a triangle $\Theta$. Recent progress, involving tools and ideas from quantum information theory (Christandl--Vrana--Zuiddam, STOC 2018, JAMS 2023) and convex optimization (Hirai, 2025), culminated in the proof that the upper support functionals are indeed spectral points over the complex numbers (Sakabe--Do\u{g}an--Walter, 2026).

In this paper, we give an even clearer picture of the situation for support functionals when $\theta$ lies along the edges of the triangle. We show that not only are these functionals spectral points, but that they are uniquely determined as spectral points by their behavior on matrix multiplication tensors. As our methods are algebraic, as a corollary this gives the first explicit nontrivial spectral points over arbitrary fields. Previously, the only spectral points that were known in positive characteristic were the three flattening ranks, corresponding to the three vertices of the triangle.

As part of our argument, we show a close connection between the edge support functionals and \emph{Harder--Narasimhan filtrations} from quiver representation theory. We thus show, using recent work in algorithmic invariant theory, that these support functionals can be computed in deterministic polynomial time. Other ingredients of our proof include a new criterion for abstractly characterizing asymptotic tensor ranks by spectral points, and a characterization of the edge support functionals in terms of matrix multiplication capacity. As another application of these tools, we prove the existence of spectral points for higher-mode tensors beyond those currently known.
\end{abstract}

\pagenumbering{gobble}
\newpage

\tableofcontents

\newpage

\pagenumbering{arabic}

\section{Introduction}

Let $U, V, W$ be finite-dimensional vector spaces over a field $\bbF$. Our main objects of study are $3$-mode tensors, namely elements $T \in U\otimes V\otimes W$. We say that a tensor $S \in U'\otimes V'\otimes W'$ \emph{restricts} to $T$ if there exist linear maps $A\colon U'\to U$, $B\colon V'\to V$, and $C\colon W'\to W$ such that $T = (A\otimes B\otimes C)S$; we write this as $S\geq T$.

The problem of determining whether one tensor restricts to another captures questions from several areas, including the complexity of matrix multiplication in theoretical computer science \cite{Str69Gaussian,Str88AsymSpec,CW90Laser,ADWXXZ25MoreAsym}, progression-free sets \cite{Tao16Capset,TS16SliceRank,CLP17Progression,EG17Capset,BCCGNSU17Capset,KSS18Capset} and sunflowers \cite{NS17Sunflower} in combinatorics, and entanglement transformations \cite{DVC00QI,WDGC13QI,Zuiddam18Thesis,Christandl24QI} in quantum information theory.

\paragraph{Notions of ranks.}
The aforementioned questions boil down to understanding restrictions between a given tensor and the diagonal tensors $\ang{n} = \sum_{i=1}^n x_iy_iz_i\in \bbF^n\otimes \bbF^n\otimes \bbF^n$, which are captured by the notions of tensor rank and subrank. The \emph{rank} $\Rk(T)$ of $T$ is the smallest $r$ such that $\ang{r} \geq T$, and the \emph{subrank} $\Q(T)$ of $T$ is the largest $q$ such that $T \geq \ang{q}$. For example, the rank of the matrix multiplication tensors determines the complexity of matrix multiplication \cite{Str69Gaussian}, while the subrank of the structural tensor of a hypergraph upper-bounds the size of a \emph{hypergraph induced matching} \cite{AVZ20Match}; this framework includes all of the combinatorial problems mentioned above and controls the size of the corresponding structures.

Beyond tensor rank, many other rank notions for tensors have been studied, including slice rank \cite{Tao16Capset,TS16SliceRank}, analytic rank \cite{GW11AR,Lov19AR}, geometric rank \cite{KMZ23GR}, $G$-stable rank \cite{Derksen22GStableRank} and (non-)commutative rank \cite{FR04NCRk}. Like tensor rank, these notions have numerous connections to computer science and combinatorics. The two most important examples for us are:
\begin{itemize}
    \item \emph{Commutative rank.} The problem of finding a deterministic polynomial-time algorithm for computing commutative rank $\CR$ is a notorious open problem in the study of deterministic \emph{polynomial identity testing (PIT)}; such an algorithm would imply circuit lower bounds that seem far beyond our current reach \cite{KI04PIT}. Motivated by this, a series of recent breakthroughs showed that if one replaces commutative rank with non-commutative rank, then efficient algorithms \emph{do} exist, leading to deterministic solutions of the \emph{non-commutative} version of PIT \cite{GGOW16PIT, GGOW20PIT, IQS18NCR, HH21NCR, ACM24NCR}.
    \item \emph{Slice rank.} In hindsight, the breakthrough upper bound on the size of \emph{cap sets} \cite{EG17Capset} in additive combinatorics comes from upper-bounding subrank using slice rank $\SR$ \cite{Tao16Capset,TS16SliceRank}. This perspective was later extended to prove strong bounds for progression-free sets \cite{CLP17Progression,BCCGNSU17Capset,KSS18Capset} and sunflowers \cite{NS17Sunflower}. In computer science, these results are closely linked to recently identified \emph{barriers} for current approaches to bounding $\omega$ \cite{BCCGNSU17Capset, Alm21Barrier, CVZ21Barrier, CLLZ25Rect}.
\end{itemize}

\paragraph{Asymptotics.}
All of the questions above ultimately concern the ranks and subranks of a growing family of tensors arising from tensor powers of a fixed base case $T^{\otimes n}$. It is therefore natural to study the asymptotic behavior of these ranks, and in particular their regularizations, such as the \emph{asymptotic rank}
\[ \AR(T) \coloneq \lim_{n\to\infty} \Rk(T^{\otimes n})^{1/n}, \]
and similarly one can define the regularizations $\AQ, \ACR, \ASR$ of $\Q, \CR, \SR$, respectively. One can also define the \emph{asymptotic restriction} between given tensors $T$ and $S$ by
\[ T^{\otimes n} \leq S^{\otimes n+o(n)}, \]
which we denote by $T \lesssim S$.

Interestingly, geometric rank, analytic rank, and $G$-stable rank are each equivalent to slice rank up to a constant factor \cite{CM22Equiv}. There is a similar relationship between non-commutative rank and commutative rank \cite{FR04NCRk}. As a consequence, the regularizations of these ranks collapse to the following four quantities, with the known relations
\begin{equation} \label{eqn:asymp-relations}
    {\AQ} \quad \leq \quad {\ASR} \quad \leq \quad {\ACR} \quad \leq \quad {\AR}.
\end{equation}

\paragraph{Asymptotic spectrum.}
In general, it is difficult to compute $\AR$ or $\AQ$, or to decide whether one tensor asymptotically restricts to another. To study this problem, Strassen developed the theory of the \emph{asymptotic spectrum} \cite{Str88AsymSpec,Str91SuppFunc,Str05Kom}. The asymptotic spectrum, denoted by $\calX$, consists of functionals $\phi$ from tensors to nonnegative reals such that $\phi$ is
\begin{itemize}
    \item additive under direct sum: $\phi(T\oplus S) = \phi(T) + \phi(S)$,
    \item multiplicative under tensor product: $\phi(T\otimes S) = \phi(T) \phi(S)$,
    \item monotone under restriction: $\phi(T) \leq \phi(S)$ for $T\leq S$,
    \item and normalized by the unit tensor: $\phi(\ang{1})=1$.
\end{itemize}
The main result, \emph{Strassen duality}, gives a dual characterization of asymptotic restriction in terms of $\calX$:
\begin{center}
    $T\lesssim S$ \quad 
    if and only if \quad 
    $\phi(T) \leq \phi(S)$ for all $\phi \in \calX$.
\end{center}
Moreover, the asymptotic rank and subrank can be characterized as
\begin{center}
    $\displaystyle \AR(T) = \max_{\phi\in\calX} \phi(T)$ \quad
    and \quad
    $\displaystyle \AQ(T) = \min_{\phi\in\calX} \phi(T)$.
\end{center}
Thus, in Strassen's words \cite[p.~105]{Str88AsymSpec}, a complete description of $\calX$ would \emph{``solve all problems concerning $\lesssim$''}. The difficulty is that Strassen's duality theorem is nonconstructive, so one still has to understand $\calX$ by explicitly constructing functionals in $\calX$, which are called \emph{spectral points}.

The axioms defining $\calX$ are very restrictive. A useful comparison, developed explicitly in \cite{RZ21CatB}, is with the notion of a \emph{formal complexity measure} in Boolean complexity \cite{Khr72Comp} (see also \cite[\S8]{Weg87Comp}, \cite[\S23.2.3]{AB09Complexity}). A formal complexity measure $\mu$ is a nonnegative function on Boolean functions $f\colon \{0,1\}^n\to\{0,1\}$ such that every literal $\ell\in \{x_i,\neg x_i\}$ has bounded value $\mu(\ell)\leq 1$. In contrast with Strassen's functionals, which are required to be exactly additive and multiplicative, a formal complexity measure only satisfies a subadditivity-type condition, namely $\mu(f\land g), \mu(f\lor g) \leq \mu(f) + \mu(g)$. The analogy continues at the level of duality: in Strassen's theory, the asymptotic tensor rank $\AR(T)$ is the maximum of $\phi(T)$ over all $\phi\in\calX$, while in Boolean complexity, the formula complexity $\rho(f)$ is the maximum of $\mu(f)$ over all formal complexity measures $\mu$. However, the latter statement is vacuous because $\rho$ itself is a formal complexity measure and dominates every other one. By contrast, although $\AR$ is defined as a maximization over spectral points, it is not itself a spectral point.

These strong axioms have two opposite consequences. On the positive side, they suggest that $\calX$ should be a rather small and rigid space. Even without any explicit description of the asymptotic spectrum, one can already derive nontrivial consequences purely from this rigidity. For instance, Strassen's theory admits an abstract formulation over general semirings rather than only tensors, and the following two examples primarily use the formal structure of asymptotic spectra, without much detailed knowledge of the underlying space $\calX$:
\begin{itemize}
    \item Zuiddam \cite{Zui19Graph} introduced an asymptotic spectrum for graphs, and showed that it captures the notion of Shannon capacity from information theory. This has been used to prove new properties of Shannon capacity (e.g.,~\cite[Example 4.6]{WZ26Spectra}) and give improved bounds on the Shannon capacity of certain graphs~\cite{dBBZ24Shannon}.
    \item Alman and Li \cite{AL25Depth2} used Strassen duality to derive a duality theorem for depth-$2$ circuits, resolving a problem about how to optimally combine fixed constructions in that model.
\end{itemize}

On the other hand, those same restrictions make it substantially harder to construct explicit elements of $\calX$, since one must verify several strong compatibility properties simultaneously. For example, we have known that tensor rank is not multiplicative ever since Strassen's matrix multiplication bound $\Rk(\ang{2,2,2})\leq 7$ \cite{Str69Gaussian}\footnote{Indeed, $\Rk(\ang{2,1,1})=\Rk(\ang{1,2,1})=\Rk(\ang{1,1,2})=2$, but $\ang{2,1,1} \otimes \ang{1,2,1}\otimes \ang{1,1,2} = \ang{2,2,2}$.}, and it was more recently shown that tensor rank is not additive either~\cite{Shi19Counter}. Likewise, all of the rank notions mentioned above, as well as their regularizations, fail to be multiplicative.

The most basic spectral points come from \emph{flattenings}: one may regard $T$ as an element of $U\otimes (V\otimes W)$, hence as a $\dim(U)\times \dim(V\otimes W)$ matrix, and take its matrix rank. Flattening along each of the three modes gives three spectral points $\zeta_{(1)}, \zeta_{(2)}$, and $\zeta_{(3)}$, which Strassen called the \emph{gauge points (Eichpunkte)} \cite[p.~120]{Str88AsymSpec}, also known as \emph{flattening ranks}.

\paragraph{Support functionals.}
Although these gauge points are rather coarse, Strassen's star-convexity theorem \cite[Theorem~6.5]{Str88AsymSpec} predicts that they should admit a much richer interpolation. Let
\[ \Theta \coloneq \{\theta \in \bbR_{\geq 0}^3 : \theta_1+\theta_2+\theta_3 = 1\}. \]
Then for every $\theta\in \Theta$, there exists a spectral point $\phi\in \calX$ such that
\begin{equation}
    \phi(\ang{E,H,L}) = E^{\theta_3+\theta_1} H^{\theta_1+\theta_2} L^{\theta_2+\theta_3} \label{eqn:itp_mm},
\end{equation}
where $\ang{E,H,L}$ denotes the tensor for $E \times H$ by $H \times L$ matrix multiplication. For each $\varkappa \in \{1,2,3\}$, the extreme case $\theta_\varkappa = 1$ is realized by the gauge point $\zeta_{(\varkappa)}$.

Strassen \cite{Str91SuppFunc} later proposed the notion of the \emph{support functional}, which can be viewed as a concrete realization of this interpolation idea. He defined the upper support functional
\begin{equation} \label{eqn:supp}
    \zeta^\theta(T) \coloneq \min_{\mathrm{bases~of~} U,V,W} \max_{P \in \calP(\supp T)} 2^{\sum_{\varkappa=1}^3 \theta_\varkappa \eH(P_\varkappa)},
\end{equation}
where $\supp(T) \subseteq [\dim U] \times [\dim V] \times [\dim W]$ is the support of $T$ with respect to the chosen bases, $\calP(\Phi)$ denotes the set of probability distributions on a finite set $\Phi$, $P_\varkappa$ is the marginal of $P$ on the $\varkappa$-th coordinate, and $\eH(\cdot)$ is the Shannon entropy. When $\theta_\varkappa=1$, the functional $\zeta^\theta$ reduces to the gauge point $\zeta_{(\varkappa)}$, and in general $\zeta^\theta$ satisfies \eqref{eqn:itp_mm}.

Strassen verified all of the axioms needed for $\zeta^\theta$ to be a spectral point, except for multiplicativity: in general he proved only that $\zeta^\theta$ is \emph{submultiplicative}. He did, however, prove multiplicativity for \emph{oblique (schr\"{a}g)} tensors \cite[\S4]{Str91SuppFunc}, a class that contains many tensors of interest in algebraic complexity theory.

\paragraph{Quantum functionals.} After Strassen introduced the support functionals, progress on explicit spectral points remained limited for nearly two decades. A major breakthrough eventually came from quantum information theory. In their influential work, Christandl, Vrana, and Zuiddam~\cite{CVZ18STOC,CVZ23QFun} defined the notion of a \emph{quantum functional}. For every $\theta\in \Theta$, they defined
\[ F_\theta(T) = \max_{P\in \Delta(T)} 2^{\sum_{\varkappa=1}^3 \theta_\varkappa \eH(P_\varkappa)} \]
where $\Delta(T)$ is the \emph{moment polytope} (also called the \emph{entanglement polytope}) associated to $T$, a convex polytope contained in the space of distributions $\calP([\dim U] \times [\dim V] \times [\dim W])$. The moment polytope admits several different characterizations, two of which are central to their proof. On the one hand, after normalizing $T$ and viewing it as a tripartite quantum state, $\Delta(T)$ describes the feasible quantum marginals of states that can be obtained from $T$ by \emph{stochastic local operations and classical communication (SLOCC)} \cite{WDGC13QI}. On the other hand, $\Delta(T)$ can be described representation-theoretically via the decomposition of tensor powers coming from \emph{Schur--Weyl duality}. The bridge between these viewpoints is \emph{geometric invariant theory (GIT)}, in particular the \emph{Kempf--Ness theory} available over the complex numbers \cite{KN78GIT}.

Using these complementary descriptions, Christandl--Vrana--Zuiddam proved that each $F_\theta$ is indeed a \emph{universal} spectral point. Here ``universal'' means that the required additivity and multiplicativity properties hold for all tensors, whereas Strassen's support functionals had only been shown to be multiplicative on a special subclass.

Quantum functionals also characterize the remaining notions in \eqref{eqn:asymp-relations}. In particular, the asymptotic slice rank always exists and satisfies \cite[\S5]{CVZ23QFun}
\begin{equation} \label{eqn:ASR_quant}
    \ASR(T) = \min_{\theta\in\Theta} F_\theta(T).
\end{equation}
Follow-up work \cite{CLZ23WSR}  showed that the asymptotic commutative rank admits a similar description:
\[ \ACR(T) = \min_{\rho\in [0, 1]} F_{(\rho,1-\rho,0)}(T). \]
Thus, all known asymptotic quantities in \eqref{eqn:asymp-relations} are characterized by a known subset of spectral points, rather than by genuinely new ones. Accordingly, it remained consistent with the available evidence that the quantum functionals $\{F_\theta\}_{\theta\in\Theta}$ might already exhaust the asymptotic spectrum of tensors over $\bbC$. Such a statement would already imply Strassen's asymptotic rank conjecture \cite{Str94ARC}, namely that every $n\times n\times n$ tensor has asymptotic rank at most $n$. This would have striking algorithmic consequences, including $\omega=2$, and would also refute the \emph{set cover conjecture} from fine-grained complexity theory~\cite{BK24ARC,Pratt24SCC,BCHKP25ARC,BKKN25ARC}.

\paragraph{Support functionals strike back.} Very recently, Sakabe, Do\u{g}an, and Walter \cite{SDW26Fun} proved a striking result: the upper support functional actually \emph{coincides} with the quantum functional, that is, $\zeta^\theta = F_\theta$. Consequently, the support functionals introduced by Strassen more than three decades ago were already universal spectral points; what was missing was a proof. The argument of Sakabe--Do\u{g}an--Walter is far from elementary and relies on tools unavailable to Strassen, including a recent duality theorem \cite{Hirai25Dual} from geodesically convex optimization.

Given this equivalence, one obtains another interesting consequence: by a standard application of the minimax theorem to \eqref{eqn:ASR_quant}, Sakabe--Do\u{g}an--Walter derived an alternative characterization of asymptotic slice rank \cite[Corollary~1.2]{SDW26Fun}. (A similar characterization also holds for asymptotic commutative rank.)
\begin{equation} \label{eqn:ASR_supp}
    \ASR(T) = \min_{\mathrm{bases~of~}U,V,W} \max_{P\in\calP(\supp T)} 2^{\min(\eH(P_1), \eH(P_2), \eH(P_3))}.
\end{equation}

The right-hand side of \eqref{eqn:ASR_supp} has been known to serve as an upper bound since the original introduction of the notion \cite[Proposition~6]{TS16SliceRank}. What is new here is that this upper bound is in fact \emph{tight}. To appreciate the content of \eqref{eqn:ASR_supp}, replace $T$ by its tensor power $T^{\otimes n}$. The right-hand side then upper bounds
\[ \ASR(T^{\otimes n}) = \ASR(T)^n, \]
but is allowed to choose bases in the much larger spaces $U^{\otimes n}$, $V^{\otimes n}$, and $W^{\otimes n}$, potentially producing support sets that are not tensor powers of the support obtained in the base case. Thus \cref{eqn:ASR_supp} exhibits a kind of \emph{parallel repetition} phenomenon \cite{Raz98Repet}: passing to tensor powers and optimizing over arbitrary bases yields no additional advantage.

\paragraph{Limitations of previous work.}
Despite the recent progress on the asymptotic spectrum, all of the results discussed above are currently proved only over the complex numbers $\bbC$. It is natural to ask whether these arguments can be adapted to other fields. Since support functionals are defined over an arbitrary field from the outset, one might hope to proceed in two steps:
\begin{enumerate}
    \item find an analogue of the quantum functional over a general field $\bbF$, and extend the proof of \cite{CVZ23QFun} to show that these new functionals are universal spectral points;
    \item then generalize \cite{SDW26Fun} to deduce the universality of the support functionals.
\end{enumerate}
At present, however, there is no straightforward way to carry out this program, although a proposal is made in \cite{CLZ23WSR}. Any extension beyond $\bbC$ would have to overcome genuinely new obstacles.

One obstacle comes from the proof of Christandl--Vrana--Zuiddam \cite{CVZ23QFun}. Their proof of universality uses the quantum-marginal description of the moment polytope, whose connection to the representation-theoretic description relies on Kempf--Ness theory and hence on analytic features specific to $\bbC$. For this reason, the argument does not appear to generalize directly to other fields. A second obstacle comes from the argument of Sakabe--Do\u{g}an--Walter \cite{SDW26Fun}, which uses the smooth structure of $\bbR$ together with tools from differential geometry and geodesically convex optimization.

We note that, since the definition of the upper support functional is purely algebraic, model-theoretic arguments allow one to transfer the universality of support functionals between different algebraically closed fields of the same characteristic. In particular, once universality is proved over $\bbC$, it automatically extends to every algebraically closed field of characteristic zero. Unfortunately though, such arguments cannot transfer results in characteristic zero to characteristic $p$.
\paragraph{Why positive characteristic?}

The reader may wonder why the positive-characteristic case deserves special attention. One reason is conceptual. The upper support functional is defined in purely algebraic terms, so even over $\bbC$ it is natural to seek an algebraic proof of its universality. Such a proof would likely reveal additional structure of tensors that is invisible in the current analytic proofs. Beyond this intrinsic motivation, there are also concrete reasons coming from theoretical computer science and combinatorics.

\emph{Algorithms and complexity.} When designing an algebraic algorithm, the characteristic of the field can genuinely affect algorithmic power. For example, Bj\"{o}rklund's \cite{Bjo14Hamil} $O(1.66^n)$-time algorithm for undirected Hamiltonicity reduces the problem to a carefully designed PIT instance, and its key cancellation mechanism is specific to characteristic $2$. Positive characteristic can also give rise to new challenges. For instance, Kaltofen's polynomial-time factoring algorithm \cite{Kal89Factor} works in characteristic zero, but in characteristic $p$ one encounters additional difficulties with polynomials of the form $f^{p^k}$. This obstruction is one of the reasons that positive characteristic creates difficulties for the \emph{algebraic hardness versus randomness} program \cite{KI04PIT,And20PIT}.

Returning to tensor-related problems, Sch\"onhage \cite[Theorem~2.8]{Sch81Tau} observed that the matrix multiplication exponent $\omega$, while \emph{a priori} depending on the ground field, in fact depends only on its characteristic. So far, the main upper-bound constructions for $\omega$ happen to work uniformly over all fields. However, there are known bounds on the ranks of small matrix multiplication tensors which work only over characteristic 2~\cite{Deepmind22MM,KM23Flip}, as well as bounds which work only over characteristic different from 2~\cite{DPS25MM4}. A better understanding of tensors in different characteristics could still provide new tools, or rule out unpromising approaches, in the search for faster matrix multiplication.

\emph{Combinatorics.} For induced-hypergraph problems such as cap sets and progression-free sets, the core of the slice-rank upper-bound argument is to carefully choose a field over which the slice rank of $T$ becomes deficient. In the cap-set problem, for instance, one analyzes the structure tensor of $\bbZ_3^n$, whose asymptotic slice rank degenerates \emph{only} in characteristic $3$~\cite[Remark~6]{Tao16Capset}.

The characterization \eqref{eqn:ASR_supp} shows that the support-based method used in earlier work to upper bound asymptotic slice rank is in fact tight. It also shows that asymptotic slice rank, despite being defined through an infinite limit, is a \emph{computable} quantity. At present, however, \eqref{eqn:ASR_supp} is proved only over $\bbC$, even though the statement itself looks entirely natural from an algebraic point of view. A better understanding in characteristic $p$ might therefore extend \eqref{eqn:ASR_supp} to all characteristics, and in turn provide a firmer foundation for related problems in additive combinatorics.

\subsection{Our results}

In this work, we make progress on understanding the asymptotic spectrum of tensors over arbitrary fields (including $\mathbb{C}$). We obtain several structural results concerning the \emph{edge} of the spectrum, including that it is uniquely characterized by matrix multiplication.

\paragraph{The edge of the support functionals.} For $\varkappa\in [3]$, let
\[ \Theta_{(\varkappa)} = \{ \theta \in \Theta : \theta_\varkappa = 0 \}. \]
Since $\Theta$ is a triangle, these are precisely its three edges. In some discussion below, we focus without loss of generality on one of the edges, $ \Theta_{(3)} = \{ (\rho, 1-\rho, 0) : \rho\in[0, 1] \}. $

Our first result shows that support functionals on the boundary are universal over \emph{every} field.
\begin{theorem}[see \Cref{cor:edge-universal}]\label{cor:edge-universal-intro}
    For any field $\bbF$ and any $\theta \in \Theta_{(\varkappa)}$, the upper support functional $\zeta^\theta$ is a universal spectral point for $\bbF$-tensors.
\end{theorem}
Before our work, for fields of positive characteristic, the only explicit spectral points that were known were the three gauge points $\zeta_{(\varkappa)}$.

Our second result gives a more structural description of these edge functionals. Recall that $\zeta^\theta(T)$ is defined by first minimizing over all choices of bases and then maximizing a weighted marginal entropy over the resulting support. We show that, along a fixed edge $\Theta_{(\varkappa)}$, this outer minimization can be realized by a single basis independent of $\theta$. This preferred basis is obtained by extending a filtration on the relevant modes, namely the influential \emph{Harder--Narasimhan filtration (HN-filtration)} \cite{HN75Filt,HdlP02HNFilt} from quiver representation theory.
\begin{theorem}[informal, see \Cref{cor:edge-hn-formula}]\label{cor:edge-hn-formula-intro}
    For any tensor $T$ and any $\varkappa\in [3]$, there exists a choice of basis such that the outer minimization in \eqref{eqn:supp} is attained by that same basis for every $\theta\in \Theta_{(\varkappa)}$.

    Specifically, let $\{(n_u, m_u)\}_{1 \le u \le r}$ be the dimension data of the HN-filtration of $T$. Then for $\theta = (\rho, 1-\rho, 0)$, we have \[\zeta^\theta(T) = \sum_{u=1}^r n_u^\rho m_u^{1-\rho}.\]
\end{theorem}

The HN-filtration has recently found connections to other topics in computer science, including algorithmic invariant theory~\cite{Cheng24HN,IOS25HN}, topological data analysis~\cite{FJ26Skyscraper}, and submodular optimization~\cite{IOS25HN}. It is also closely related to non-commutative rank, and recent work in that direction showed that one can compute HN-filtrations in deterministic polynomial time for the class of quivers relevant to us \cite{Cheng24HN,IOS25HN}. Consequently, the edge support functionals are not merely computable, but \emph{efficiently computable}.

\begin{corollary}[informal, see \Cref{cor:alg}]
    For any field and any $\theta\in \Theta_{(\varkappa)}$, the support functional $\zeta^\theta$ can be computed in deterministic polynomial time.
\end{corollary}

These edge functionals also recover asymptotic commutative rank. More precisely, one first fixes a mode; for definiteness we use the $W$-mode. Over $\bbC$, Sakabe--Do\u{g}an--Walter proved \cite[Theorem~5.14]{SDW26Fun} that
\begin{equation} \label{eqn:acr_supp}
    \ACR(T) = \min_{\theta \in \Theta_{(3)}} \zeta^\theta(T).
\end{equation}

We show that the same identity holds over arbitrary fields.
\begin{theorem}[see \Cref{thm:acr_formula}]
    \Cref{eqn:acr_supp} holds over any field.
\end{theorem}

One striking feature of all of these results is their \emph{robustness} with respect to the ground field. We do not assume that $\bbF$ is algebraically closed, perfect, or even large; in particular, the results also hold over finite fields. This is notable because, \emph{a priori}, extending the field may make additional support sets attainable through changes of basis, and several algebro-geometric methods in the literature do require hypotheses such as algebraic closedness or perfectness, for instance in the theory of $G$-stable rank \cite[\S1B]{Derksen22GStableRank}. Field-extension invariance was previously known for non-commutative rank, and our work reveals a close connection between non-commutative rank and the edge support functionals.

\paragraph{The edge of the spectrum of matrix multiplication.}

Beyond proving that these edge functionals exist as spectral points, we obtain a stronger characterization in terms of matrix multiplication tensors. Note that if $\theta=(\rho,1-\rho, 0)\in \Theta_{(3)}$, then
\[ \zeta^\theta(\ang{E,H,L}) = E^\rho H L^{1-\rho}. \]
Therefore, for any tensor $T$ with $T \geq \ang{E,H,L}$, monotonicity immediately gives $\zeta^\theta(T) \geq E^\rho H L^{1-\rho}$. Our next theorem shows that, asymptotically, this already characterizes $\zeta^\theta$:
\begin{theorem}[see \Cref{cor:edge-mm-char}]\label{cor:edge-mm-char-intro}
    For $\theta \in \Theta_{(\varkappa)}$, the value of $\zeta^\theta(T)$ is given by the most valuable matrix multiplication tensor asymptotically restrictable from $T$:
    \[ \zeta^\theta(T) = \sup \{ \zeta^\theta(\ang{E,H,L})^{1/n} : T^{\otimes n} \geq \ang{E,H,L} \}. \]
\end{theorem}

As a consequence, we obtain a uniqueness statement inside the asymptotic spectrum.
\begin{theorem}[see \Cref{cor:edge-unique}]\label{cor:edge-unique-intro}
    If a spectral point $\phi$ has $\phi(\ang{1,n,1}) = n$, then there exists $\theta \in \Theta_{(3)}$ such that $\phi = \zeta^\theta$.
\end{theorem}

This can be viewed as progress toward \emph{exhausting} the asymptotic spectrum. Even if future work were to prove that more support functionals are universal, or to construct genuinely new spectral points, one still needs ways to rule out other hypothetical spectral points. Our uniqueness theorem does exactly this along one edge. We emphasize that the full asymptotic spectrum $\calX$ could \emph{a priori} be very large, even infinite-dimensional. What our result demonstrates is that one can sometimes translate structural properties of tensors into rigidity statements about $\calX$ itself, thereby excluding classes of hypothetical spectral points. Note that even for the case of tensors over $\bbC$, this is a stronger conclusion than that of \cite{SDW26Fun} for edge functionals.

\paragraph{Towards more spectral points.} The spectral characterizations of asymptotic slice rank and asymptotic (non-)commutative rank suggest a broader principle: for many \emph{reasonable} rank-like quantities, their regularization should arise either as a maximization or as a minimization over a suitable subset of the asymptotic spectrum. We confirm this principle by giving a convenient criterion for proving such characterizations.

\begin{theorem}[informal, see \Cref{prop:asymp-F,prop:f-asymp}] \label{thm:asymp-F-intro}
    Let $F$ be a ``reasonable'' rank parameter which is additive and supermultiplicative. Then, its regularization is a minimization of spectral points, i.e., there is a subset $\calZ \subseteq \calX$ such that \[ \uF(T) = \min_{\phi \in \calZ} \phi(T).\] The analogous result also holds with ``supermultiplicative'' replaced by ``submultiplicative,'' replacing the minimization with maximization.
\end{theorem}

We remark that for the characterization of submultiplicative rank parameters, our result is very similar to the work of Bug\'ar and Vrana \cite{BV25Func}, see \Cref{rmk:bv25} for further discussion. However, we take a different approach, which also allows us to characterize supermultiplicative rank parameters, a case not covered by previous methods. This side appears to have more fruitful applications: it applies to many tensor ranks of interest and provides new insights into higher-order tensors.

For $d$-mode tensors where $d\geq 4$, a recent preprint of Moshkovitz and Zhu \cite{MZ26CRank} introduced a multilinear generalization of commutative rank, which we still denote by $\CR(T)$. Using Theorem~\ref{thm:asymp-F-intro}, we obtain the following \emph{a priori characterization} of asymptotic commutative rank.
\begin{corollary}[see \Cref{cor:higher-mode-acr-minimize}] \label{cor:higher-mode-acr-minimize-intro}
    The asymptotic (multilinear) commutative rank $\ACR$ is a minimization of spectral points.
\end{corollary}

Even over $\bbC$, higher-order quantum functionals are not completely understood\footnote{In \cite{CVZ23QFun}, quantum functionals were not defined directly in terms of moment polytopes. Instead, the authors separately defined \emph{lower} and \emph{upper} quantum functionals using quantum-marginal and representation-theoretic descriptions, respectively; both definitions are also associated with certain classes of bipartitions of the $d$ modes. The quantum functionals discussed above correspond only to the \emph{singleton partitions}. For these partitions, the moment-polytope formulation applies and the lower and upper functionals coincide. See also the discussion in \Cref{sec:higher-mode} and the recent work \cite{BCFLN26QFunc}.}. Moreover, we prove that higher-order quantum functionals are not sufficient to characterize this higher-order $\ACR$:

\begin{proposition}[informal, see e.g.,~\Cref{thm:separation}] \label{prop:high-order-intro}
    For $d \geq 4$, the asymptotic spectrum of $d$-mode tensors contains points other than the quantum functionals.
\end{proposition}
In other words, there must be more spectral points we have not yet identified. We give two different proofs of Proposition~\ref{prop:high-order-intro} which focus on different tensors; one (\Cref{thm:separation}) using Corollary~\ref{cor:higher-mode-acr-minimize-intro}, and another (\Cref{rmk:interpolation}) using a higher-mode star-convexity result of~\cite[\S12]{WZ26Spectra}. We hope that our criterion will provide a useful guide for the further study and eventual discovery of new spectral points for higher-mode tensors.

Finally, along the way, we note that our framework may be applied to other rank notions (such as asymptotic slice rank and $G$-stable rank); see \Cref{cor:ASR_minimize} for details.

\paragraph{Concurrent work.} Independently and concurrently with our work, Botero, Christandl, Fraser, Leigh, and Nieuwboer \cite{BCFLN26QFunc} obtained results analogous to our \Cref{prop:high-order-intro} (see Section 4 therein). Their proof proceeds via a more detailed analysis of the higher-order (lower and upper) quantum functionals, which differs from our approach.

\subsection{Technical overview}

\paragraph{Abstract results.} We begin by establishing useful properties of abstract semirings in \Cref{sec:abstract_functionals}. These results tell us when a well-behaved functional can be recovered from the asymptotic spectrum. Consider an \emph{upper} functional $F$, which one may think of as an already regularized functional $\uF$ in the setting of \Cref{thm:asymp-F-intro}. What conditions are necessary for $F$ to be a maximization of spectral points? Clearly, $F$ must be monotone and normalized, but it need only be \emph{sub}additive and \emph{sub}multiplicative, and hence need not itself be a spectral point. One further compatibility condition is necessary: $F$ must satisfy the same scaling rule as every spectral point, namely $F(na^m)=nF(a)^m$. We call this property \emph{monomial homogeneity}. Our first abstract result (\Cref{thm:maximization}) shows that these conditions are also sufficient.

\emph{Proof overview of \Cref{thm:maximization}.}
The basic idea in proving sufficiency is to reduce the problem to Strassen duality. Since duality tells us that \emph{asymptotic rank} is always a maximization over spectral points, it is enough to find a refinement $\leq^F$ of the original preorder $\leq$ in which the asymptotic rank becomes exactly $F$, since the spectrum of the refinement is a subset of the original spectrum.

The natural way to do this is to enlarge the original preorder by formally declaring that every object is at most the integer predicted by $F$, namely $a\leq^F \ceil{F(a)}$, and then taking the smallest preorder generated by these relations. If this construction works, one would expect the corresponding asymptotic rank to reproduce $F$. The difficulty is to show that this enlarged preorder has not become too large. In particular, we need to prove that $F$ is still monotone for the new preorder, i.e.,
\[ s\leq^F t \implies F(s)\leq F(t). \]
A naive inductive argument quickly runs into trouble, because $F(s) \leq F(t)$ does not necessarily imply $F(s+u)\leq F(t+u)$. What we do instead is to refactor an arbitrary relation $s\leq^F t$ into a cleaner form to work with. It turns out that the most interesting case is the following form: $s$ and $t$ are connected by the following two relations
\[ s \leq b + cd, \qquad b+\ceil{F(c)}d \leq t, \]
because any relation in $\leq^F$ can in turn be written as a concatenation of the above relations. Moreover, the monomial homogeneity condition mentioned above is exactly what enables us to prove the monotonicity of $F$ under such relations. This yields our characterization of maximization.

\emph{The minimization side.}
The statement for lower functionals (which are superadditive and supermultiplicative) is proved in the same spirit. One now builds a preorder forcing $\floor{f(a)}\leq a$ and aims to recover $f$ from asymptotic subrank rather than asymptotic rank. The argument is again parallel, but the key inequality needs one additional hypothesis, which we call \emph{monomial dominance}.

\paragraph{Structure along the edges.}

We now study the upper support functional $\zeta^\theta$ in greater detail, with the goal of motivating the HN-filtration that appears in Theorem~\ref{cor:edge-hn-formula-intro}. This will also provide the structure we need in the proofs of our main results.

Recall that $\zeta^\theta$ is defined by first choosing bases, then maximizing a weighted marginal entropy over the resulting support, and finally minimizing over all choices of bases; see \eqref{eqn:supp}. On the edge $\theta=(\rho,1-\rho,0)$, this expression simplifies dramatically:
\[
\zeta^{(\rho,1-\rho,0)}(T)
= \min_{\text{bases}} \max_{P\in\calP(\supp T)} 2^{\rho \eH(P_1) + (1-\rho)\eH(P_2)}.
\]
Since the third weight is zero, the objective depends only on the first two marginals, so it no longer depends on the full $3$-dimensional support, but only on its projection to the $U\times V$ coordinates. In other words, once we restrict to an edge, the inner optimization becomes essentially a two-dimensional problem (\Cref{prop:support_proj}).

Our first step is therefore to analyze this weighted entropy maximization problem for a set $\Phi\subseteq J\times K$; this is carried out in \Cref{sec:marginal_entropy}. The resulting picture is strikingly rigid. The main structural statement is \Cref{prop:argmax_2d}: after reordering $J$ and $K$, an optimal support must decompose into a block upper-triangular pattern whose diagonal blocks are \emph{balanced}, meaning that each block supports a distribution with the prescribed uniform marginals. This balancedness condition itself admits a purely combinatorial Hall-type characterization, proved in \Cref{prop:balance_char}. Thus the optimization problem reveals a concrete discrete structure underlying the support functional.

The next step is to identify a basis for $T$ in which this discrete structure arises naturally. Here it is useful to reinterpret a tensor through its first two modes. Fixing the third mode gives a list of linear maps $U\to V$, or equivalently a matrix space; we discuss this viewpoint in \Cref{sec:quiver}. A standard way to package such a collection of maps is as a \emph{representation} of the \emph{Kronecker quiver}, but no prior background in quiver representation theory is really needed: in our setting, this is simply convenient language for discussing a tuple of matrices together with its invariant subspaces and quotient pieces.

From this viewpoint, the relevant notion is \emph{semistability}. We remark that this notion originates in the geometric invariant theory (GIT) of general quiver representations and has several very different characterizations. The only one we need here is linear-algebraic, and our arguments do not rely on the machinery of GIT. Our main observation consists of the following two correspondences:
\begin{itemize}
    \item Semistability can be characterized by the absence of \emph{shrunk subspaces} (\Cref{prop:unstable}): a tensor $T$ is semistable if and only if no subspace of $U$ is mapped by the matrix tuple into an unexpectedly small subspace of $V$. This in turn turns out (\Cref{lem:semistable_iff_balanced}) to be equivalent to saying that, for any choice of bases, the support of $T$, projected to the $U\times V$ coordinates, is balanced in the sense described above.
    \item The structural result known as the \emph{Harder--Narasimhan filtration} (\Cref{thm:HN-filtration}) provides a canonical way to decompose any representation into semistable pieces. In other words, this filtration is a nested sequence of linear subspaces of $U$ and $V$ that, after choosing bases adapted to the filtration, makes the projected support block triangular. Surprisingly, the combinatorial condition (\Cref{prop:argmax_2d}) from the optimization problem matches this filtration exactly.
\end{itemize}

These correspondences strongly suggest that the HN-filtration is the right basis for the support functional on the edge. It already exhibits the specific form that attains a maximizer in the entropy maximization and makes the value of the optimization explicit, and beyond the combinatorial structure, its additional linear-algebraic properties will enable the arguments that follow.

\paragraph{Comparing with matrix multiplication.}

We are finally prepared to prove Theorem~\ref{cor:edge-universal-intro}. In principle, with the HN-filtration characterization in hand, one could try to directly argue that $\zeta^{(\rho,1-\rho,0)}$ is a spectral point without making any reference to matrix multiplication. We instead go through matrix multiplication, because this yields the stronger characterization of $\zeta^\theta$ as the most valuable matrix multiplication tensor extractable from $T$ (Theorem~\ref{cor:edge-mm-char-intro}), and also implies the uniqueness statement (Theorem~\ref{cor:edge-unique-intro}).

The main tool introduced for this purpose, in \Cref{sec:value}, is the \emph{lower value functional}
\[
  \xi_\gamma(T) \;=\; \sup\bigl\{ \bigl(E^{\gamma_1} H^{\gamma_2} L^{\gamma_3}\bigr)^{1/N} : T^{\otimes N} \geq \ang{E,H,L} \bigr\},
\]
which measures the most valuable matrix multiplication tensor that can be asymptotically extracted from $T$. The functional $\xi_\gamma$ is a direct analogue of the helper function used in the Coppersmith--Winograd laser method \cite{CW90Laser}.

For $\theta = (\rho,1-\rho,0)$, the value of $\zeta^\theta$ on matrix multiplication tensors is $\zeta^\theta(\ang{E,H,L}) = E^\rho H L^{1-\rho}$, so $\zeta^\theta$ should coincide with $\xi_\gamma$ for $\gamma = (\rho, 1, 1-\rho)$. We verify this by proving $\zeta^\theta = \xi_\gamma$ in three steps. For simplicity, we describe the argument for a \emph{semistable} tensor $T$ of dimensions $n \times m$ (in the first two modes). In this case the HN-filtration is trivial --- $T$ itself is a single semistable piece --- and the entropy formula gives $\zeta^\theta(T) = n^\rho m^{1-\rho}$ directly. The passage from this semistable case to the general case is a standard (though technical) concentration argument for tensor powers, going back to the asymptotic sum inequality of Schönhage \cite{Sch81Tau}.

The three steps are as follows.

\begin{itemize}
    \item \emph{Step 1: $\xi_\gamma \leq \phi$ for any spectral point $\phi$ agreeing with $\zeta^\theta$ on matrix multiplication.} This is almost immediate: if $T^{\otimes N} \geq \ang{E,H,L}$, then monotonicity of $\phi$ gives $\phi(T)^N \geq \phi(\ang{E,H,L}) = E^\rho H L^{1-\rho}$, and taking $N$-th roots and the supremum yields $\xi_\gamma(T) \leq \phi(T)$.
    \item \emph{Step 2: $\phi \leq \zeta^\theta$.} In fact, \emph{any} $n\times m$ tensor satisfies $T \leq \ang{n,1,m}$, so monotonicity of $\phi$ gives $\phi(T) \leq \phi(\ang{n,1,m}) = n^\rho m^{1-\rho} = \zeta^\theta(T)$.
    \item \emph{Step 3: $\zeta^\theta \leq \xi_\gamma$.} This is the hardest step, and amounts to showing that $T$ contains a sufficiently valuable matrix multiplication tensor. In the square case $n = m$, semistability of $T$ is equivalent to $T$ having full non-commutative rank $\NCR(T) = n$. It has been known since Fortin--Reutenauer \cite[Corollary~2]{FR04NCRk} that $\CR(T) \geq \NCR(T)/2 = n/2$, so $T \geq \ang{1, n/2, 1}$ and hence $\xi_\gamma(T) \geq n/2$. Since semistability is preserved under tensor powers, the constant factor $1/2$ is removed by passing to $T^{\otimes N}$.
    
    The rectangular case $n \leq m$ is more delicate. Here we want $T \geq \ang{1, \Omega(n), \Omega(m/n)}$, which would give $\xi_\gamma(T) \geq \Omega(n^\rho m^{1-\rho})$; removing the constant again uses tensor powers and semistability. To establish this restriction, we consider the \emph{compression theorem} originally due to Strassen \cite[\S6]{Str88AsymSpec} (in the form proved by Wigderson and Zuiddam \cite[Theorem~11.9]{WZ26Spectra}). The compression theorem provides a general lower bound on the largest matrix multiplication extractable from any tensor; we give an improvement in the case of semistable tensors. This analysis is carried out in \Cref{sec:basis_shift}.
\end{itemize}

\section{Preliminaries}

For a positive integer $n$, we write $[n] = \{1,\dots,n\}$.

For any nonnegative function $F\colon S\to \bbR_{\geq 0}$ on a semigroup $S$, we write its \emph{regularization} as
\[ \uF(x) \coloneq \lim_{n\to\infty} F(x^n)^{1/n} \]
when the limit exists.

\subsection{Tensors}

Fix a base field $\bbF$. Unless otherwise specified, all tensors in this paper are $3$-mode $\bbF$-tensors. Thus, if $U$, $V$, and $W$ are finite-dimensional vector spaces, a tensor is an element $T\in U\otimes V\otimes W$. We refer to $U$, $V$, and $W$ as the first, second, and third \emph{modes} of $T$.

Following the notation common in the matrix multiplication literature \cite{ADWXXZ25MoreAsym}, we choose bases $\{x_1,\dots,x_n\}$, $\{y_1,\dots,y_m\}$, and $\{z_1,\dots,z_p\}$ of $U$, $V$, and $W$, respectively, and write a tensor $T$ with coefficients $a_{ijk}\in\bbF$ by omitting the tensor symbols:
\[ T = \sum_{i=1}^n \sum_{j=1}^m \sum_{k=1}^p a_{ijk} x_i y_j z_k. \]

For tensors $S\in U\otimes V\otimes W$ and $T\in U'\otimes V'\otimes W'$, we write $S\oplus T$ for their direct sum inside
\[ (U\oplus U') \otimes (V\oplus V') \otimes (W \oplus W'), \]
and $S\otimes T$ for their tensor product inside
\[ (U\otimes U') \otimes (V\otimes V') \otimes (W\otimes W'). \]

A \emph{restriction} from $S\in U\otimes V\otimes W$ to $T \in U'\otimes V'\otimes W'$ is given by linear maps $A\colon U \to U'$, $B\colon V \to V'$, and $C\colon W \to W'$ such that
\[ T = (A\otimes B\otimes C)S. \]
Equivalently,
\[ T = \sum_{i,j,k} a_{ijk} A(x_i) B(y_j) C(z_k). \]
In this case we say that $S$ \emph{restricts to} $T$, and write $T\leq S$.

A \emph{degeneration} from $S\in U\otimes V\otimes W$ to $T\in U'\otimes V'\otimes W'$ is given by polynomial families of linear maps
\[ A(\varepsilon)\colon U\to U', \qquad B(\varepsilon)\colon V\to V', \qquad C(\varepsilon)\colon W\to W' \]
with entries in $\bbF[\varepsilon]$, together with an integer $q\geq 0$, such that
\[ (A(\varepsilon)\otimes B(\varepsilon)\otimes C(\varepsilon))S
= \varepsilon^q T + \varepsilon^{q+1} \widetilde T(\varepsilon) \]
for some $\widetilde T(\varepsilon)\in U'\otimes V'\otimes W'\otimes \bbF[\varepsilon]$.
In this case we say that $S$ \emph{degenerates to} $T$, and write $T\degen S$.
Clearly, every restriction is a degeneration, by taking the three families to be constant.

The most basic tensors are the diagonal tensors. For any $n\geq 1$, let
\[ \ang{n} = \sum_{i=1}^n x_i y_i z_i, \]
and define the rank $\Rk(T)$ of a tensor $T$ to be the minimum $r$ such that $T \leq \ang{r}$. Similarly, the subrank $\Q(T)$ is the maximum $q$ such that $\ang{q} \leq T$. The tensor product $\ang{n} \otimes T$ may also be written as $n\odot T$.

\subsection{Strassen duality}

In this subsection we briefly recall Strassen's theory of asymptotic spectra \cite{Str88AsymSpec}. More precisely, we will use the general semiring formulation developed by Wigderson and Zuiddam \cite{WZ26Spectra}. Let $\calR$ be a commutative semiring with $0$ and $1$. A preorder $\leq$ on $\calR$ is called a \emph{Strassen preorder} if every element is nonnegative ($0\leq a$), the preorder is compatible with addition and multiplication, and
\begin{itemize}
    \item (Strong Archimedean) for every nonzero $a$, there exists $n\in\bbN$ such that $1\leq a\leq n$;\footnote{The formulation in \cite{WZ26Spectra} omits $0$ from the underlying semiring, so its strong Archimedean condition is stated for every element. Applying that framework to $\calR\setminus\{0\}$ gives exactly the formulation used here; adjoining $0$ does not affect the arguments.}
    \item (Embedding of natural numbers) for any $n,m\in\bbN$, one has $n\leq m$ in $\bbN$ if and only if $n\leq m$ in $\calR$.
\end{itemize}

In this general setting, the \emph{asymptotic closure} $\lesssim$ of $\leq$ is defined by
\begin{center}
    $a\lesssim b$ \quad if \quad $a^n \leq 2^{o(n)} b^n$.
\end{center}
The abstract notions of rank and subrank are
\begin{center}
    $\Rk_\leq(a) = \min \{ n \in \bbN : n\geq a \}$ \quad and \quad $\Q_\leq(a) = \max\{q\in \bbN : q \leq a\}$,
\end{center}
and we write $\Rk$ and $\Q$ when the preorder is clear from context. Their regularizations are
\begin{center}
    $\displaystyle \AR(a) = \inf_{n\geq 1} \Rk(a^n)^{1/n}$ \quad and \quad
    $\displaystyle \AQ(a) = \sup_{n\geq 1} \Q(a^n)^{1/n}$,
\end{center}
called the asymptotic rank and subrank.

By a \emph{functional} we simply mean a map $F\colon \calR \to \bbR$. The \emph{asymptotic spectrum} of a Strassen preordered semiring $(\calR,\leq)$, denoted by $\calX = \calX_{(\calR, \leq)}$, is the set of nonnegative functionals $\phi\colon \calR \to \bbR_{\geq 0}$ that are normalized ($\phi(1)=1$), additive, multiplicative, and monotone with respect to $\leq$.

\begin{theorem}[Abstract Strassen duality {\cite[\S3]{WZ26Spectra}}] \label{thm:duality_abstract}
    Let $(\calR, \leq)$ be a Strassen preordered semiring, and $\calX$ be the asymptotic spectrum of $\calR$. Then $\calX$ can be endowed with a compact Hausdorff topology such that evaluation at every $a$: $\phi \mapsto \phi(a)$ is a continuous map from $\calX$ to $\bbR$. Moreover,
    \begin{enumerate}
        \item The asymptotic closure $\lesssim$ is characterized by 
        \[ a\lesssim b \iff \forall \phi\in \calX, \phi(a) \leq \phi(b). \]
        \item The asymptotic rank is characterized by
        \[ \AR(a) = \max_{\phi\in \calX} \phi(a). \]
    \end{enumerate}
\end{theorem}

The next theorem shows that, although degeneration is weaker than restriction at the finite level, the two notions induce the same asymptotic preorder on tensors.
\begin{theorem}[Strassen duality for tensors {\cite[Theorem~3.2]{Str88AsymSpec}, \cite[Theorem~3.41]{WZ26Spectra}}] \label{thm:duality_tensor}
    For two tensors $T, S$, the following statements are equivalent:
    \begin{enumerate}
        \item $T^{\otimes n} \leq S^{\otimes n+o(n)}$;
        \item $T^{\otimes n} \degen S^{\otimes n+o(n)}$;
        \item $\phi(T) \leq \phi(S)$ for every $\phi\in \calX$.
    \end{enumerate}
\end{theorem}

\subsection{Matrix multiplication}

Another important family of tensors is given by the matrix multiplication tensors. For positive integers $n,m,p$, let
\[ \ang{n,m,p} = \sum_{i=1}^n \sum_{j=1}^m \sum_{k=1}^p x_{ij}y_{jk}z_{ki}. \]
One checks that
\[ \ang{n,m,p}\otimes \ang{n',m',p'} \cong \ang{nn', mm', pp'}. \]
Thus matrix multiplication tensors are generated, under tensor product, by the special cases $\ang{n,1,1}$, $\ang{1,m,1}$, and $\ang{1,1,p}$. As a consequence, the restriction of a spectral point to matrix multiplication tensors is determined by only three parameters.

\begin{proposition}[{\cite[Proposition~4.3]{Str88AsymSpec}}] \label{prop:param-specmm}
    For any $\phi\in\calX$, there exist unique real numbers $\gamma_1,\gamma_2,\gamma_3$
    such that
    \[ \phi(\ang{n,m,p}) = n^{\gamma_1} m^{\gamma_2} p^{\gamma_3} \]
    for any $n,m,p$.
\end{proposition}

We write $\gamma(\phi) = (\gamma_1,\gamma_2,\gamma_3)$ for the triple given by \Cref{prop:param-specmm}. This leads to the notion of the \emph{spectrum of matrix multiplication}.
\begin{definition}[{\cite[\S 4]{Str88AsymSpec}}] \label{def:specmm}
    The (logarithmic) spectrum $\specMM \subset \bbR^3$ of matrix multiplication is
    \[ \specMM = \gamma(\calX), \]
    namely the set of triples $(\gamma_1,\gamma_2,\gamma_3)$ arising from \Cref{prop:param-specmm}.
\end{definition}

Via Strassen duality, several basic properties of matrix multiplication tensors translate into geometric statements about $\specMM$. The only nontrivial input we use here is Strassen's asymptotic subrank estimate $\AQ(\ang{n,n,n})=n^2$ \cite[Theorem~6.6]{Str87RelBilin}.
\begin{fact} \label{fact:reg_MM}
    We have $\specMM \subseteq [0, 1]^3$. Moreover, the quantity $\gamma_1+\gamma_2+\gamma_3$ attains minimum value $2$ and maximum value $\omega$ on $\specMM$.
\end{fact}

Finally, we recall Strassen's \emph{star-convexity} theorem. For a set $C\subset \bbR^n$ and a point $p\in C$, we say that $C$ is \emph{star-convex with respect to $p$} if for every $q\in C$, the entire line segment from $p$ to $q$ lies in $C$, that is,
\[ \lambda p + (1-\lambda) q\in C \qquad \text{for all } \lambda\in [0,1]. \]
\begin{theorem}[{\cite[Theorem~6.5]{Str88AsymSpec}, see also \cite[Part II, III]{WZ26Spectra}}] \label{thm:star-convex}
    The spectrum of matrix multiplication $\specMM$ is star-convex with respect to $(1,0,1), (1,1,0)$ and $(0,1,1)$.
\end{theorem}
In particular, this theorem \emph{a priori} predicts the existence of spectral points in the triangle
\[ \triangle = \{ \gamma \in [0, 1]^3 : \gamma_1+\gamma_2+\gamma_3 = 2 \}: \]
for every $\gamma\in\triangle$, there exists $\phi\in \calX$ such that $\gamma(\phi) = \gamma$.

\subsection{Commutative rank}

Let $T\in U\otimes V\otimes W$, and choose bases $x_1,\dots,x_n$ of $U$, $y_1,\dots,y_m$ of $V$, and $z_1,\dots,z_p$ of $W$. Writing
\[ T = \sum_{i=1}^n \sum_{j=1}^m \sum_{k=1}^p a_{ijk} x_i y_j z_k, \]
we may regard $T$ as the symbolic matrix
\[ M_T(z) = \sum_{k=1}^p z_k A_k \in \Mat_{n\times m}(\bbF[z_1,\dots,z_p]), \]
where $A_k = (a_{ijk})_{i,j}\in \Mat_{n\times m}(\bbF)$.

\begin{definition}
    The \emph{commutative rank} of $T$ with respect to the third mode is
    \[ \CR_{(3)}(T) \coloneq \operatorname{rank}_{\bbF(z_1,\dots,z_p)} M_T(z). \]
    Equivalently, $\CR_{(3)}(T)$ is the generic rank of the matrix pencil $M_T(z)$.

    More generally, one can define $\CR_{(1)}$ and $\CR_{(2)}$ by choosing a different distinguished mode. Unless otherwise specified, we write
    \[ \CR(T) \coloneq \CR_{(3)}(T). \]
\end{definition}

\begin{proposition}
    The quantity $\CR_{(3)}(T)$ is independent of the choice of bases. Moreover, commutative rank is invariant under field extension, monotone under restriction, and additive under direct sum:
    \[
        \CR_{\bbF}(T) = \CR_{\bbK}(T_\bbK), \qquad
        S\leq T \implies \CR(S)\leq \CR(T), \qquad
        \CR(T\oplus S) = \CR(T) + \CR(S),
    \]
    where $\bbK/\bbF$ is any field extension.
\end{proposition}
\begin{proof}
    A change of basis in $U$ or $V$ left- or right-multiplies $M_T(z)$ by an invertible matrix over $\bbF$, so it does not change the rank. A change of basis in $W$ replaces the variables $z_1,\dots,z_p$ by an invertible linear combination of themselves, which induces an automorphism of the rational function field $\bbF(z_1,\dots,z_p)$. Hence $\CR_{(3)}(T)$ is basis-independent.

    For a field extension $\bbK/\bbF$, the commutative rank is the size of the largest minor of $M_T(z)$ that is nonzero as a polynomial in $\bbF[z_1,\dots,z_p]$. Such a minor is nonzero over $\bbF$ if and only if it remains nonzero over $\bbK$, so $\CR_{\bbF}(T)=\CR_{\bbK}(T_\bbK)$.

    If $S=(A\otimes B\otimes C)T$, then $M_S(z')$ is obtained from $M_T(z)$ by left and right multiplication by scalar matrices and by substituting each variable $z_k$ with a linear form in the new variables $z'_1,\dots,z'_{p'}$. These operations cannot increase rank, so $\CR(S)\leq \CR(T)$.

    For direct sum, after choosing compatible bases, the symbolic matrix of $T\oplus S$ is block diagonal with blocks $M_T(z)$ and $M_S(w)$, so the rank is additive. \qedhere
\end{proof}

When $\bbF$ is infinite, commutative rank admits an equivalent description via restriction.
\begin{proposition} \label{prop:cr-restr}
    If $\bbF$ is infinite, then
    \[ \CR(T) = \max\{ q\in \bbN : \ang{1,q,1} \leq T \}. \]
\end{proposition}
\begin{proof}
    If $\ang{1,q,1}\leq T$, then by monotonicity,
    \[ q = \CR(\ang{1,q,1}) \leq \CR(T). \]

    Conversely, let $q=\CR(T)$. Then some $q\times q$ minor of $M_T(z)$ is a nonzero polynomial in $z_1,\dots,z_p$. Since $\bbF$ is infinite, there exists a specialization $\lambda=(\lambda_1,\dots,\lambda_p)\in \bbF^p$ at which this minor remains nonzero. Equivalently, if $C\colon W\to \bbF$ is the linear map defined by $C(z_k)=\lambda_k$, then the contraction
    \[ (\id_U\otimes \id_V \otimes C)T \in U\otimes V \]
    has matrix rank $q$. Any rank-$q$ matrix is isomorphic to $\ang{1,q,1}$, so $\ang{1,q,1}\leq T$.
\end{proof}

\begin{corollary}
    Commutative rank is supermultiplicative:
    \[ \CR(T\otimes S) \geq \CR(T)\CR(S). \]
\end{corollary}
\begin{proof}
    By the proposition above, we may replace $\bbF$ by an infinite extension field. Let
    \[ q = \CR(T), \qquad s = \CR(S). \]
    Then
    \[ \ang{1,q,1} \leq T, \qquad \ang{1,s,1} \leq S, \]
    so
    \[ \ang{1,qs,1} \cong \ang{1,q,1}\otimes \ang{1,s,1} \leq T\otimes S. \]
    Therefore $\CR(T\otimes S) \geq qs = \CR(T)\CR(S)$. \qedhere
\end{proof}

We define the \emph{asymptotic commutative rank} by
\[ \ACR_{(3)}(T) \coloneq \lim_{n\to \infty} \CR_{(3)}(T^{\otimes n})^{1/n}. \]
Since $\CR(T^{\otimes (m+n)}) \geq \CR(T^{\otimes m})\CR(T^{\otimes n})$ and $\CR(T)\leq \Rk(T)$, Fekete's lemma implies that the limit exists, and
\[ \ACR_{(3)}(T) = \sup_{n\geq 1} \CR_{(3)}(T^{\otimes n})^{1/n}. \]
Unless there is danger of confusion, we abbreviate $\ACR = \ACR_{(3)}$.

Finally, there is a closely related notion called \emph{non-commutative rank}. Fortin and Reutenauer proved \cite[Corollary~2]{FR04NCRk} that
\[ \CR(T) \leq \NCR(T) \leq 2\CR(T). \]
Taking tensor powers and $n$-th roots, we see that the asymptotic non-commutative rank is also equal to $\ACR(T)$.

\section{Abstract functionals} \label{sec:abstract_functionals}

In this section, we work in a general Strassen preordered commutative semiring $\calR$. We first list some basic properties of a functional $F\colon \calR \to \bbR_{\geq 0}$:
\begin{itemize}
    \item A function $F$ is \emph{monotone} if $a\leq b\implies F(a)\leq F(b)$.
    \item A function $F$ is \emph{monomially homogeneous} if $F(ma^n) = mF(a)^n$ for all $m,n\in \bbN$.
    \item A function $F$ is \emph{homogeneous} if
    $F(p(a)) = p(F(a))$ for all polynomials $p(T) \in \bbN[T]$.
\end{itemize}

Note that every monomially homogeneous functional is normalized by taking $n=0$.

\begin{example}
    Let $\calZ$ be a nonempty closed subset of $\calX$, and consider the maximization (minimization, respectively)
    over $\calZ$,
    \[ F(a) = \max_{\phi\in \calZ} \phi(a), \quad f(a) = \min_{\phi\in\calZ} \phi(a). \]
    It is not hard to verify that $F$ ($f$, respectively) is monotone, subadditive (superadditive, respectively),
    submultiplicative (supermultiplicative, respectively), and homogeneous.
    In this section we prove a partial converse to this observation.
\end{example}

\subsection{Characterizing maximization}

\begin{theorem} \label{thm:maximization}
    For a function $F \colon \calR\to \bbR_{\geq 0}$, the following are equivalent:
    \begin{itemize}
        \item $F$ is monotone, subadditive, submultiplicative, and monomially homogeneous.
        \item $F$ is the maximization over a subset of the spectrum: there exists a nonempty closed subset $\calZ\subset\calX_\calR$
        such that
        \[ F(a) = \max_{\phi\in\calZ} \phi(a).\]
        Moreover, one can take $\calZ = \{\phi : \forall a\in \calR, \phi(a) \leq F(a)\}$.
    \end{itemize}
\end{theorem}

To prove this theorem, we first need to study a finer preorder than $\leq$ that captures the information of $F$.
Let $\leq$ be the original preorder, and let $\leq^F$ be the preorder obtained as follows:
It is the smallest semiring preorder such that $a\leq^F b$ whenever $a \leq b$, and $a\leq^F \ceil{F(a)}$.

The following lemma clarifies the structure of $\leq^F$:
\begin{lemma}[Zig-zag normal form]
    For two elements $s,t\in\calR$, $s \leq^F t$ iff there exists a sequence $s=a_0, \dots, a_n = t$ and $\{b_i, c_i, d_i\}_{1\leq i\leq n}$, such that
    \[ a_{i-1} \leq b_i + c_i d_i, \quad b_i + \ceil{F(c_i)} d_i \leq a_i. \]
\end{lemma}

\begin{proof}
    The ``if'' direction is straightforward: since $\leq^F$ has $c_i \leq^F \ceil{F(c_i)}$ and the original preorder $\leq$,
    we have
    \[ a_{i-1} \; \leq^F \; b_i + c_id_i \; \leq^F \; b_i + \ceil{F(c_i)}d_i \; \leq^F \; a_i. \]
    Thus $s\leq^F t$ follows by transitivity.

    For the ``only if'' direction, it suffices to verify that the existence of such a zig-zag defines a semiring preorder extending $\leq$ and incorporating all additional comparisons $a\leq^F \ceil{F(a)}$. The minimality in the definition of $\leq^F$ then implies that every relation in $\leq^F$ admits a zig-zag normal form.
    \begin{itemize}
        \item Extends $\leq$: If $s\leq t$, then we can just take
        \[ \underbracket{s}_{a_0} \leq \underbracket{s}_{b_1} + \underbracket{0}_{c_1} 
        \underbracket{0}_{d_1}, \quad s + \ceil{F(0)} \cdot 0 \leq \underbracket{t}_{a_1}. \]
        \item Incorporates $F$: For any $a$, we have
        \[ \underbracket{a}_{a_0} \leq \underbracket{0}_{b_1} + \underbracket{a}_{c_1} 
        \underbracket{1}_{d_1}, \quad 0 + \ceil{F(a)} \cdot 1 \leq \underbracket{\ceil{F(a)}}_{a_1}. \]
        \item Additive: If $s, t$ have a corresponding zig-zag normal form, then for any $u$, the normal form of $s+u, t+u$ is given by
        \[ \underbracket{a_{i-1} + u}_{a'_{i-1}} \leq \underbracket{b_i + u}_{b_i'} + \underbracket{c_i}_{c_i'}
            \underbracket{d_i}_{d_i'}, \quad
           b_i+u + \ceil{F(c_i)} d_i \leq \underbracket{a_{i}+u}_{a_i'}. \]
        \item Multiplicative: If $s, t$ have a corresponding zig-zag normal form, then for any $u$, the normal form of $us, ut$ is given by
        \[ \underbracket{u a_{i-1}}_{a'_{i-1}} \leq \underbracket{u b_i}_{b_i'} + \underbracket{c_i}_{c_i'}
            \underbracket{u d_i}_{d_i'}, \quad
           u b_i + \ceil{F(c_i)} u d_i \leq \underbracket{ua_{i}}_{a_i'}. \]
        \item Transitivity: If we have a zig-zag normal form for $s, t$ and for $t, u$, then we can concatenate them to obtain the one for $s, u$. \qedhere
    \end{itemize}
\end{proof}

\begin{lemma}\label{lem:F-lifting}
    If $F$ satisfies the first set of conditions in \Cref{thm:maximization}, then for any $b,c,d\in\calR$,
    \[ F(b + cd) \leq F(b + \ceil{F(c)}d). \]
\end{lemma}

\begin{proof}
    We proceed by calculation:
    \begin{align*}
        F(b+cd)^n &= F((b + cd)^n) & \text{(monomially homogeneous)}\\
        &\leq \sum_{k=0}^n \binom{n}{k} F(b^{n-k} (cd)^{k} ) & \text{(subadditive)}\\
        &\leq \sum_{k=0}^n \binom{n}{k} F(c)^{k} F(b^{n-k} d^{k} ) & \text{(submultiplicative)}\\
        &\leq \sum_{k=0}^n \binom{n}{k} \ceil{F(c)}^{k} F(b^{n-k} d^{k} )\\
        &= \sum_{k=0}^n F(\tbinom{n}{k} \ceil{F(c)}^{k} b^{n-k} d^{k}). & \text{(monomially homogeneous)}
    \end{align*}
    Now, as $\binom{n}{k} \ceil{F(c)}^{k} b^{n-k} d^{k}$ appears in the binomial expansion of
    $(b + \ceil{F(c)}d)^n$, we have
    \begin{align*}
        \phantom{F(b+cd)^n} &\leq \sum_{k=0}^n F((b + \ceil{F(c)}d)^n) & \text{(monotone)}\\
        &= (n+1) F(b + \ceil{F(c)}d)^n. & \text{(monomially homogeneous)}
    \end{align*}
    Thus we have $F(b + cd) \leq (n+1)^{1/n} F(b + \ceil{F(c)}d)$. Taking $n\to\infty$ we conclude.
\end{proof}

We remind the reader that $\calX_{(\calR,\leq^F)}$ denotes the asymptotic spectrum of $\calR$ equipped with the preorder $\leq^F$. Since $\leq^F$ is finer than $\leq$, this spectrum can be viewed as a closed subset of $\calX_{(\calR,\leq)}$. The next proposition describes that subset explicitly in terms of the additional inequalities introduced by $\leq^F$.

\begin{proposition} \label{prop:F-order}
    If $F$ satisfies the first set of conditions in \Cref{thm:maximization}, then $\leq^F$
    is a Strassen preorder. Moreover, $F$ is $\leq^F$-monotone: $s\leq^F t \implies F(s) \leq F(t)$. Under the identification above, the asymptotic spectrum
    $\calX_{(\calR, \leq^F)} \subset \calX_{(\calR, \leq)}$ is the subset $\calZ$ cut out by
    $\phi(a) \leq F(a)$ for all $a\in\calR$.
\end{proposition}
\begin{proof}
    Almost all properties required for $\leq^F$ to be a Strassen preorder are inherited from $\leq$,
    except the \emph{embedding of natural numbers} which \emph{a priori} could potentially be violated when extending a preorder.
    This will be verified once we prove that $F$ is $\leq^F$-monotone: If natural numbers $n, m$ satisfy $n\leq^F m$,
    then since $F(n)\leq F(m)$ and $F$ is monomially homogeneous, we must have $n\leq m$ in $\bbN$.

    To prove $F$ is $\leq^F$-monotone, we start from a zig-zag normal form of $s, t$.
    By definition, we have
    \[ a_{i-1} \leq b_i + c_i d_i, \quad b_i + \ceil{F(c_i)} d_i \leq a_i, \]
    this gives
    \[ F(a_{i-1}) \stackrel{\leq\text{-monotone}}{\leq} F(b_i+c_id_i) \stackrel{\text{\Cref{lem:F-lifting}}}{\leq}
    F(b_i + \ceil{F(c_i)}d_i) \stackrel{\leq\text{-monotone}}{\leq} F(a_i). \]
    Combining this sequence of inequalities, we obtain $F(s) = F(a_0) \leq F(a_n) = F(t)$.

    It remains to verify the claimed description of the spectrum. If $\phi\in\calX_{(\calR,\leq^F)}$, then the added relations give $\phi(a)\leq\ceil{F(a)}$. Applying the corresponding relation to every power $a^n$ gives
    \[ \phi(a)^n = \phi(a^n) \leq \ceil{F(a^n)} = \ceil{F(a)^n} \]
    by monomial homogeneity. Taking $n\to\infty$ gives $\phi(a) \leq F(a)$. Conversely, if $\phi\in\calX_{(\calR,\leq)}$ satisfies $\phi(a)\leq F(a)$ for every $a$, then $\phi(a)\leq\ceil{F(a)}$, so $\phi$ is monotone under every added generator and hence under $\leq^F$. Thus $\calX_{(\calR,\leq^F)}=\calZ$. This subset is closed in the compact Hausdorff topology from \Cref{thm:duality_abstract}, because it is an intersection of closed evaluation inequalities.
\end{proof}

\begin{proof}[Proof of \Cref{thm:maximization}]
    In light of \Cref{prop:F-order}, we only need to characterize $F$ by the asymptotic rank in the preorder $\leq^F$.
    Since $a\leq^F \ceil{F(a)}$, we immediately have
    $\Rk_{\leq^F} (a) \leq \ceil{F(a)}$. On the other hand, if $a\leq^F r$ for some $r\in\bbN$,
    then $F(a) \leq F(r)$, and by monomial homogeneity we have $r = F(r) \geq F(a)$, thus $r \geq \ceil{F(a)}$,
    we conclude $\Rk_{\leq^F} (a) = \ceil{F(a)}$. Hence the asymptotic rank is
    \[ \AR_{\leq^F}(a) = \lim_{n\to\infty} \Rk_{\leq^F}(a^n)^{1/n} = \lim_{n\to\infty} \ceil{F(a)^n}^{1/n} = F(a). \]
    From this characterization, we can write $F$ as maximization by Strassen's duality (\Cref{thm:duality_abstract}):
    \[ F(a) = \AR_{\leq^F}(a) = \max_{\phi\in\calX_{(\calR, \leq^F)}} \phi(a) = \max_{\phi\in\calZ} \phi(a). \qedhere \]
\end{proof}

\subsubsection{Applications}

Using the characterization above, we derive a partial result for upper support functionals $\zeta^\theta$. We begin with a general statement.

\begin{proposition} \label{prop:asymp-F}
    If $F$ is monotone, normalized, \emph{additive} and submultiplicative, then
    \[ \uF(a) \coloneq \lim_{n\to\infty} F(a^n)^{1/n} \]
    always exists, and satisfies the conditions of \Cref{thm:maximization}.
\end{proposition}

\begin{proof}
    By Fekete's lemma, the limit exists and can be written as
    \[ \uF (a) = \inf_{n\geq 1} F(a^n)^{1/n}. \]

    We now verify that $\uF$ satisfies the first set of conditions in \Cref{thm:maximization}: monotonicity, subadditivity, submultiplicativity, and monomial homogeneity.

    Monotonicity: If $a\leq b$, then $a^n\leq b^n$ for all $n$, so
    \[ \uF(a) = \lim_{n\to\infty} F(a^n)^{1/n}
    \leq \lim_{n\to\infty} F(b^n)^{1/n} = \uF(b). \]

    Subadditivity: For any $a, b\in\calR$, by the definition of $\uF$, for every $\varepsilon>0$ we have
    \[ F(a^n) = O(2^{\varepsilon n} \uF(a)^{n}), \quad F(b^n) = O(2^{\varepsilon n} \uF(b)^{n}), \]
    where the $O(\cdot)$ here hides a constant depending on $\varepsilon$. Thus
    \begin{align*}
        F((a+b)^n) &\leq \sum_{k=0}^n \binom n k F(a^k) F(b^{n-k}) \\
        &= \sum_{k=0}^n \binom{n}{k} O(2^{\varepsilon n} \cdot \uF(a)^k \uF(b)^{n-k})\\
        &= O(2^{\varepsilon n} (\uF(a) + \uF(b))^n).
    \end{align*}
    Taking $n$-th roots and letting $n\to\infty$ gives
    \[ \uF(a+b) \leq 2^{\varepsilon}(\uF(a) + \uF(b)). \]
    Since this holds for every $\varepsilon > 0$, we conclude that $\uF(a+b) \leq \uF(a) + \uF(b)$.

    Submultiplicativity: For any $a, b\in\calR$, since $F$ is submultiplicative, we have
    \[ \uF(ab) = \lim_{n\to\infty} F((ab)^n)^{1/n} \leq \lim_{n\to\infty} F(a^n)^{1/n} \cdot F(b^n)^{1/n} 
    = \uF(a) \uF(b). \]

    Monomial homogeneity: First, $\uF$ preserves powers:
    \[ \uF(a^m) = \lim_{n\to\infty} F(a^{mn})^{1/n}
    = \lim_{mn\to\infty} [F(a^{mn})^{1/mn}]^{m} = \uF(a)^m. \]
    Next, we show that $\uF$ preserves multiplication by natural numbers:
    \[ \uF(ka) = \lim_{n\to\infty} F(k^n a^n)^{1/n}
    = \lim_{n\to\infty} k F(a^n)^{1/n} = k\uF(a). \qedhere \]
\end{proof}

\begin{remark} \label{rmk:bv25}
    We note that Bug\'ar and Vrana derived a more general statement \cite[Theorem~3.23]{BV25Func}. Our additivity assumption cannot be weakened to subadditivity. A counterexample similar to \cite[Remark~3.26]{BV25Func} is presented below. Their notion of \emph{regularization} \cite[Definition~3.13]{BV25Func} differs slightly from ours in general, but it agrees with our definition when $F$ is additive. Moreover, their regularized functional is always a maximization of asymptotic spectral points, even when $F$ is only subadditive rather than additive.

    Let $\calR = \{(0,0)\}\cup \bbZ_{\geq 1}^2$
    with entrywise addition and multiplication, and let $\leq$ be the entrywise preorder.
    Then $1_{\calR}=(1,1)$ and $0_{\calR}=(0,0)$, and $\leq$ is a Strassen preorder.

    Define $F(x,y)=\max(x,\sqrt y)$. One checks that $F$ is monotone, normalized, subadditive, and submultiplicative. In this case $\uF=F$, but it is not a maximization of spectral points because
    \[ 2 = \uF(2(1, 2)) \neq 2\uF(1, 2) = 2\sqrt 2. \]
\end{remark}

Strassen claimed a characterization of the upper support functional $\zeta^\theta$ \cite[\S7 (ii)]{Str91SuppFunc}: there exists a \emph{unique minimal} closed subset $\sfZ^\theta$ such that
\[ \utilde{\zeta}^\theta(T) = \max_{\phi\in \sfZ^\theta} \phi(T), \]
while no published proof is known. Here we show how \Cref{prop:asymp-F} yields a weaker characterization. Bug\'ar and Vrana gave a stronger version than ours by showing that this closed subset need only contain spectral points dominating the lower support functional, namely $\phi \geq \zeta_\theta$; see the comments after \cite[Corollary~3.25]{BV25Func}.
\begin{corollary}
    There exists a closed subset $\calZ^\theta$ such that
    \[ \utilde{\zeta}^\theta(T) = \max_{\phi\in \calZ^\theta} \phi(T). \]
\end{corollary}
\begin{proof}
    It is already known that $\zeta^\theta$ is monotone, normalized, additive and submultiplicative \cite[Theorem~2.8]{Str91SuppFunc}. Then by \Cref{prop:asymp-F}, $\utilde{\zeta}^\theta(T)$ satisfies
    the conditions of \Cref{thm:maximization}, thus
    there exists a closed subset $\calZ$ such that
    \[ \utilde{\zeta}^\theta(T) = \max_{\phi\in \calZ} \phi(T). \]
    From the choice of $\calZ$ in \Cref{thm:maximization}, one sees that this construction naturally gives a \emph{maximal} subset rather than a minimal one.
\end{proof}

\subsection{Sufficient condition for minimization}

\begin{theorem} \label{thm:minimization}
    If a function $f\colon \calR\to\bbR_{\geq 0}$ is monotone, superadditive, supermultiplicative, monomially homogeneous,
    and
    \begin{itemize}
        \item (Monomially dominant) For any $s, t$, there exists $k = k(n)$
        such that
        \[ \binom{n}{k} f(s^k t^{n-k}) \geq 2^{-o(n)} f(s+t)^n. \]
    \end{itemize}
    Then $f$ is a minimization of a subset of spectral points: There exists a nonempty closed subset
    $\calZ \subset \calX_{\calR}$ such that $f(a) = \min_{\phi\in\calZ} \phi(a)$. Moreover, one can take $\calZ = \{\phi : \forall a\in \calR, \phi(a) \geq f(a)\}$.
\end{theorem}

Similar to the proof of \Cref{thm:maximization}, we define a similar preorder and the corresponding normal form. Let the preorder $\leq_f$ be defined as the smallest semiring preorder extending $\leq$ and also satisfying $\floor{f(a)} \leq a$ for all $a$, then we have

\begin{lemma}[Zig-zag normal form]
    For two elements $s,t\in\calR$, $s \leq_f t$ iff there exists a sequence $s=a_0, \dots, a_n = t$ and $\{b_i, c_i, d_i\}_{1\leq i\leq n}$, such that
    \[ a_{i-1} \leq b_i + \floor{f(c_i)} d_i, \quad b_i + c_i d_i \leq a_i. \]
\end{lemma}

We omit its proof since it's almost the same.

\begin{lemma}
    If $f$ satisfies the first set of conditions in \Cref{thm:minimization}, then for any $b,c,d\in\calR$,
    \[ f(b + \floor{f(c)}d) \leq f(b + cd). \]
\end{lemma}
\begin{proof}
    We again prove by calculation:
    \begin{align*}
        f(b + \floor{f(c)} d)^n &= f ((b + \floor{f(c)} d)^n) & \text{(monomially homogeneous)}\\
        &\leq 2^{o(n)} \binom{n}{k} f(b^{n-k} (\floor{f(c)}d)^k) & \text{(monomially dominant)}\\
        &= 2^{o(n)} \floor{f(c)}^k f(\tbinom{n}{k} b^{n-k} d^k) & \text{(monomially homogeneous)}\\
        &\leq 2^{o(n)} f(c)^k f(\tbinom{n}{k}b^{n-k} d^k)\\
        &\leq 2^{o(n)} f(\tbinom{n}{k}b^{n-k} (cd)^k)  & \text{(supermultiplicative)}\\
        &\leq 2^{o(n)} f((b + cd)^n) & \text{(monotone)}\\
        &\leq 2^{o(n)} f(b + cd)^n. & \text{(monomially homogeneous)}
    \end{align*}
    Taking $n\to\infty$, we conclude.
\end{proof}

\begin{proposition} \label{prop:f-order}
    If $f$ satisfies the first set of conditions in \Cref{thm:minimization}, then $\leq_f$
    is a Strassen preorder. Moreover, $f$ is $\leq_f$-monotone: $s\leq_f t \implies f(s) \leq f(t)$. The asymptotic spectrum
    $\calX_{(\calR, \leq_f)} \subset \calX_{(\calR, \leq)}$ is the closed subset $\calZ$ cut out by
    $\phi(a) \geq f(a)$ for all $a\in\calR$.
\end{proposition}
We again omit the proof of the statement above since there is nothing new.

\begin{proof}[Proof of \Cref{thm:minimization}]
    Similar to the proof of \Cref{thm:maximization}, it is not hard to see that $\Q_{\leq_f}(a) = \floor{f(a)}$.
    Then by \cite[Theorem~3.26]{WZ26Spectra}, we have
    \begin{align*}
        \min_{\phi\in\calX_{(\calR, \leq_f)}} \phi(a) &= \sup\{ (s/t)^{1/n} : s,t,n \in \bbZ_{\geq 1}, s\leq_f t a^n \}\\
        &= \sup\{ (\Q_{\leq_f}(ta^n)/t)^{1/n} : t, n\in\bbZ_{\geq 1} \}\\
        &= \sup\{ (\floor{t f(a)^n}/t)^{1/n} : t, n\in\bbZ_{\geq 1} \}\\
        &= f(a). \qedhere
    \end{align*}
\end{proof}

\subsubsection{Applications} \label{sec:min-appl}

\begin{proposition} \label{prop:f-asymp}
    If $f$ is monotone, normalized, \emph{additive} and supermultiplicative, then
    \[ \uf(a) \coloneq \lim_{n\to\infty} f(a^n)^{1/n} \]
    always exists, and satisfies the conditions of \Cref{thm:minimization}.
\end{proposition}
\begin{proof}
    Since $f$ is monotone, we have $f(a) \leq \Rk(a)$ and thus $f(a^n)^{1/n} \leq \Rk(a^n)^{1/n} \leq \Rk(a)$. By Fekete's lemma (see the proof of \cite[Corollary~2.12]{WZ26Spectra}), the limit defining $\uf$ exists and can be written as
    \[ \uf(a) = \sup_{n \geq 1} f(a^n)^{1/n}. \]
    One can mimic the proof of \Cref{prop:asymp-F} to verify monotonicity, superadditivity, supermultiplicativity, and monomial homogeneity. So we only need to prove that $\uf$ is monomially dominant.
    By the definition of $\uf$, we have
    \[ f((a+b)^n) \geq 2^{-o(n)} \uf(a + b)^n. \]
    Thanks to additivity, the left-hand side expands to $n+1$ monomials
    \[ f((a+b)^n) = \sum_{k=0}^n \binom n k f(a^k b^{n-k}), \]
    so some $k$ satisfies
    \[ \binom n k \uf ( a^k b^{n-k}) \geq \binom n k f( a^k b^{n-k}) \geq (n+1)^{-1} 2^{-o(n)} \uf(a+b)^n. \qedhere \]
\end{proof}

Since commutative rank satisfies all the conditions, we immediately obtain the following characterization.
\begin{corollary} \label{cor:acr-minimize-edge}
    The asymptotic commutative rank is a minimization of spectral points. In particular, if
    \[ \calZ_{(3)} \coloneq \{ \phi : \gamma_2(\phi) = 1 \}, \]
    then
    \[ \ACR_{(3)}(T) = \min_{\phi\in \calZ_{(3)}} \phi(T). \]
\end{corollary}
\begin{proof}
    For any $\phi\in \calZ$ guaranteed by \Cref{thm:minimization}, it should satisfy
    $\phi(T) \geq \ACR(T)$. Choosing $T=\ang{1,2,1}$, we have $2^{\gamma_2} \geq 2$, i.e., $\gamma_2 \geq 1$.
    We always have $\gamma_2 \leq 1$, so $\gamma_2 = 1$.
    This gives $\calZ \subseteq \calZ_{(3)}$.

    On the other hand, for any $T$, we have $T^{\otimes n} \gtrsim \ang{1, \CR(T^{\otimes n}), 1}$: the corresponding restriction holds after passing to an infinite extension of $\bbF$, and asymptotic restriction descends to the ground field. Evaluating $\phi$ gives $\phi(T)^{n} \geq \CR(T^{\otimes n})$. Taking $n\to\infty$, we obtain
    $\phi(T) \geq \ACR(T)$, thus $\calZ_{(3)} \subseteq \calZ$.
\end{proof}

It is natural to ask whether the same argument can be applied to asymptotic slice rank. The obstacle is that, although slice rank is additive \cite{Gowers21SR}, it is not supermultiplicative \cite[Example~5.2]{CVZ23QFun}.

A plausible workaround is to replace slice rank by an equivalent notion. The most promising candidate is $G$-stable rank (which is also additive): over a perfect field, $G$-stable rank and slice rank are equivalent up to a constant factor \cite[\S1C]{Derksen22GStableRank}. Moreover, in characteristic zero, Derksen proved that $G$-stable rank is supermultiplicative \cite[\S5C]{Derksen22GStableRank}\footnote{Although the theorem is stated over $\bbC$, the same conclusion holds over every characteristic-zero field. Indeed, by model-theoretic arguments it extends to every algebraically closed field of characteristic zero, and then descends to an arbitrary characteristic-zero field $\bbF$ because $G$-stable rank is invariant under field extension \cite[Theorem~2.5]{Derksen22GStableRank}.} and conjectured that the characteristic-zero assumption is unnecessary \cite[p.~1074]{Derksen22GStableRank}. Thus, once supermultiplicativity of $G$-stable rank is established over a perfect field $\bbF$, the same minimization principle would also apply to $\ASR$.

\begin{corollary} \label{cor:ASR_minimize}
    Let $\bbF$ be a perfect field. If $G$-stable rank is supermultiplicative over $\bbF$, then asymptotic slice rank $\ASR$ over $\bbF$ is a minimization of spectral points.
\end{corollary}

\subsubsection{Higher-mode tensors} \label{sec:higher-mode}

We briefly discuss how \Cref{thm:minimization} applies to higher-mode tensors. Let
\[ T \in V_1 \otimes \cdots \otimes V_d, \qquad d \geq 4. \]
Many familiar rank notions extend naturally to this setting, but new phenomena appear because there are now many inequivalent ways to bipartition the set of modes.

Let $\Par$ be a nonempty collection of bipartitions of $[d]$, that is, each element of $\Par$ is of the form $\{S, [d]\setminus S\}$ with $\emptyset \subsetneq S \subsetneq [d]$. The corresponding \emph{partition rank} was introduced by Naslund \cite{Nas20PR}. A tensor has \emph{$\Par$-rank one} if it can be written as $u \otimes v$ for some bipartition $\{S, [d]\setminus S\}\in \Par$, where
\[ u \in \bigotimes_{\varkappa\in S} V_\varkappa, \qquad v \in \bigotimes_{\varkappa\notin S} V_\varkappa. \]
The \emph{partition rank} $\PR^\Par(T)$ is the least number of $\Par$-rank-one tensors whose sum is $T$. When $\Par$ consists of all nontrivial bipartitions, we simply write $\PR(T)$.

This notion interpolates several familiar ranks. If $\Par=\{\{S,[d]\setminus S\}\}$ consists of a single bipartition, then $\PR^\Par(T)$ is exactly the rank of the corresponding flattening. If
\[ \Pars \coloneq \left\{ \{\{\varkappa\}, [d]\setminus\{\varkappa\}\} : \varkappa\in[d] \right\}, \]
then $\PR^{\Pars}(T)$ is the slice rank.

Thus partition ranks already recover several standard tensor parameters. For $3$-mode tensors there are essentially no further choices beyond these up to permutation of the modes, whereas for $d\geq 4$ they form a much richer family.

The higher-order asymptotic spectrum is also less understood. Over $\bbC$, Christandl--Vrana--Zuiddam constructed the quantum functionals in every order \cite{CVZ23QFun}, and Sakabe--Do\u{g}an--Walter proved that they coincide with Strassen's support functionals \cite{SDW26Fun}. In particular, the spectral characterization of asymptotic slice rank remains valid over $\bbC$ in higher order as well \cite[Corollary~5.7]{CVZ23QFun}, \cite[Corollary~1.2]{SDW26Fun}.

Besides these genuinely $d$-mode functionals, one also gets spectral points of order $d$ by grouping several modes together and viewing a $d$-mode tensor as a tensor of smaller order; composing any spectral point on the lower-order tensor space with this grouping operation gives a spectral point on the original $d$-mode tensor space. At present, these are all the explicit spectral points that we know. We also note that \cite{CVZ23QFun} proposed a more general family of quantum functionals, parameterized by probability distributions supported on arbitrary collections $\Par$ of bipartitions, but only proved the spectral-point property in the case $\Par=\Pars$. These are what are now usually called the quantum functionals.

On the other hand, for partition rank itself we currently do not know an additive and sufficiently multiplicative surrogate to which a variant of \Cref{prop:f-asymp} could be applied. This is in contrast with the slice-rank discussion above, where $G$-stable rank provides precisely such a surrogate.

A natural candidate for partition rank would be geometric rank, which is additive \cite[Lemma~4.3]{KMZ23GR} and equivalent to partition rank up to a constant factor by \cite[Corollary~3]{CM23PR}. However, geometric rank is not supermultiplicative: an explicit counterexample is given in \cite[Remark~6.5]{KMZ23GR}. It would therefore be very interesting to prove for geometric rank some weaker multiplicativity property strong enough to support an analogue of \Cref{prop:f-asymp}, which would in turn yield an \emph{a priori} spectral characterization of asymptotic partition rank.

Fortunately, Moshkovitz and Zhu \cite{MZ26CRank} recently introduced a higher-order analogue of commutative rank. We fix the pair of modes $\{1,2\}$ for definiteness; the same construction works for any unordered pair of distinct modes. Choose bases $x_{1,1},\dots,x_{1,n_1}$ of $V_1$, $x_{2,1},\dots,x_{2,n_2}$ of $V_2$, and bases of the remaining modes, and write
\[ T = \sum_{i_1,\dots,i_d} a_{i_1,\dots,i_d} x_{1,i_1} x_{2,i_2}\cdots x_{d,i_d}. \]
Introduce indeterminates $\ov{z} = (z_{3,1},\dots,z_{3,n_3};\dots; z_{d,1},\dots,z_{d,n_d})$, and define the $n_1\times n_2$ matrix
\[ M_T(\ov{z}) \coloneq \left( \sum_{i_3,\dots,i_d} a_{i_1,\dots,i_d} z_{3,i_3}\cdots z_{d,i_d} \right)_{i_1,i_2}. \]
The \emph{multilinear commutative rank} of $T$ with respect to the pair $(1,2)$ is
\[ \CR_{1,2}(T) \coloneq \operatorname{rank}_{\bbF(\ov{z})} M_T(\ov{z}). \]
When $d=3$, this is exactly the commutative rank $\CR_{(3)}(T)$ discussed earlier.

The same arguments as in the $3$-mode case show that $\CR_{1,2}(T)$ is independent of the choice of bases, monotone under restriction, additive under direct sum, and supermultiplicative under tensor product. Thus its regularization
\[ \ACR_{1,2}(T) \coloneq \lim_{n\to\infty} \CR_{1,2}(T^{\otimes n})^{1/n} \]
is well defined by Fekete's lemma. From now on we abbreviate $\CR=\CR_{1,2}$ and $\ACR=\ACR_{1,2}$.

\begin{corollary} \label{cor:higher-mode-acr-minimize}
    For every $d\geq 4$, the asymptotic multilinear commutative rank $\ACR$ is a minimization of spectral points.
\end{corollary}

At present this is only an \emph{a priori} characterization: we do not know whether the currently available spectral points already suffice to realize this minimum. In fact, we suspect they do not, and that higher-mode tensors will require genuinely new spectral points beyond the presently known support or quantum functionals and the points obtained from lower-order tensors by grouping modes together.

Moshkovitz and Zhu also proved that multilinear commutative rank is closely related to a partition rank. Let $\Par_{1,2}$ denote the set of all bipartitions separating $1$ and $2$. Then over an infinite field,
\[ \CR_{1,2}(T) \leq \PR^{\Par_{1,2}}(T) \leq 2^{d-2} \CR_{1,2}(T). \]
Hence $\CR_{1,2}$ and $\PR^{\Par_{1,2}}$ have the same regularization. In particular, \Cref{cor:higher-mode-acr-minimize} also provides an \emph{a priori} spectral characterization for this asymptotic partition rank.

We now prove that there must exist a spectral point beyond the known ones. We first list all known spectral points in the $4$-mode setting. The first family, of course, consists of the $4$-mode quantum functionals. We can also obtain spectral points by grouping modes together and applying the known spectral points to the grouped tensor. Let $\pi$ be a partition of $[4]$ into nonempty subsets. We write $F_\theta$ for the quantum functional corresponding to $\theta\in\calP(\pi)$. For example, when $\pi$ partitions $[4]$ into singletons, the functionals $F_\theta$ for $\theta\in\calP([4])$ are exactly the $4$-mode quantum functionals, while when $\pi = \{\{1,2\}, \{3\}, \{4\}\}$, the functional $F_\theta$ is a $3$-mode quantum functional after recognizing $V_1\otimes V_2$ as a single mode. We let
\[ \calQ_\pi = \{F_\theta : \theta \in \calP(\pi)\} \]
be the set of quantum functionals obtained from the partition $\pi$.

Also, $\pi$ cannot be the trivial partition $\{[4]\}$. Moreover, when $|\pi|=2$, the corresponding spectral points are just flattening ranks, which already appear either among the $4$-mode quantum functionals (when $\pi$ is a $1$-$3$ partition) or among the $3$-mode quantum functionals obtained after grouping (when $\pi$ is a $2$-$2$ partition). Therefore, we can summarize all known spectral points by
\[ \calQ = \calQ_{[4]} \cup \bigcup_{1\leq i < j \leq 4} \calQ_{\pi_{ij}}, \]
where $\pi_{ij}$ is the partition of $[4]$ into $\{i,j\}$ and all singletons in its complement. For brevity, we write $\calQ_{ij}$ for $\calQ_{\pi_{ij}}$.

Let $\calZ$ be the set of spectral points guaranteed by \Cref{thm:minimization} for $\ACR_{1,2}$. Then $\calZ$ is the set of spectral points $\phi$ such that $\phi(T) \geq \ACR_{1,2}(T)$ for all $T$. We will show that $\calZ$ must contain some point beyond $\calQ$. To do so, we first observe that some spectral points in $\calQ$ can already be ruled out.

\begin{lemma} \label{lem:exclude-spectral-points}
    
Let $F_\theta \in \calQ$, with $\theta \in \calP(\pi)$. Then $F_\theta \in \calZ$ if and only if $1$ and $2$ are still partitioned, and $\theta$ places all of its mass on $1, 2$, i.e.,
    \[ \sum_{\substack{S\in \pi \\ S \cap \{1, 2\} \neq \emptyset}} \theta_S = 1. \]
\end{lemma}

To prove this, we first introduce a useful family of tensors, namely diagonal tensors supported only on a subset of modes. For a nonempty $S\subseteq [4]$, we define $\ang{n}_S$ to be the diagonal tensor of size $n$ on the modes in $S$. For example, if we denote the four modes by $x,y,z,w$, then $\ang{n}_{1,3}$ is the tensor
\[ \ang{n}_{1,3} = \sum_{i=1}^n x_i y_1 z_i w_1. \]

\begin{proof}[Proof of \Cref{lem:exclude-spectral-points}]
    We first prove the ``only if'' direction. Let $T = \ang{n}_{1,2}$. Then $T$ is essentially the $n\times n$ identity matrix on the modes $V_1, V_2$.
    Therefore, we have $\CR_{1,2}(T) = n$ and thus $\ACR_{1,2}(T) = n$. On the other hand, from the support functional formulation of $F_\theta$, it's not hard to compute that, if $1, 2$ are partitioned by $\pi$, then
    \[ F_\theta(T) = n^{\sum_{S\in \pi : S\cap\{1,2\}\neq \emptyset} \theta_S}. \]
    Since $F_\theta \in \calZ$, we have $F_\theta(T) \geq \ACR_{1,2}(T)$, so the exponent must be at least $1$, which gives the desired conclusion.

    However, if $1, 2$ are not partitioned by $\pi$, since $T$ becomes a ``unit tensor'' in $(V_1\otimes V_2) \otimes V_3 \otimes V_4$, we have $F_\theta(T) = 1$; in that case $F_\theta \geq \ACR_{1,2}$ cannot hold.

    For the ``if'' direction, we may assume without loss of generality that $\bbF$ is infinite. For a tensor $T$ with $\CR_{1,2}(T) = q$, this means that $T \geq \ang{q}_{1,2}$, so we have $F_\theta(T) \geq F_\theta(\ang{q}_{1,2}) = q$. Taking tensor powers, we conclude that $F_\theta(T) \geq \ACR_{1,2}(T)$.
\end{proof}

Therefore, to show the existence of some spectral point in $\calZ$ beyond $\calQ$, it suffices to construct a tensor $T$ such that $\ACR_{1,2}(T)$ is strictly smaller than $F_\theta(T)$ for every $F_\theta \in \calQ \cap \calZ$, i.e.,
\[ \ACR_{1,2}(T) < \min_{F_\theta \in \calQ \cap \calZ} F_\theta(T). \]
Factoring $\calQ \cap \calZ = (\calQ_{[4]} \cap \calZ) \cup \bigcup_{\{i,j\} \neq \{1,2\}} (\calQ_{ij} \cap \calZ)$, we can rewrite the right-hand side in terms of the different components of quantum functionals:
\begin{itemize}
    \item For $\calQ_{[4]} \cap \calZ$, these are the functionals $F_{(\rho, 1-\rho, 0, 0)}$. Since the third and fourth modes both carry zero weight, this is equivalent to recognizing $V_3\otimes V_4$ as a single mode, so this also coincides with $\calQ_{34} \cap \calZ$. Their minimization is exactly the asymptotic commutative rank for the corresponding $3$-mode tensors. We denote this by $\ACR_{(\{3, 4\})}$ to distinguish it from $\ACR_{1,2}$; as in $\ACR_{(3)}$, the parentheses indicate that the grouped modes $V_3\otimes V_4$ form the distinguished mode, while $V_1$ and $V_2$ are the two sides of the matrix.
    \item For the other $F_\theta \in \calQ_{ij} \cap \calZ$, namely those with $i\in [2]$ and $j\notin [2]$, their minimization amounts to first flattening $T$ into a tensor over the modes $V_i \otimes V_j$, $V_{3-i}$, and the remaining mode, and then taking the asymptotic commutative rank, treating $V_i \otimes V_j$ and $V_{3-i}$ as the two sides of the matrix. We write this as $\ACR_{\{i, j\}}$ for short.
    \item Since the partition corresponding to $\calQ_{12}$ didn't separate $1$ and $2$, we have $\calQ_{12} \cap \calZ = \emptyset$.
\end{itemize}
In sum, we can rewrite the desired inequality as
\[ \ACR_{1,2} (T) < \ACR_{(\{3,4\})}(T), \]
and
\[ \ACR_{1,2} (T) < \ACR_{\{i,j\}}(T) \]
for all $i\in [2]$ and $j\notin [2]$.

Our candidate for separating these quantities is the following tensor:
\[ T_{p,q} = \ang{p}_{1,3} \oplus \ang{q}_{1,4} \oplus \ang{p}_{2,4} \oplus \ang{q}_{2,3}. \]

We first give a basic estimate for $\CR_{\{i,j\}}(T_{p,q})$ for $i\in [2]$ and $j\notin [2]$.
\begin{lemma} \label{lem:acr_ij}
    We have the following lower bounds:
    \[ \ACR_{\{1,3\}}(T_{p,q}), \ACR_{\{2,4\}}(T_{p,q}) \geq q, \quad \ACR_{\{1,4\}}(T_{p,q}), \ACR_{\{2,3\}}(T_{p,q}) \geq p. \]
\end{lemma}
\begin{proof}
    We only prove $\ACR_{\{1,3\}}(T_{p,q}) \geq q$; the others are analogous. First, since $T_{p,q} \geq \ang{q}_{2,3}$, we may restrict attention to this summand. After flattening $T_{p,q}$ into a tensor over the modes $V_1 \otimes V_3$, $V_2$, and $V_4$, the tensor $\ang{q}_{2,3}$ becomes the identity matrix of size $q$ on the modes $V_1\otimes V_3$ and $V_2$, so $\CR_{\{1,3\}}(T_{p,q}) \geq q$. Taking tensor powers and limits gives $\ACR_{\{1,3\}}(T_{p,q}) \geq q$.
\end{proof}

We next compute $\ACR_{(\{3,4\})}(T_{p,q})$. For this we use \Cref{thm:acr_formula}, which will be proved later.
\begin{lemma} \label{lem:acr_34}
    We have $\ACR_{(\{3,4\})}(T_{p,q}) = 2(\sqrt p + \sqrt q)$.
\end{lemma}
\begin{proof}
    After flattening $T_{p,q}$ into a tensor over the modes $V_1$, $V_2$, and $V_3 \otimes V_4$, we can recognize $T_{p,q}$ as
    \[ \ang{p,1,1} \oplus \ang{q,1,1} \oplus \ang{1,1,p} \oplus \ang{1,1,q}, \]
    so by \Cref{thm:acr_formula} we have
    \[ \ACR_{(\{3,4\})}(T_{p,q}) = \min_{\rho \in [0, 1]} p^\rho + q^\rho + p^{1-\rho} + q^{1-\rho} = 2(\sqrt p + \sqrt q) \]
    which is attained at $\rho = 1/2$.
\end{proof}

Finally, we cannot compute $\ACR_{1,2}(T_{p,q})$ by spectral methods, since we do not yet know the spectral points that characterize it, but we can analyze it directly. By definition, to compute $\ACR_{1,2}(T_{p,q})$ we need to estimate $\CR_{1,2}(T_{p,q}^{\otimes n})$ for large $n$. Since $\CR_{1,2}$ is additive, we only need to consider the expansion
\[ T_{p,q}^{\otimes n} = \bigoplus_{n_{13} + n_{14} + n_{23} + n_{24} = n} \binom{n}{n_{13}, n_{14}, n_{23}, n_{24}} \odot \ang{p^{n_{13}}}_{1,3} \otimes \ang{q^{n_{14}}}_{1,4} \otimes \ang{q^{n_{23}}}_{2,3} \otimes \ang{p^{n_{24}}}_{2,4}. \]

We express these tensor products as \emph{graph tensors}, which have been studied before, e.g.~in \cite{CZ19Graph,CVZ19Graph}, although we use a slightly different notation. We use weighted graphs to denote the tensor products
\begin{align*}
    \Ang{\fourcycle{n_{13}}{n_{23}}{n_{24}}{n_{14}}} & \coloneq \ang{n_{13}}_{1,3} \otimes \ang{n_{23}}_{2,3} \otimes \ang{n_{24}}_{2,4} \otimes \ang{n_{14}}_{1,4} \\
    &= \sum_{\substack{i_{13} \in [n_{13}] \\ i_{23} \in [n_{23}] \\ i_{24} \in [n_{24}] \\ i_{14} \in [n_{14}]}} x_{i_{13} i_{14}} y_{i_{23} i_{24}} z_{i_{13} i_{23}} w_{i_{14} i_{24}}.
\end{align*}
More generally, for a graph $G$ with vertex set $[d]$ and edge weights $n_{ij}$, we define the corresponding graph tensor $\ang{G}$ as a $d$-mode tensor by
\[ \ang{G} = \bigotimes_{1\leq i < j \leq d} \ang{n_{ij}}_{i,j}. \]

Now we compute the multilinear commutative rank of these graph tensors. We remark that this is a special case of a more general result known as quantum max-flow \cite{CFSSM16MaxFlow}.
\begin{lemma} \label{lem:quantum-max-flow}
    The multilinear commutative rank is
    \[ \CR_{1,2}\left(\Ang{\fourcycle{n_{13}}{n_{23}}{n_{24}}{n_{14}}}\right) = \min(n_{13}, n_{23}) \cdot \min(n_{14}, n_{24}). \]
\end{lemma}
\begin{proof}
    Let the tensor be $T$. Clearly $\CR_{1,2}(T) \leq n_{13} \cdot n_{14}$ and $\CR_{1,2}(T) \leq n_{23} \cdot n_{24}$, since these are the side lengths of the matrix. To prove $\CR_{1,2}(T) \leq n_{13} n_{24}$, we use the decomposition
    \[ T = \sum_{\substack{i_{13}\in [n_{13}] \\ i_{24} \in [n_{24}]}} \left(\sum_{i_{14} \in [n_{14}]} x_{i_{13} i_{14}} w_{i_{14} i_{24}}\right) \left(\sum_{i_{23} \in [n_{23}]} y_{i_{23} i_{24}} z_{i_{13} i_{23}} \right), \]
    which gives a rank decomposition of size $n_{13} n_{24}$. The same argument gives $\CR_{1,2}(T) \leq n_{14} n_{23}$ as well. This proves the upper bound.

    For the lower bound, without loss of generality we may assume $n_{13}=n_{23}$ and $n_{14}=n_{24}$; otherwise, we can simply restrict to a smaller tensor. In this case, $\ang{n_{13}}_{1,3}\otimes \ang{n_{13}}_{2,3}$ is the matrix multiplication tensor $\ang{n_{13}, 1, n_{13}}$ over the modes $V_1, V_2, V_3$. Since we know that $\ang{n,1,n} \geq \ang{1,n,1} = \ang{n}_{1,2}$, this gives
    $\ang{n_{13}}_{1,3}\otimes \ang{n_{13}}_{2,3} \geq \ang{n_{13}}_{1,2}$, and hence
    \[ T \geq \ang{n_{13}}_{1,2} \otimes \ang{n_{14}}_{1,2} = \ang{n_{13} n_{14}}_{1,2}, \]
    so $\CR_{1,2}(T) \geq n_{13} n_{14}$, which completes the proof.
\end{proof}

\begin{proposition}
    The asymptotic multilinear commutative rank of $T_{p,q}$ is
    \[ \ACR_{1,2}(T_{p,q}) = \min_{\rho\in [0, 1]} 2(p^\rho + q^{1-\rho}). \]
\end{proposition}
\begin{proof}
    \Cref{lem:quantum-max-flow} in fact tells us an important simplification:
    \[ \CR_{1,2}\left( \Ang{\fourcycle{n_{13}}{n_{23}}{n_{24}}{n_{14}}} \right) = \CR_{1,2}\left( \Ang{\fourcycle{n_{13}}{n_{23}}{}{}} \right) \CR_{1,2}\left( \Ang{\fourcycle{}{}{n_{24}}{n_{14}}} \right), \]
    where the unweighted edges have unit weight. We now use the additivity of $\CR_{1,2}$ to get
    \begin{align*}
        \CR_{1,2}(T_{p,q}^{\otimes n})
        &= \CR_{1,2}\left( \bigoplus_{n_{13} + n_{14} + n_{23} + n_{24} = n} \binom{n}{n_{13}, n_{14}, n_{23}, n_{24}} \odot \Ang{\fourcycle{p^{n_{13}}}{q^{n_{23}}}{p^{n_{24}}}{q^{n_{14}}}} \right) \\
        &= \sum_{n_{13} + n_{14} + n_{23} + n_{24} = n} \binom{n}{n_{13}, n_{14}, n_{23}, n_{24}} \CR_{1,2}\left( \Ang{\fourcycle{p^{n_{13}}}{q^{n_{23}}}{p^{n_{24}}}{q^{n_{14}}}} \right) \\
        &= \sum_{k = 0}^n \binom{n}{k} \CR_{1,2}((\ang{p}_{1,3} \oplus \ang{q}_{2,3})^{\otimes k}) \CR_{1,2} ((\ang{q}_{1,4} \oplus \ang{p}_{2,4})^{\otimes n-k}).
    \end{align*}
    But computing $\CR_{1,2}((\ang{p}_{1,3} \oplus \ang{q}_{2,3})^{\otimes k})$ as a multilinear commutative rank is the same as omitting the fourth mode and computing commutative rank as a $3$-mode tensor. Therefore, by \Cref{thm:acr_formula} we have
    \[ \CR_{1,2}((\ang{p}_{1,3} \oplus \ang{q}_{2,3})^{\otimes k}) = r^{k-o(k)}, \]
    where
    \[ r = \min_{\rho\in [0, 1]} p^\rho + q^{1-\rho}. \]
    The same argument applies to $\CR_{1,2} ((\ang{q}_{1,4} \oplus \ang{p}_{2,4})^{\otimes n-k})$, so we have
    \begin{align*}
        \CR_{1,2}(T_{p,q}^{\otimes n}) &= \sum_{k=0}^n \binom{n}{k} r^{k-o(k)} r^{(n-k)-o(n-k)}\\
        &= \sum_{k=0}^n \binom{n}{k}r^{n-o(n)}\\
        &= 2^n r^{n-o(n)}.
    \end{align*}
    This gives $\ACR_{1,2}(T_{p,q}) = 2r = \min_{\rho\in [0, 1]} 2(p^\rho + q^{1-\rho})$, as claimed.
\end{proof}

\begin{theorem} \label{thm:separation}
    There exists a spectral point $\phi \in \calZ \setminus \calQ$ that is needed to characterize $\ACR_{1,2}$, but is not among the quantum functionals, even allowing those obtained by grouping modes together.
\end{theorem}
\begin{proof}
    We take $p \neq q$, both sufficiently large so that $\min(p,q) > 2(\sqrt p + \sqrt q)$. Then, by \Cref{lem:acr_ij,lem:acr_34}, we have
    \[ \min_{\phi \in \calZ \cap \calQ} \phi(T_{p,q}) = \ACR_{(\{3,4\})}(T_{p,q}) = 2(\sqrt p + \sqrt q), \]
    since the remaining terms $\ACR_{\{i,j\}}(T_{p,q})$ are all at least $\min(p,q) > 2(\sqrt p + \sqrt q)$. On the other hand, $\ACR_{1,2}(T_{p,q}) = \min_{\rho\in [0, 1]} 2(p^\rho + q^{1-\rho})$ is strictly smaller than $2(\sqrt p + \sqrt q)$ when $p \neq q$ (the objective attains $2(\sqrt p+\sqrt q)$ when $\rho=1/2$). This gives the desired separation.
\end{proof}

\begin{remark} \label{rmk:interpolation}
    Finally, we remark that, for the purpose of proving the existence of a spectral point beyond $\calQ$ (\Cref{prop:high-order-intro}), one could also appeal to a higher-order analogue of star-convexity proved by Wigderson and Zuiddam \cite[Theorem~12.13]{WZ26Spectra}. In the $4$-mode case, their result says the following. For any spectral point $\phi$ and any flattening rank along one mode $\zeta_{(\varkappa)}$, there is an interpolation between these two spectral points over graph tensors: for any $\lambda \in [0, 1]$, there is a spectral point $\phi_\lambda$ such that for any $4$-vertex weighted graph $G$,
\[ \phi_\lambda(\ang{G}) = \zeta_{(\varkappa)}(\ang{G})^\lambda \phi(\ang{G})^{1-\lambda}. \]
    We indeed have spectral points that are not of the form $\zeta_{(\varkappa)}$, for example the flattening rank along the bipartition $\{\{1,2\}, \{3, 4\}\}$, which we denote by $\zeta_{\{1,2\}, \{3,4\}}$. Applying the above interpolation to $\zeta_{\{1,2\}, \{3,4\}}$ together with the flattening rank $\zeta_{(1)}$ along the first mode, we obtain a spectral point $\phi_\lambda$ such that
\[ \phi_\lambda(\ang{G}) = \zeta_{(1)}(\ang{G})^\lambda \zeta_{\{1,2\}, \{3,4\}}(\ang{G})^{1-\lambda}. \]
    Although this particular functional can be covered by a quantum functional in $\calQ_{34}$, it can be further interpolated with $\zeta_{(3)}$ to obtain a functional $\phi_{\alpha \beta}$ such that
\[ \phi_{\alpha \beta}(\ang{G}) = \zeta_{(1)}(\ang{G})^\alpha \zeta_{(3)}(\ang{G})^\beta \zeta_{\{1,2\}, \{3,4\}}(\ang{G})^{1-\alpha-\beta}. \]
    When $\alpha, \beta > 0$ and $\alpha + \beta < 1$, one can verify that no spectral point in $\calQ$ can cover $\phi_{\alpha \beta}$. This provides an alternative proof of the existence of some spectral point beyond $\calQ$.
\end{remark}

\subsubsection{Discussion of the monomially dominant condition} \label{sec:monomial_dominant}

We show that the monomially dominant hypothesis in \Cref{thm:minimization} is sufficient but not necessary.

\begin{proposition} \label{prop:counterexample-monomial-dom}
    There exists a minimization of spectral points that is not monomially dominant.
\end{proposition}
\begin{proof}
    Let $\calR=\bbN[a,b]$ be the free commutative semiring on two generators. We define a preorder on $\calR$ by
    \[ p \leq q \iff p(1,2)\leq q(1,2) \text{ and } p(2,1)\leq q(2,1), \]
    where $p,q\in \bbN[a,b]$ are viewed as bivariate polynomials.
    It is immediate that this is a Strassen preorder. It is not hard to verify that the asymptotic spectrum of $\calR$ is
    \[ \calX_{(\calR,\leq)}=\{(1,2),(2,1)\}, \]
    where we identify a spectral point with its values on $(a,b)$.

    Next, the subrank of $p\in\calR$ is
    \[ \Q(p)=\max\{n\in\bbN : n\leq p\}=\min\{p(1,2),p(2,1)\}. \]
    This can be written as
    \[ \Q(p)=\min_{\phi\in\calX_{(\calR,\leq)}} \phi(p), \]
    thus $\Q$ is a minimization of spectral points.

    We now show that $\Q$ is not monomially dominant. We have
    \[ \Q(a+b)=\min\{3,3\}=3, \]
    so
    \[ \Q((a+b)^n)=3^n. \]
    On the other hand, for every $0\leq k\leq n$,
    \[
        \Q(a^k b^{n-k})
        =\min\{2^{n-k},2^k\}
        =2^{\min(k,n-k)}.
    \]
    Hence, writing $\lambda=k/n$, we get
    \begin{align*}
        \binom{n}{k}\Q(a^k b^{n-k})
        &= \binom{n}{k} 2^{\min(k,n-k)} \\
        &\leq 2^{\eH(\lambda)n} \cdot 2^{\min(\lambda,1-\lambda)n} \\
        &\leq 2^{(\eH(1/2)+1/2)n}
        = (2^{3/2})^n.
    \end{align*}
    Since $2^{3/2}<3$, this is exponentially smaller than $3^n=\Q(a+b)^n$. Therefore, no choice of $k=k(n)$ can satisfy
    \[ \binom{n}{k}\Q(a^k b^{n-k}) \geq 2^{-o(n)} \Q(a+b)^n, \]
    and so $\Q$ is not monomially dominant. \qedhere
\end{proof}

For some asymptotic spectra of interest, the asymptotic subrank does satisfy the monomially dominant condition. For example, in the framework used to study Shannon capacity \cite{Zui19Graph}, the subrank corresponds to the independence number $\alpha(G)$ of graphs, which is clearly additive: $\alpha(G\sqcup H) = \alpha(G) + \alpha(H)$. Therefore Shannon capacity, i.e., the asymptotic subrank in this framework, is monomially dominant.

However, tensor subrank is not additive \cite[Theorem~5.1]{DMZ24QR}, and the same is true for border subrank \cite[Remark~12]{BCD25BQ}, so at present we cannot tell whether the monomially dominant property holds for asymptotic subrank via \Cref{prop:f-asymp}, even though it is a minimization of spectral points. We believe that either a positive or a negative answer to this question would reveal a deep structure of tensor restriction:
\begin{question}
    Is asymptotic tensor subrank monomially dominant?
\end{question}

\section{Weighted marginal entropy maximization} \label{sec:marginal_entropy}

Let $I_1, I_2, \dots, I_d$ be finite index sets of sizes $n_1,\dots,n_d$, respectively, and let $\Phi\subset I_1 \times \cdots \times I_d$ be a nonempty subset. We let $\calP(\Phi)$ denote the set of probability distributions on $\Phi$, and for each $P\in \calP(\Phi)$ we write $P(\alpha)$ for the probability of the outcome $\alpha\in \Phi$. For $1\leq \varkappa \leq d$ we let $P_\varkappa$ be the marginal distribution of $P$ on the set $I_\varkappa$.

For $\theta \in \bbR_{\geq 0}^d$ satisfying $\theta_1 + \cdots + \theta_d = 1$, we define the maximal $\theta$-weighted marginal entropy of $\Phi$ to be
\[ \eH_\theta(\Phi) = \max_{P \in \calP(\Phi)} \sum_{\varkappa = 1}^d \theta_\varkappa \eH(P_\varkappa). \]
This notion is crucial in the definition of support functionals defined by Strassen \cite{Str91SuppFunc}.

By the basic entropy inequality $\eH(P) \leq \log_2  |{\supp (P)}|$, we have $\eH_\theta(\Phi) \leq \sum_{\varkappa=1}^d \theta_\varkappa \log_2 n_\varkappa$. An almost immediate consequence is that $\eH_\theta(\Phi)$ attains this upper bound if and only if there is a distribution $P\in \calP(\Phi)$ such that every marginal $P_\varkappa$ is uniform. When such a distribution exists, we say $\Phi$ is \emph{balanced}.

\subsection{Maximizer condition}

Strassen analyzed the maximizer of the concave optimization problem defining $\eH_\theta$ to deduce an equivalence between two different definitions of upper support functional \cite[Proposition~2.3]{Str91SuppFunc}. In this section, we extend Strassen's analysis from $3$ modes to $d$ modes.

\begin{proposition}\label{prop:entropy-maximization}
    A distribution $P\in \calP(\Phi)$ is a maximizer of $\eH_\theta(\Phi)$ if and only if, for every $\alpha\in \Phi$,
    \[ P(\alpha) > 0 \implies \alpha \in \argmax h_P, \]
    where
    \[ h_P(\alpha) = \sum_{\varkappa=1}^d \theta_\varkappa \log \frac 1{P_\varkappa(\alpha_\varkappa)}. \]
\end{proposition}

\begin{proof}
    Let
    \[ F(Q) \coloneq \sum_{\varkappa=1}^d \theta_\varkappa \eH(Q_\varkappa), \]
    so that $\eH_\theta(\Phi)=\max_{Q\in\calP(\Phi)} F(Q)$.
    Since each marginal map $Q \mapsto Q_\varkappa$ is linear and Shannon entropy is concave, $F$ is a concave function on the simplex $\calP(\Phi)$.

    We first compute the first variation of $F$ at a point $P$.
    Let $u\colon \Phi \to \bbR$ satisfy $\sum_{\alpha\in\Phi} u(\alpha)=0$, and write $u_\varkappa$ for its marginal on $I_\varkappa$.
    Then for all sufficiently small $t$, the point $P+tu$ still lies in the affine span of $\calP(\Phi)$, and
    \begin{align*}
        \left.\frac{\dif}{\dif t}\right|_{t=0} F(P+tu)
        &= \sum_{\varkappa=1}^d \theta_\varkappa \sum_{i\in I_\varkappa} u_\varkappa(i)
            \log \frac{1}{P_\varkappa(i)} \\
        &= \sum_{\alpha\in\Phi} u(\alpha)
            \sum_{\varkappa=1}^d \theta_\varkappa \log \frac{1}{P_\varkappa(\alpha_\varkappa)} \\
        &= \sum_{\alpha\in\Phi} u(\alpha) h_P(\alpha).
    \end{align*}
    Indeed, differentiating $-x\log x$ produces an additive constant, but this constant disappears because each $u_\varkappa$ has total sum $0$.

    We now prove the necessity.
    Assume that $P$ maximizes $F$.

    First, if $\alpha,\beta\in\Phi$ both satisfy $P(\alpha),P(\beta)>0$, then for all sufficiently small $t$ the perturbation
    \[ P+t(\delta_\alpha-\delta_\beta) \]
    remains in $\calP(\Phi)$.
    Since $t=0$ is a maximum point of the one-variable concave function
    $t \mapsto F(P+t(\delta_\alpha-\delta_\beta))$, its derivative at $0$ must vanish, and therefore
    \[ h_P(\alpha)-h_P(\beta)=0. \]
    So $h_P$ is constant on the support of $P$; denote this common value by $M$.

    Next, fix any $\beta\in\Phi$.
    The segment $(1-t)P+t\delta_\beta$ lies in $\calP(\Phi)$ for $0\le t\le 1$, hence optimality of $P$ gives
    \[ 0 \ge \left.\frac{\dif}{\dif t}\right|_{t=0^+} F\bigl((1-t)P+t\delta_\beta\bigr). \]
    Using the first-variation formula with $u=\delta_\beta-P$, we obtain
    \begin{align*}
        0
        &\ge h_P(\beta)-\sum_{\alpha\in\Phi} P(\alpha) h_P(\alpha) \\
        &= h_P(\beta)-M.
    \end{align*}
    Hence $h_P(\beta)\le M$ for every $\beta\in\Phi$, so every point in the support of $P$ belongs to $\argmax h_P$.

    Conversely, assume that every $\alpha$ with $P(\alpha)>0$ belongs to $\argmax h_P$.
    Let $M=\max_{\alpha\in\Phi} h_P(\alpha)$.
    Then $h_P(\alpha)=M$ on the support of $P$, and $h_P(\alpha)\le M$ for all $\alpha\in\Phi$.

    For any $Q\in\calP(\Phi)$, concavity of $F$ gives
    \[ F(Q) \le F(P) + \sum_{\alpha\in\Phi} (Q(\alpha)-P(\alpha)) h_P(\alpha). \]
    Therefore
    \begin{align*}
        F(Q)-F(P)
        &\le \sum_{\alpha\in\Phi} Q(\alpha) h_P(\alpha) - \sum_{\alpha\in\Phi} P(\alpha) h_P(\alpha) \\
        &\le M \sum_{\alpha\in\Phi} Q(\alpha) - M \sum_{\alpha\in\Phi} P(\alpha) \\
        &= 0.
    \end{align*}
    Thus $F(Q)\le F(P)$ for all $Q\in\calP(\Phi)$, so $P$ is indeed an argmax of $\eH_\theta(\Phi)$.
\end{proof}

\subsection{The two-dimensional case}

In this section, we give a more structured conclusion of optimal distributions in the $d=2$ case. In this section we let $J, K$ denote the two index sets and $\Phi \subset J \times K$. We begin with a finer description of the maximizer distribution.

\begin{proposition} \label{prop:argmax_2d}
    Assume $\Phi$ has no empty rows or columns. Consider a partition
    \begin{center}
        $J = J_1 \sqcup \cdots \sqcup J_m$ \quad and \quad $K = K_1 \sqcup \cdots \sqcup K_m$, 
    \end{center}
    such that:
    \begin{enumerate}
        \item (Diagonal balanced) For any $1\leq u\leq m$, the set $\Phi^{(u)} = \Phi \cap (J_u \times K_u)$ is a balanced subset of
        $J_u\times K_u$.
        \item (Blockwise triangular) Any $(j, k)\in \Phi$ lives in some $J_u \times K_v$ where $u\geq v$.
        \item (Convexity) Writing $\mu_u = |J_u| / |K_u|$, we have
        \[ \mu_1 > \mu_2 > \cdots > \mu_m. \]
    \end{enumerate}
    For each $u$, let $P^{(u)}$ be a distribution on $\Phi^{(u)}$ with uniform marginals. Then
    \begin{center}
        $\displaystyle P^\star = \sum_{u=1}^m p_u P^{(u)}$, \quad where \quad
        $\displaystyle p_u = \frac{|J_u|^\rho |K_u|^{1-\rho}}{\sum_{v=1}^m |J_v|^\rho |K_v|^{1-\rho}}$,
    \end{center}
    is a distribution attaining $\eH_\theta(\Phi)$ for $\theta = (\rho, 1-\rho)$, and
    \[ \eH_\theta(\Phi) = \log_2 \left( \sum_{u=1}^m |J_u|^\rho |K_u|^{1-\rho} \right). \]
\end{proposition}
\begin{proof}
    Let $D = \sum_{u=1}^m |J_u|^\rho |K_u|^{1-\rho}$. For each $j\in J_u$, the $J$-marginal is
    \begin{align*}
        P^\star_J(j) &= p_u P^{(u)}_J(j) \\
        &= \frac{|J_u|^{\rho} |K_u|^{1-\rho}}{D} \cdot \frac 1{|J_u|}\\
        &= \frac{\mu_u^{\rho-1}}{D}
    \end{align*}
    since $P^{(u)}$ has uniform marginals. Similarly, for each $k \in K_u$ the $K$-marginal is
    \[ P^\star_K(k) = \frac{|J_u|^{\rho} |K_u|^{1-\rho}}{D} \cdot \frac 1{|K_u|} = \frac{\mu_u^{\rho}}{D}. \]
    For every $(j, k) \in \Phi$, since $j\in J_u, k\in K_v$ with $u\geq v$, it has $h$-value
    \begin{align*}
        h_{P^\star}(j, k) &= \rho \log \left(\frac{D}{\mu_u^{\rho-1}}\right)
        + (1-\rho) \log \left(\frac{D}{\mu_v^{\rho}}\right)\\
        &= \log D + \rho(1-\rho) (\log \mu_u  - \log \mu_v)\\
        &\leq \log D.
    \end{align*}
    For any $(j, k)\in \supp P^\star$ we have $u=v$, so it satisfies the condition of \Cref{prop:entropy-maximization}, hence $P^\star$ is a maximizer.

    Finally, we compute $\eH_\theta(\Phi)$. The distribution $P^\star$ can be generated by first sampling the block index $u$ with probability $|J_u|^\rho |K_u|^{1-\rho} / D$ and then sampling $(j, k)\sim P^{(u)}$. By the chain rule,
    \[ \eH(P^\star_J) = \eH(u) + \eH(P^{(u)}_J \mid u) \]
    and, since each $P^{(u)}_J$ is the uniform distribution on $J_u$, we have $\eH(P^{(u)}_J \mid u) = \bbE_u[\log_2 |J_u|]$. The analogous statement holds for $K$. Therefore,
    \begin{align*}
        \eH_\theta(\Phi) &= \rho (\eH(u) + \bbE_u[\log_2 |J_u|]) + (1-\rho)(\eH(u) + \bbE_u[\log_2 |K_u|])\\
        &= \eH(u) + \rho \bbE_u[\log_2 |J_u|] + (1-\rho) \bbE_u[\log_2 |K_u|]\\
        &= \bbE_u \left[ \log_2 \frac{D}{|J_u|^\rho |K_u|^{1-\rho}} \right] + \bbE_u[\log_2 |J_u|^\rho |K_u|^{1-\rho}]\\
        &= \log_2 D. \qedhere
    \end{align*}
\end{proof}

Recall that, when $|J| = |K|$, saying that the set $\Phi$ itself is balanced means that there is a \emph{bistochastic matrix} whose support is contained in $\Phi$. The \emph{Birkhoff--von Neumann theorem} says that any bistochastic matrix can be written as a convex combination of permutation matrices; this in particular implies that $\Phi$ contains a permutation. In the language of graph theory, this means that $\Phi$ is balanced if and only if the bipartite graph with bi-adjacency matrix $\Phi$ admits a perfect matching. This is equivalent to \emph{Hall's condition}: $\Phi$ admits a perfect matching if and only if every $J' \subseteq J$ satisfies $|N(J')| \geq |J'|$, where $N(J') = \bigcup_{j\in J'} \Phi_j$, and $\Phi_j = \{k : (j, k)\in \Phi\}$.

We now describe a generalization of Hall's condition to the case when $|J|$ may not equal $|K|$.
\begin{proposition}\label{prop:balance_char}
    The subset $\Phi$ is balanced if and only if for every $J'\subseteq J$,
    \[ |N(J')| \geq \frac{|K|}{|J|} |J'|. \]
\end{proposition}
\begin{proof}
    We formulate the search for a distribution on $\Phi$ with uniform marginals as a maximum-flow problem on a bipartite network. For a comprehensive introduction to network flows and the foundational theorems used here, we refer the reader to Ahuja, Magnanti, and Orlin \cite[Chap.~6]{AMO93Flow}.

    We construct a directed flow network $G = (V, E)$ as follows. Let the vertex set be $V = \{s, t\} \sqcup J \sqcup K$, where $s$ is a designated source node and $t$ is a designated sink node. We define the directed edges and their capacities $c(u, v)$ as:
    \begin{enumerate}
        \item For each $j \in J$, add an edge $(s, j)$ with capacity $c(s, j) = \frac{1}{|J|}$.
        \item For each $k \in K$, add an edge $(k, t)$ with capacity $c(k, t) = \frac{1}{|K|}$.
        \item For each pair $(j, k) \in \Phi$, add an edge $(j, k)$ with infinite capacity, i.e., $c(j, k) = \infty$.
    \end{enumerate}

    A distribution on $\Phi$ with uniform marginals is the same as a vector $(p_{jk})_{(j,k)\in\Phi}$ attaining the value $1$ in the linear program
    \[ \text{maximize } \sum_{(j,k)\in\Phi} p_{jk} \quad \text{subject to} \quad p_{jk}\geq 0, \quad \sum_{k} p_{jk} \le \frac{1}{|J|}\ (j\in J), \quad \sum_{j} p_{jk} \le \frac{1}{|K|}\ (k\in K). \]
    Indeed, if the total mass is $1$, then all these inequality constraints must be tight.
    The variables $p_{jk}$ in this linear program correspond to the flow on the edges $(j, k)$. The constraints $\sum_{k} p_{jk} \le \frac{1}{|J|}$ and $\sum_{j} p_{jk} \le \frac{1}{|K|}$ correspond to the capacity constraints on the edges leaving the source $s$ and entering the sink $t$, respectively. The optimum of this linear program is exactly the maximum flow from $s$ to $t$ in $G$. 

    Since the total capacity of the edges leaving $s$ is $\sum_{j \in J} \frac{1}{|J|} = 1$, the maximum flow is bounded above by $1$. The LP optimum attains $1$ (meaning there is a distribution $P \in \calP(\Phi)$ having uniform marginals, and thus $\Phi$ is balanced) if and only if the maximum flow in $G$ is exactly $1$.

    By the Max-Flow Min-Cut theorem \cite[Theorem~6.3]{AMO93Flow}, the maximum flow in a network equals the minimum capacity of an $s$-$t$ cut. An $s$-$t$ cut is a partition of $V$ into two disjoint sets $S$ and $T$ such that $s \in S$ and $t \in T$. The capacity of the cut is the sum of the capacities of all edges directed from $S$ to $T$. Therefore, the maximum flow is $1$ if and only if every $s$-$t$ cut has a capacity of at least $1$.

    Consider an arbitrary $s$-$t$ cut $(S, T)$. Let $A = S \cap J$ and $B = S \cap K$. The capacity of this cut is given by:
    $$ C(S, T) = \sum_{j \in J \setminus A} c(s, j) + \sum_{k \in B} c(k, t) + \sum_{j \in A, k \in K \setminus B} c(j, k). $$

    If there exists an edge $(j, k) \in \Phi$ such that $j \in A$ and $k \notin B$, the term $c(j, k)$ contributes $\infty$ to the cut capacity. Because we are only interested in whether the minimum cut capacity is at least $1$, cuts with infinite capacity trivially satisfy the condition. To find the minimum cut, we restrict our attention to cuts with finite capacity. This requires that $B \supseteq \{k \in K : \exists j \in A \text{ such that } (j, k) \in \Phi\}$. 

    Let $N(A) = \bigcup_{j \in A} \Phi_j$ denote the neighborhood of $A$ in $\Phi$. The condition for a finite-capacity cut is simply $B \supseteq N(A)$. For such cuts, the capacity simplifies to:
    $$ C(S, T) = \frac{|J \setminus A|}{|J|} + \frac{|B|}{|K|} = 1 - \frac{|A|}{|J|} + \frac{|B|}{|K|}. $$

    For the minimum cut capacity to be at least $1$, we must have:
    $$ 1 - \frac{|A|}{|J|} + \frac{|B|}{|K|} \ge 1 \implies \frac{|B|}{|K|} \ge \frac{|A|}{|J|}, $$
    for all subsets $A \subseteq J$ and all valid $B \supseteq N(A)$. For a fixed $A$, the cut capacity is minimized when $B$ is chosen to be as small as possible, which is exactly $B = N(A)$. 

    Thus, the condition holds if and only if for every $A \subseteq J$:
    $$ \frac{|N(A)|}{|K|} \ge \frac{|A|}{|J|}. $$

    Replacing $A$ with $J'$, we obtain the generalized Hall's condition.
\end{proof}

\section{Matrix spaces as quiver representations} \label{sec:quiver}

In this section, we use the language of quiver representations to organize matrix spaces and the structural notions attached to them. Although this viewpoint has deep connections to geometric invariant theory, we will only need the linear-algebraic features relevant to our argument. We will now review the basics of quiver representations which are relevant to our setting. Along the way we prove several useful facts which are likely known by experts in this area, but for which we did not find standard references.

We first recall the notion of a quiver representation, following the textbook of Derksen and Weyman \cite{DW17Quiver}. A \emph{quiver} $Q$ is a pair $Q=(Q_0, Q_1)$ consisting of a \emph{vertex set} $Q_0$ and an \emph{arrow set} $Q_1$ (directed edges). Each arrow $a\in Q_1$ has tail $ta$ and head $ha$. A \emph{representation} $\calV$ of $Q$ consists of
\begin{itemize}
    \item a finite-dimensional vector space $\calV(x)$ for each vertex $x\in Q_0$, and
    \item a linear map $\calV(a)\colon \calV(ta) \to \calV(ha)$ for each arrow $a\in Q_1$.
\end{itemize}

For our purposes, the main example is a $k$-tuple of $n\times m$ matrices. This can be viewed as a representation of a quiver with two vertices $0,1$ and $k$ arrows from $0$ to $1$. This quiver is the \emph{(generalized) Kronecker quiver} $\theta(k)$.
\[\begin{tikzcd}
	0 & 1
	\arrow["\cdots"{description}, draw=none, from=1-1, to=1-2]
	\arrow[shift left=3, from=1-1, to=1-2]
	\arrow[shift right=3, from=1-1, to=1-2]
\end{tikzcd}\]
If $U=\calV(0)$ and $V=\calV(1)$ have dimensions $n,m$ respectively, then the maps $\calV(a_1),\dots,\calV(a_k)$ are exactly a $k$-tuple of $n\times m$ matrices.

From now on, we assume $Q=\theta(k)$ is this Kronecker quiver, so every $\calV$ can be regarded as a matrix tuple. Let $\Rep_{(n, m)}(\theta(k))$ denote the space of $\theta(k)$-representations on fixed vector spaces $U,V$. This is an $nmk$-dimensional vector space. Hence the following three objects are equivalent:
\begin{center}
    $\theta(k)$-representation $\calV$ \quad $\cong$ \quad
    $n\times m$ matrix $k$-tuple $\mathcal{B}$ \quad $\cong$ \quad
    $n\times m\times k$ tensor $T$.
\end{center}

For subspaces $U'\subseteq U$ and $V'\subseteq V$, if $\calV(a_i)(U') \subseteq V'$ for all $i$, then the restricted maps define a representation $\calV'$ on $U',V'$. In this case, we call $\calV'\subseteq\calV$ a \emph{subrepresentation}.

If $\calV'$ is a subrepresentation of $\calV$, then each map $\calV(a_i)$ induces a map $U/U' \to V/V'$. The resulting quotient representation is denoted by $\calV/\calV'$.

\subsection{Shrunk subspaces}

An integer pair $\sigma=(\sigma_0,\sigma_1)\in\bbZ^2$ is called a \emph{weight}, and defines
\[ \sigma(\calV) = \sigma_0 \dim U + \sigma_1 \dim V. \]
If $\sigma(\calV)=0$, we say that $\calV$ is \emph{$\sigma$-semistable} if
\[ \sigma(\calV') \le 0 \]
for every subrepresentation $\calV'$. For a representation on vector spaces of dimensions $(n,m)$, the notion of semistability used in this paper is exactly $\sigma$-semistability for the weight
\[ \sigma=(m,-n), \]
since
\[ \sigma(\calV') \le 0 \iff m\dim U' - n\dim V' \le 0 \iff \frac{\dim U'}{\dim V'} \le \frac{n}{m}. \]

For a nonzero representation $\calV$ on vector spaces $U,V$, define
\[ \mu^\star(\calV) = \frac{\dim U}{\dim V}, \]
with the convention that $\mu^\star(\calV)=\infty$ when $\dim V=0$. We say that $\calV$ is \emph{semistable} if
\[ \mu^\star(\calV') \leq \mu^\star(\calV) \]
for every nonzero subrepresentation $\calV'$, and \emph{unstable} otherwise. This is the notion of semistability used throughout the paper. It is equivalent to the standard quiver-theoretic notion for the Kronecker quiver, but we will not need that perspective here.

For a subspace $U^* \subseteq U$, let
\[ \calV(U^*) \coloneq \sum_{i=1}^k \calV(a_i)(U^*) \subseteq V. \]
In other words, $\calV(U^*)$ is the span of the images of $U^*$ under the $k$ maps defining $\calV$. The following criterion is an immediate reformulation of instability for the slope $\mu^\star$.
\begin{proposition} \label{prop:unstable}
    A representation $\calV$ is unstable if and only if there exists a nonzero subspace $U^*\subseteq U$ such that
    \[ \frac{\dim \calV(U^*)}{\dim U^*} < \frac{\dim V}{\dim U}. \]
\end{proposition}
\begin{proof}
    Suppose first that such a subspace $U^*$ exists, and set $V^*=\calV(U^*)$. Then $(U^*,V^*)$ is a subrepresentation of $\calV$, and
    \[ \mu^\star(U^*,V^*) = \frac{\dim U^*}{\dim V^*} > \frac{\dim U}{\dim V} = \mu^\star(\calV). \]
    Hence $\calV$ is not semistable, i.e., it is unstable.

    Conversely, suppose that $\calV$ is unstable. Then by definition there is a nonzero subrepresentation $(U^*,V^*)$ with
    \[ \mu^\star(U^*,V^*) > \mu^\star(\calV). \]
    Since $\calV(U^*)\subseteq V^*$, we have
    \[ \frac{\dim \calV(U^*)}{\dim U^*} \leq \frac{\dim V^*}{\dim U^*} < \frac{\dim V}{\dim U}, \]
    as required.
\end{proof}

Any subspace $U^*$ satisfying \Cref{prop:unstable} will be called a \emph{shrunk subspace} of $\calV$. In the square case $\dim U=\dim V$, the inequality becomes
\[ \dim \calV(U^*) < \dim U^*, \]
which is exactly the classical notion of a shrunk subspace in the study of non-commutative rank \cite{IQS17NCR,IQS18NCR}.

\begin{proposition} \label{prop:filt_preserve_stab}
    Let $\calV'$ be a subrepresentation of $\calV$. Assume that
    \begin{itemize}
        \item $\mu^\star(\calV') = \mu^\star(\calV/\calV')$, equivalently $\mu^\star(\calV') = \mu^\star(\calV)$,
        \item $\calV'$ and $\calV/\calV'$ are semistable.
    \end{itemize}
    Then $\calV$ is also semistable.
\end{proposition}
\begin{proof}
    Since $\mu^\star(\calV')=\mu^\star(\calV/\calV')=\mu^\star(\calV)$, let us denote this common value by $\mu^\star$.

    Let $U^*\subseteq U$ be arbitrary. We decompose $\calV(U^*)$ according to the subrepresentation $\calV'$:
    \[ \dim \calV(U^*) = \dim \calV(U^* \cap U') + \dim \frac{\calV(U^*)}{\calV(U^* \cap U')}. \]
    Since $\calV'$ is semistable, \Cref{prop:unstable} gives
    \[ \dim \calV(U^* \cap U') \geq (\mu^\star)^{-1} \dim(U^*\cap U'). \]

    For the second summand, note that
    \begin{align*}
        \dim \frac{\calV(U^*)}{\calV(U^* \cap U')} &\geq \dim \frac{\calV(U^*) + V'}{\calV(U^* \cap U') + V'}\\
        &= \dim \frac{\calV(U^*) + V'}{V'}\\
        &= \dim (\calV/\calV')\left(\frac{U^* + U'}{U'}\right)\\
        &\geq (\mu^\star)^{-1} \dim \frac{U^* + U'}{U'}\\
        &= (\mu^\star)^{-1} \dim \frac{U^*}{U^* \cap U'},
    \end{align*}
    where the second inequality uses the semistability of $\calV/\calV'$ and \Cref{prop:unstable}.

    Adding the two bounds yields
    \[ \dim \calV(U^*) \geq (\mu^\star)^{-1} \dim U^*. \]
    Since this holds for every subspace $U^*\subseteq U$, \Cref{prop:unstable} implies that $\calV$ is semistable.
\end{proof}

Let $\calV$ be a representation of $\theta(k)$ and $\calV'$ be a representation of $\theta(k')$. We define their tensor product $\calV \otimes \calV'$ to be a representation of $\theta(k k')$ given by arrows $a_{ii'}$ where $1\leq i\leq k, 1\leq i'\leq k'$. The two linear spaces of $\calV\otimes\calV'$ are simply $U\otimes U'$ and $V\otimes V'$, and the linear maps are given by $(\calV \otimes \calV')(a_{ii'}) = \calV(a_i) \otimes \calV'(a_{i'})$.
\begin{proposition} \label{prop:tensor_preserve_stab}
    If $\calV$ and $\calV'$ are semistable, then so is $\calV \otimes \calV'$.
\end{proposition}
\begin{proof}
    Let $U^* \subseteq U \otimes U'$ be arbitrary. We write
    \begin{align*}
        (\calV \otimes \calV')(U^*) &= \sum_{i=1}^k \sum_{i'=1}^{k'} (\calV(a_i) \otimes \calV'(a_{i'}))(U^*)  \\
        &= \sum_{i=1}^k \sum_{i'=1}^{k'} (\calV(a_i) \otimes \id_{V'}) (\id_{U} \otimes \calV'(a_{i'})) (U^*).
    \end{align*}
    Let $\calV^* = \calV \otimes \id_{V'}$ and $\calV'^* = \id_U \otimes \calV'$. Thus $\calV^*$ is a $\theta(k)$-representation on
    \[ U\otimes V' \to V\otimes V', \]
    and $\calV'^*$ is a $\theta(k')$-representation on
    \[ U\otimes U' \to U\otimes V'. \]
    Choose a complete flag in $V'$. Tensoring this flag with $U$ and with $V$ gives a filtration of $\calV^*$ whose successive quotients are all isomorphic to $\calV$. Repeatedly applying \Cref{prop:filt_preserve_stab}, we conclude that $\calV^*$ is semistable, with
    \[ \mu^\star(\calV^*)=\mu^\star(\calV). \]
    Likewise, a complete flag in $U$ induces a filtration of $\calV'^*$ whose successive quotients are all isomorphic to $\calV'$, so $\calV'^*$ is semistable and
    \[ \mu^\star(\calV'^*)=\mu^\star(\calV'). \]
    Therefore,
    \begin{align*}
        \dim (\calV \otimes \calV')(U^*) &= \dim \calV^*(\calV'^*(U^*))\\
        & \geq \mu^\star(\calV)^{-1} \dim \calV'^*(U^*)\\
        & \geq \mu^\star(\calV)^{-1} \mu^\star(\calV')^{-1} \dim U^*\\
        &= \mu^\star(\calV \otimes \calV')^{-1} \dim U^*.
    \end{align*}
    Since this holds for every subspace $U^*\subseteq U\otimes U'$, \Cref{prop:unstable} shows that $\calV \otimes \calV'$ is semistable.
\end{proof}

Let $\bbK / \bbF$ be a field extension and $\calV$ a representation over $\bbF$. We write $\calV_\bbK$ for its base change, whose underlying spaces are $U_\bbK = \bbK \otimes_\bbF U$ and $V_\bbK = \bbK \otimes_\bbF V$. Just as non-commutative rank is invariant under base change \cite[Lemma~5.3]{IQS18NCR}, semistability is also preserved. We give an elementary proof; see also \cite[Proposition~2.4]{HS20Stab} for general quivers.

\begin{proposition}\label{prop:base_change_stability}
    For a semistable representation $\calV$ over $\bbF$, its base change $\calV_\bbK$ is also semistable.
\end{proposition}
\begin{proof}
    We first reduce the general case to a finite field extension. Suppose that $\calV_\bbK$ is unstable. By \Cref{prop:unstable}, this is witnessed by a shrunk subspace $U^*$. The existence of such a subspace is a purely algebraic condition, so Hilbert's Nullstellensatz implies that a witness exists over the algebraic closure $\ov{\bbF}$. Since $U^*$ is spanned by finitely many vectors, these vectors are defined over some finite extension of $\bbF$, over which $U^*$ remains a shrunk subspace.

    We may therefore assume that $\bbK/\bbF$ is finite of degree $d = [\bbK : \bbF]$, and let $\alpha_1,\dots,\alpha_d$ be an $\bbF$-basis of $\bbK$. If $\calV_\bbK$ were unstable, then by \Cref{prop:unstable} there would be a shrunk subspace $U^*$ over $\bbK$ such that $\dim_{\bbK}\calV_\bbK(U^*) < (\mu^\star)^{-1} \dim_{\bbK} (U^*)$, where we write $\mu^\star$ as shorthand for $\mu^\star(\calV) = \mu^\star(\calV_\bbK)$. Recall that
    \[ \calV_\bbK(U^*) = \sum_{i=1}^k \calV_\bbK(a_i) (U^*), \]
    and the right-hand side is already a $\bbK$-subspace. Hence
    \[ \calV_\bbK(U^*) = \sum_{j=1}^d \sum_{i=1}^k \alpha_j \calV_\bbK(a_i) (U^*). \]
    Viewing $U^*$ as an $\bbF$-subspace of $U_\bbK$, we may write the right-hand side as a sum of the ranges of $\bbF$-linear maps:
    \[ \calV_\bbK(U^*) = \sum_{j=1}^d \sum_{i=1}^k (\alpha_j \otimes_\bbF \calV(a_i)) (U^*) = (\calW \otimes \calV)(U^*), \]
    where $\calW$ is the $\theta(d)$-representation over $\bbF$ whose arrow maps are multiplication by $\alpha_j$ on $\bbK$.

    For every nonzero $x\in \bbK$, we have $\calW(\bbF x) = \bbK$, so $\calW$ is semistable. Thus \Cref{prop:tensor_preserve_stab} shows that $\calW \otimes \calV$ is semistable, and such a shrunk subspace $U^*$ cannot exist.
\end{proof}

\subsection{Harder--Narasimhan filtration}

The following theorem tells us that every quiver representation admits a canonical decomposition into semistable factors, known as the \emph{Harder--Narasimhan filtration (HN-filtration)}. This result actually holds for any slope function and even has a much broader statement not restricted to quiver representations. But we only state the following version for our purpose.

\begin{theorem}[Harder--Narasimhan filtration {\cite[Theorem~2.5]{HdlP02HNFilt}}] \label{thm:HN-filtration}
    Every representation $\calV$ has a filtration
    \[ 0 = \calV_0 \subsetneq \calV_1 \subsetneq \cdots \subsetneq \calV_r = \calV, \]
    such that
    \begin{enumerate}
        \item For every $1\leq i < r$ we have $\mu^\star(\calV_i / \calV_{i-1}) > \mu^\star(\calV_{i+1}/\calV_i)$.
        \item For every $1\leq i\leq r$, $\calV_i/\calV_{i-1}$ is semistable.
    \end{enumerate}
    Moreover, such a filtration is unique.
\end{theorem}

Since the HN-filtration is unique, we say these $\calV_i/\calV_{i-1}$ are the \emph{subquotients} of $\calV$, and the \emph{dimension data} of $\calV$ is $(n_i, m_i) = (\dim (U_i/U_{i-1}), \dim(V_i/V_{i-1}))$.

We say a representation $\calV$ is \emph{concise} if it doesn't have slope $0$ or $\infty$ in its HN-filtration.
\begin{proposition} \label{prop:tensor_hn}
    Let $\calV, \calV'$ be two concise representations. Then their tensor product $\calV \otimes \calV'$ is concise, and each slope appearing in the HN-filtration of $\calV \otimes \calV'$ is a product of one slope in $\calV$ and one slope in $\calV'$. More precisely, let
    \[ \mu_i = \frac{n_i}{m_i}, \quad \mu'_j = \frac{n'_j}{m'_j} \]
    be the slopes of the subquotients of $\calV$ and $\calV'$. Then the subquotient of $\calV\otimes \calV'$ having slope $\nu$ has dimension
    \[ \left( \sum_{\mu_i \mu'_j = \nu} n_i n'_j,  \sum_{\mu_i \mu'_j = \nu} m_i m'_j \right). \]
\end{proposition}

\begin{proof}
    Write $\calV^* = \calV \otimes \calV'$, and let the lengths of the filtrations of $\calV$ and $\calV'$ be $r$ and $r'$, respectively. Choose an ordering $\{(i_t, j_t)\}_{1\leq t\leq r \cdot r'}$ of all pairs $(i,j)$ such that
    \begin{center}
        $\nu_1 \geq \nu_2 \geq \cdots \geq \nu_{r\cdot r'}$ \quad where \quad $\nu_t = \mu_{i_t} \mu'_{j_t}$.
    \end{center}
    Then we define the filtration by
    \[ \calV_t^* =(U^*_t, V^*_t), \quad U^*_t = \sum_{s=1}^t U_{i_s} \otimes U'_{j_s}, \quad V^*_t = \sum_{s=1}^t V_{i_s} \otimes V'_{j_s}, \]
    and it's not hard to see that $\calV^*_t / \calV^*_{t-1} = (\calV_{i_t}/\calV_{i_t-1}) \otimes (\calV'_{j_t}/\calV'_{j_t-1})$, which is semistable with slope $\nu_t$ by \Cref{prop:tensor_preserve_stab}.

    This filtration need not yet be the Harder--Narasimhan filtration, because its slopes need not be strictly decreasing. By \Cref{prop:filt_preserve_stab}, grouping adjacent subquotients of equal slope produces the desired coarser filtration, with subquotients of the stated dimensions.
\end{proof}

\section{Basis shifting} \label{sec:basis_shift}

In this section we prove a structural result on tensors, refining the \emph{basis shifting} theorem of Wigderson--Zuiddam \cite[Theorem~11.9]{WZ26Spectra}. Results of this type go back to Strassen \cite[\S6]{Str88AsymSpec}, where they were used to establish a compression property of matrix multiplication tensors and, in turn, the star-convexity of $\specMM$.

We will then use this refinement to prove another compression statement: any semistable tensor in the sense of \cref{sec:quiver} restricts to a matrix multiplication tensor while losing only a constant fraction of the dimensions in the $U,V$ modes. Recall that Fortin--Reutenauer \cite[Corollary~2]{FR04NCRk} proved
\[ \CR(T) \geq \NCR(T)/2. \]
In particular, if $T$ has full non-commutative rank, then
\[ T \geq \ang{1,\lceil n/2\rceil,1}. \]
Our result can be viewed as a rectangular generalization of this phenomenon.

Unlike in \cref{sec:quiver}, here we switch the roles of the tensor modes. For an $n\times m\times k$ tensor $T\in U\otimes V\otimes W$, let $U=\bbF^n$, $V=\bbF^m$, and $W=\bbF^k$ denote the first, second, and third modes, respectively. We regard $T$ as an $n\times k$ array of vectors $T_{il}\in V$ ($1\le i\le n$, $1\le l\le k$), whose rows correspond to $U$ and columns to $W$. For $A\in\GL_n(\bbF)$ and $B\in\GL_k(\bbF)$, we write $A\cdot T\cdot B$ for the array obtained by linear recombination of rows and columns via left and right multiplication by $A$ and $B$.

Throughout this section, $\bbF$ is an infinite field. We say a property holds for a \emph{generic} choice of parameters $x_1,\dots,x_N\in\bbF$ if it holds on a nonempty Zariski open subset of $\bbF^N$.

\begin{theorem} \label{thm:shift}
    For any $n\times k$ $\bbF^m$-valued array $T$, there exist integers
    \[ \lambda_1 \geq \cdots \geq \lambda_k \geq 0 \]
    such that, for a generic pair $(A,B)\in \GL_n(\bbF)\times \GL_k(\bbF)$, there is a lower unitriangular matrix $L$ for which, writing
    \[ T' = LA\cdot T\cdot B, \]
    \begin{enumerate}
        \item The vectors $T'_{il}$ with $1\leq l\leq k$ and $1\leq i\leq \lambda_l$ form a basis of $\Span T' \subseteq \bbF^m$.
        \label{item:basis}
        \item For every $1\leq p\leq k$, if $i' > \lambda_p$ and $l' \geq p$, then $T'_{i'l'}$ is a linear combination of $\{T'_{il}\}_{i\leq \lambda_l,\; l\leq p}$.
        \label{item:span}
    \end{enumerate}
\end{theorem}

\begin{proof}
    Let $V=\bbF^m$. For $0\leq p\leq k$, let $V_p$ denote the span of the first $p$ columns of $A\cdot T\cdot B$, with the convention $V_0=0$. Since left multiplication only recombines rows, the spaces $V_p$ depend only on the first $p$ columns of $B$, and not on $A$ or $L$.

    We define $\lambda_1,\ldots,\lambda_k$ recursively. First, let $\lambda_1$ be the maximal possible value of $\dim V_1$ as $B_{\cdot,1}$ varies. More generally, once $\lambda_1,\ldots,\lambda_{p-1}$ have been defined, let $\lambda_p$ be the maximal possible value of
    \[ \dim V_p-\dim V_{p-1} \]
    among choices of the first $p$ columns of $B$ satisfying
    \[ \dim V_j=\lambda_1+\cdots+\lambda_j \qquad (1\leq j\leq p-1). \]
    We will prove along the way that
    \[ \lambda_1\geq \lambda_2\geq \cdots \geq \lambda_k\geq 0. \]

    For each $p$, the condition $\dim V_p\geq d$ is the nonvanishing of some $d$-minor, hence is Zariski open in the first $p$ columns of $B$. It follows inductively that there is a nonempty Zariski open subset of $\GL_k(\bbF)$ on which
    \[ \dim V_p=\lambda_1+\cdots+\lambda_p \qquad (1\leq p\leq k). \]
    Fix such a generic choice of $B$.

    We now choose $A$ generically so that, for every $p$, the vectors $(A\cdot T\cdot B)_{1p},\dots,(A\cdot T\cdot B)_{\lambda_p,p}$ project to a basis of $V_p/V_{p-1}$. This is again a nonempty Zariski-open condition, so it holds simultaneously for all $p$ for generic $A$.

    We next construct the lower unitriangular matrix $L$ column by column. For $p=1$, write
    \[ L = \begin{bmatrix}
        L^{(-1)} & 0\\
        L^{(1)} & I_{n - \lambda_1}
    \end{bmatrix}, \]
    where $L^{(-1)}$ is an undetermined $\lambda_1\times\lambda_1$ lower unitriangular matrix, and $L^{(1)}$ is chosen so that $T'_{i1}=0$ for all $i>\lambda_1$. Thus $T'_{11},\dots,T'_{\lambda_1,1}$ form a basis of $V_1$.

    We claim that property \ref{item:span} holds for $p=1$. Choose a basis of $V$ whose first $\lambda_1$ vectors are $e_i=T'_{i1}$. Suppose, for contradiction, that some $T'_{i'l'}$ with $i'>\lambda_1$ does not lie in $\Span\{e_i\}_{i\leq \lambda_1}$. Then it has a nonzero coefficient on some basis vector $e_t$ with $t>\lambda_1$. Consider rows
    \[ R=[\lambda_1] \cup \{i'\}, \]
    columns
    \[ C=[\lambda_1] \cup \{t\}. \]
    The corresponding minors of columns $1$ and $l'$ have the form
    \[ (T'_{\cdot,1})_{R,C} = \begin{bmatrix}
        I_{\lambda_1} \\ & 0
    \end{bmatrix}, \quad (T'_{\cdot,l'})_{R,C} = \begin{bmatrix}
        *_{\lambda_1 \times \lambda_1} & *\\
        * & \beta
    \end{bmatrix}, \]
    where $\beta\neq 0$. Therefore
    \[ (\alpha T'_{\cdot,1} + T'_{\cdot,l'})_{R,C} = \begin{bmatrix}
        \alpha I + * & *\\
        * & \beta
    \end{bmatrix}, \]
    and the coefficient of $\alpha^{\lambda_1}$ in this determinant is $\beta\neq 0$. Hence for some $\alpha\in \bbF$, the column $\alpha T'_{\cdot,1}+T'_{\cdot,l'}$ spans a space of dimension strictly larger than $\lambda_1$. Since replacing the first column of $B$ by $B_{\cdot,1}+\alpha B_{\cdot,l'}$ produces exactly this new first column, this contradicts the maximality of $\lambda_1$.

    We now turn to $p=2$. Since property \ref{item:span} is already known for $p=1$, every entry $T'_{i2}$ with $i>\lambda_1$ lies in $V_1$. Hence every new contribution to $V_2/V_1$ comes from the first $\lambda_1$ rows, and therefore $\lambda_2\leq \lambda_1$. By our generic choice of $B$ and $A$, the vectors $(A\cdot T\cdot B)_{12},\dots,(A\cdot T\cdot B)_{\lambda_2,2}$ project to a basis of $V_2/V_1$.

    We now further decompose $L^{(-1)}$ as
    \[ L^{(-1)} = \begin{bmatrix}
        L^{(-2)} & 0\\
        L^{(2)} & I_{\lambda_1-\lambda_2}
    \end{bmatrix} \]
    where $L^{(-2)}$ is an undetermined $\lambda_2\times\lambda_2$ lower unitriangular matrix, and $L^{(2)}$ is chosen so that $T'_{i2}\in V_1$ for every $\lambda_2<i\leq \lambda_1$. Thus $T'_{12},\dots,T'_{\lambda_2,2}$ form a basis of a complement of $V_1$ in $V_2$.

    We claim that property \ref{item:span} also holds for $p=2$. Choose a basis of $V$ whose first $\lambda_1$ vectors are $e_i=T'_{i1}$ and whose next $\lambda_2$ vectors are $e_{\lambda_1+i}=T'_{i2}$. Suppose, for contradiction, that some $T'_{i'l'}$ with $i'>\lambda_2$ and $l'\geq 2$ does not lie in $\Span\{e_i\}_{i\leq \lambda_1+\lambda_2}$. Since property \ref{item:span} is already known for $p=1$, necessarily $i'\leq \lambda_1$. Then $T'_{i'l'}$ has a nonzero coefficient on some $e_t$ with $t>\lambda_1+\lambda_2$. Consider rows
    \[ R= [\lambda_2] \cup \{i'\}, \]
    columns
    \[ C=\{\lambda_1+1,\dots,\lambda_1+\lambda_2,t\}. \]
    The corresponding submatrices of columns $2$ and $l'$ are
    \[ (T'_{\cdot,2})_{R,C} = \begin{bmatrix}
        I_{\lambda_2} \\ & 0
    \end{bmatrix}, \quad (T'_{\cdot,l'})_{R,C} = \begin{bmatrix}
        *_{\lambda_2\times \lambda_2} & *\\
        * & \beta
    \end{bmatrix}. \]
    The determinant of the corresponding minor of $\alpha T'_{\cdot,2}+T'_{\cdot,l'}$ is a nonzero polynomial in $\alpha$, so for generic $\alpha\in \bbF$ it is nonzero. Hence replacing the second column of $B$ by $B_{\cdot,2}+\alpha B_{\cdot,l'}$ while keeping $B_{\cdot,1}$ fixed produces a new choice for which $\dim V_2-\dim V_1>\lambda_2$, contradicting the maximality of $\lambda_2$.

    The induction step for general $p$ is identical. Assume that we have already arranged the first $p-1$ columns so that:
    \begin{itemize}
        \item $\dim V_j=\lambda_1+\cdots+\lambda_j$ for every $j<p$;
        \item the vectors $T'_{1j},\dots,T'_{\lambda_j,j}$ form a basis of a complement of $V_{j-1}$ in $V_j$;
        \item for every $j<p$ and every $i>\lambda_j$, one has $T'_{ij}\in V_{j-1}$.
    \end{itemize}
    Then every new contribution to $V_p/V_{p-1}$ must come from the first $\lambda_{p-1}$ rows of column $p$, so $\lambda_p\leq \lambda_{p-1}$. By the generic choice of $A$, the first $\lambda_p$ entries of column $p$ project to a basis of $V_p/V_{p-1}$, and by choosing the next block of $L$ we can force $T'_{ip}\in V_{p-1}$ for all $i>\lambda_p$. The same determinant argument as above shows that if some $T'_{i'l'}$ with $i'>\lambda_p$ and $l'\geq p$ did not lie in $V_p$, then replacing column $p$ of $B$ by $B_{\cdot,p}+\alpha B_{\cdot,l'}$ would increase $\dim V_p-\dim V_{p-1}$, a contradiction.

    Iterating up to $p=k$ gives $\dim V_k=\lambda_1+\cdots+\lambda_k$ and
    \[ V_k=\Span T'. \]
    Therefore the vectors $T'_{il}$ with $1\leq l\leq k$ and $1\leq i\leq \lambda_l$ form a basis of $\Span T'$, establishing \ref{item:basis}; and the inductive containment statement is exactly \ref{item:span}.
\end{proof}

\subsection{Compressing semistable tensors}

\begin{theorem} \label{thm:mm_from_semistable}
    Assume $\bbF$ is infinite. Let $T$ be an $n\times m\times k$ semistable tensor with $n\leq m$. Then for every $1\leq p \leq m/n$ we have
    \[ T \geq \Ang{1, \left\lceil \frac{(m-(p-1)n)n}{n+m}\right\rceil , p}. \]
\end{theorem}
In particular, when $m \geq (2p-1)n$, this gives $T \geq \ang{1, \lceil n/2\rceil, p}$; moreover, $T \geq \ang{1, \lceil n/2\rceil, \lfloor(m+n)/(2n)\rfloor}$.
\begin{proof}
    Let us consider $T$ as a representation $\calV$. By \Cref{thm:shift}, we change basis on $U$ and $W$ (i.e., recombine the $k$ different linear maps) such that there exist integers
    \[ \lambda_1 \geq \cdots \geq \lambda_k \geq 0, \]
    and the vectors $\{\calV(a_l)(x_i)\}_{i\leq \lambda_l}$ form a basis of $V$, so we may change the basis on $V$ to let
    \[ \calV(a_1)(x_i) = y_i, \quad \calV(a_2)(x_i) = y_{\lambda_1+i}, \quad \dots,\quad  \calV(a_l)(x_i) = y_{\lambda_1+\cdots+\lambda_{l-1}+i}. \]
    Moreover, property \ref{item:span} guarantees that for $i > \lambda_p$, we have $\calV(x_i) \subset V_p$, where $V_p = \Span \{y_i : i \leq \lambda_1+\cdots+\lambda_p\}$. Consider a restriction of $T$ by zeroing out all $\{x_i : i > \lambda_p\}$ and $\{z_l : l > p\}$, the resulting tensor is
    \[ T \geq \sum_{i=1}^{\lambda_p} \sum_{l=1}^p x_i y_{\lambda_1+\cdots+\lambda_{l-1} + i} z_l, \]
    this can be relabeled to
    \[ \sum_{i=1}^{\lambda_p} \sum_{l=1}^p x_i y'_{il} z_l \cong \ang{1,\lambda_p, p}. \]

    We then use the semistable condition to show that $\lambda_p$ is large when $p$ is small.
    For each $1\leq p\leq m/n$, we let $U^*_p = \Span\{x_i : i > \lambda_p\}$, since $\calV(U^*_p) \subseteq V_p$, we have
    \[ \frac{m}{n}(n-\lambda_p) = \mu^\star(\calV)^{-1} \dim(U^*_p) \stackrel{\text{Prop.~\ref{prop:unstable}}}{\leq} \dim \calV(U^*_p) \leq \dim (V_p) = \lambda_1 + \cdots + \lambda_p. \]
    We then substitute the trivial upper bound $\lambda_l \leq n$ for all $l < p$, this gives
    \[ \lambda_p \geq \frac{(m-(p-1)n)n}{n+m}. \qedhere \]
\end{proof}

\section{Lower value functional} \label{sec:value}

In this section, we introduce the notion of the lower value functional. Throughout this section, whenever we write $\gamma$, we mean a vector $\gamma=(\gamma_1,\gamma_2,\gamma_3) \in \bbR_{\geq 0}^3$ with at least one nonzero coordinate. We then define $\xi_\gamma$ by
\[ \xi_\gamma (T) = \sup \left\{ (E^{\gamma_1} H^{\gamma_2} L^{\gamma_3})^{1/N} : T^{\otimes N} \geq \ang{E,H,L} \right\}. \]

The name originates from the study of fast matrix multiplication. Starting from Coppersmith and Winograd \cite[\S8]{CW90Laser}, the notion of \emph{value}
\[ V_\tau(T) = \sup \left\{ \left(\sum_{k=1}^K (E_k H_k L_k)^{\tau}\right)^{1/N} : T^{\otimes N} \neged \bigoplus_{k=1}^K \ang{E_k, H_k, L_k} \right\} \]
has been crucial in the analysis of the laser method. Both this and $\xi_\gamma$ try to capture the ``most valuable'' matrix multiplication tensor that one can asymptotically obtain from $T$. The difference is that our definition only allows restriction to a single matrix multiplication tensor, not a direct sum. Even when $\gamma=(\tau,\tau,\tau)$, it is not clear \emph{a priori} whether the two notions coincide. Our reason for introducing $\xi_\gamma$ is that it admits a particularly convenient convex-geometric formulation, see \Cref{sec:value-body}.

The following properties are immediate from the definition.
\begin{fact}
    For any $\gamma$, the lower value functional $\xi_\gamma$ is
    \begin{enumerate}
        \item monotone,
        \item supermultiplicative,
        \item moreover, for any $T$ and $n$, $\xi_\gamma(T^{\otimes n}) = \xi_\gamma(T)^n$, and
        \item for any $\lambda > 0$, $\xi_{\lambda\gamma}(T) = \xi_{\gamma}(T)^\lambda$.
    \end{enumerate}
\end{fact}

\begin{proposition}
    If $\xi_\gamma(\ang{n}) \geq n$, then $\xi_\gamma$ is superadditive.
\end{proposition}
\begin{proof}
    For any tensors $T,S$, the binomial expansion gives
    \[ (T\oplus S)^{\otimes n} = \bigoplus_{k=0}^n \binom{n}{k} \odot T^{\otimes k} S^{\otimes (n-k)}
    \geq \binom{n}{k} \odot T^{\otimes k} S^{\otimes (n-k)}, \]
    so $(T\oplus S)^{\otimes n}$ restricts to each monomial term. By monotonicity, supermultiplicativity, and the assumption on $\xi_\gamma(\ang{n})$, we obtain
    \[ \xi_\gamma (T\oplus S)^n \geq \binom{n}{k}
    \xi_\gamma(T)^k \xi_\gamma(S)^{n-k}, \]
    for every $0\leq k\leq n$. Summing over $k$, we get
    \[ (n+1)\xi_\gamma(T\oplus S)^n \geq (\xi_\gamma(T) + \xi_\gamma(S))^n. \]
    Taking $n\to\infty$, we conclude that
    \[ \xi_\gamma(T\oplus S) \geq \xi_\gamma(T) + \xi_\gamma(S). \qedhere \]
\end{proof}

\subsection{Feasible region} \label{sec:value-body}

For any nonzero tensor $T$, we define the \emph{feasible region} $\Pi(T)$ of $T$ by
\[ \Pi(T) \coloneq \overline{\left\{ \left(\frac{h_1}{n}, \frac{h_2}{n}, \frac{h_3}{n}\right) : T^{\otimes n} \geq \ang{2^{h_1}, 2^{h_2}, 2^{h_3}} \right\}}, \]
or equivalently
\[ \Pi(T) = \overline{\left\{ \left(\frac{\lg E}{n}, \frac{\lg H}{n}, \frac{\lg L}{n}\right) : T^{\otimes n} \geq \ang{E, H, L} \right\}}. \]

With this definition, we immediately get the following facts.
\begin{fact}
    For any nonzero $T$,
    \begin{enumerate}
        \item The feasible region $\Pi(T)$ is a convex compact subset of $\bbR_{\geq 0}^3$, and contains $0$.
        \item The lower value functional can be written as 
        \[ \xi_{\gamma}(T) = \max_{\eta \in \Pi(T)} 2^{\sum_{\varkappa=1}^3 \gamma_\varkappa \eta_\varkappa}. \]
    \end{enumerate}
\end{fact}

It causes no harm to replace restriction by asymptotic restriction:
\begin{lemma} \label{lem:feasible_asymp}
    For any nonzero tensor $T$,
    \[ \Pi(T) = \overline{\left\{ \left(\frac{\lg E}{n}, \frac{\lg H}{n}, \frac{\lg L}{n}\right) : T^{\otimes n} \gtrsim \ang{E, H, L} \right\}}. \]
\end{lemma}
\begin{proof}
    Let $\Pi'(T)$ denote the set defined using asymptotic restriction. It is immediate that $\Pi(T) \subseteq \Pi'(T)$. To show the reverse inclusion, it suffices to prove that if $\ang{E,H,L} \lesssim T$, then $(\lg E,\lg H,\lg L) \in \Pi(T)$. Indeed, by Strassen duality (\Cref{thm:duality_tensor}), we have $\ang{E^n,H^n,L^n} \leq T^{\otimes n+o(n)}$. By the definition of $\Pi(T)$ this gives
    \[ \left(\frac{n\lg E}{n+o(n)}, \frac{n\lg H}{n+o(n)}, \frac{n\lg L}{n+o(n)}\right) \in \Pi(T), \]
    since $\Pi(T)$ is closed, taking $n\to\infty$ proves the claim.
\end{proof}

Again, Strassen duality implies the following dual characterization:
\begin{proposition}
    For any nonzero tensor $T$ and $\eta\in \bbR_{\geq 0}^3$, $\eta \in \Pi(T)$ if and only if
    \begin{center}
        $\displaystyle\sum_{\varkappa=1}^3 \gamma_\varkappa(\phi) \eta_\varkappa \leq \lg \phi(T)$ \quad for all \quad $\phi\in\calX$.
    \end{center}
    In other words, $\Pi(T)$ is cut out by hyperplanes
    \[ \Pi(T) = \bbR_{\geq 0}^3 \cap \bigcap_{\phi\in \calX} \left\{ \eta : \sum_{\varkappa=1}^3 \gamma_\varkappa(\phi) \eta_\varkappa \leq \log_2 \phi(T) \right\}. \]
\end{proposition}
\begin{proof}
    If $\eta \in \Pi(T)$, then there exists a sequence of $(h_1,h_2,h_3,n)$ such that $T^{\otimes n} \geq \ang{2^{h_1}, 2^{h_2}, 2^{h_3}}$, and $(h_1,h_2,h_3)/n \to \eta$. Plugging in any $\phi\in \calX$, we have $\phi(T)^n \geq \phi(\ang{2^{h_1}, 2^{h_2}, 2^{h_3}})$. By the definition of $\gamma(\phi)$, we have $\sum_{\varkappa=1}^3 \gamma_\varkappa(\phi) h_\varkappa \leq n\lg \phi(T)$. Taking the limit $(h_1,h_2,h_3)/n \to \eta$, we obtain $\sum_{\varkappa=1}^3 \gamma_\varkappa(\phi) \eta_\varkappa \leq \lg \phi(T)$.

    Conversely, if $\eta$ satisfies all the inequalities, we consider a sequence of rational approximations $p_\varkappa = h_\varkappa / n \in [0, \eta_\varkappa]$ such that $(p_1,p_2,p_3) \to \eta$. Then we must have $\sum_{\varkappa=1}^3 \gamma_\varkappa(\phi) h_\varkappa \leq n \sum_{\varkappa=1}^3 \gamma_\varkappa(\phi) \eta_\varkappa \leq n \lg \phi(T)$, that is,
    \[ \phi(\ang{2^{h_1}, 2^{h_2}, 2^{h_3}}) \leq \phi(T)^n. \]
    By Strassen's duality, this means $\ang{2^{h_1}, 2^{h_2}, 2^{h_3}} \lesssim T^{\otimes n}$, thus $(p_1,p_2,p_3) \in \Pi(T)$ by \Cref{lem:feasible_asymp}. Taking the limit, we conclude $\eta\in \Pi(T)$.
\end{proof}

Exponentiating both sides yields the following corollary.
\begin{corollary} \label{cor:value_leq_func}
    For any $\phi\in \calX$, we always have $\xi_{\gamma(\phi)} \leq \phi$.
\end{corollary}

When $\gamma = (\rho, 1, 1-\rho)$, we have
\[ \xi_\gamma(\ang{n}) \geq \xi_\gamma(\ang{1,n,1}) \geq n, \]
so $\xi_\gamma$ is superadditive in this case, and therefore has all the formal properties one expects from a lower functional. In \Cref{sec:final_strike}, we will show that
\[ \xi_{(\rho,1,1-\rho)} = \zeta^{(\rho,1-\rho,0)}, \]
and hence that $\xi_\gamma$ is in fact a spectral point on this edge.

Could one hope to extend this argument and show that some other $\xi_\gamma$ is a spectral point as well? This appears to be substantially more difficult. Indeed, \Cref{cor:value_leq_func} implies that whenever $\gamma\in\specMM$, one has
\[ \xi_\gamma(\ang{n}) \leq n. \]
On the other hand, we also have the scaling identity
\[ \xi_{\lambda\gamma}(T)=\xi_\gamma(T)^\lambda \qquad (\lambda>0). \]
Combining these two observations, we see that if $\xi_\gamma$ were itself a spectral point, then $\gamma$ would have to be \emph{radially extremal} in $\specMM$: whenever $\lambda\gamma\in\specMM$, one must have $\lambda\leq 1$.

This already suggests that away from the edges the situation is substantially more delicate. Suppose, for instance, that one wanted to prove an identity of the form
\[ \zeta^\theta = \xi_\gamma, \qquad \gamma=(\theta_3+\theta_1,\theta_1+\theta_2,\theta_2+\theta_3). \]
Then the discussion above would force $\lambda\gamma\notin\specMM$ for every $\lambda>1$. In the interior case $\theta_1,\theta_2,\theta_3>0$, this would in fact imply $\omega=2$: if $\omega>2$, then $\specMM$ strictly contains the base triangle $\triangle$, and star-convexity would produce a point of the form $(1+\varepsilon)\gamma$ with $\varepsilon>0$, contradicting radial extremality.

There are two ways to read this. Optimistically, if future work eventually proves $\omega=2$, then the lower value functionals introduced here may become useful tools for further exhausting the asymptotic spectrum. Pessimistically, it may mean that away from the boundary the functionals $\xi_\gamma$ are not the correct scaffolding to work on.

\section{Main theorem} \label{sec:main}

In this section we will frequently use the notion of \emph{types}. For a finite index set $I$, consider its $k$-fold Cartesian product $I^{\times k}$, where each element $\ov{i}\in I^{\times k}$ can be written as a sequence $\ov{i} = (i_1,\dots,i_k)$ with $i_l \in I$ for all $l$. For any fixed $\ov{i}$, the type $\lambda = \lambda(\ov{i}) \in \bbN^I$ of $\ov{i}$ records the number of occurrences of each element $a\in I$ in $\ov{i}$: $\lambda_a = |\{l : i_l = a\}|$. Note that $\lambda / k \in \calP(I)$ is naturally a distribution on $I$, so we write $P(\lambda) = \lambda/k$.

Say $I = [n]$, and fix a particular type $\lambda$ with sum $|\lambda| = k$. The number of sequences $\ov{i} \in I^{\times k}$ having type $\lambda$ can be written as a multinomial coefficient, which we bound via entropy:
\begin{equation}\label{multibound}
|\{\ov{i} \in I^{\times k} : \lambda(\ov{i}) = \lambda \}| = \binom{k}{\lambda_1, \dots, \lambda_{n}} \leq 2^{\eH(P(\lambda)) k}.
\end{equation}

\subsection{Edge of matrix multiplication spectrum}

We now determine the edge of $\specMM$ cut out by the condition $\gamma_2=1$. From \Cref{fact:reg_MM}, one only gets the coarse bounds
\[ 0\leq \gamma_1,\gamma_3\leq 1, \qquad 1\leq \gamma_1+\gamma_3\leq \omega-1. \]
The next proposition shows that, on this edge, there is in fact no room beyond the line segment already predicted by Strassen's star-convexity theorem.

\begin{proposition} \label{prop:MM_edge}
    For any $\phi\in\calX$, if $\gamma_2(\phi)=1$, then $\gamma_1(\phi)+\gamma_3(\phi)=1$. Equivalently,
    \[ \gamma(\calZ_{(3)}) = \{ (\rho,1,1-\rho) : \rho\in[0,1] \}. \]
\end{proposition}

The key ingredient is \emph{Sch\"onhage's direct sum identity}.
\begin{lemma}[Sch\"onhage {\cite[Lemma~6.1]{Sch81Tau}}]
    For any $n,m\geq 2$, one has the degeneration
    \[ \ang{n,1,m} \oplus \ang{1,(n-1)(m-1),1} \degen \ang{nm + 1}. \]
\end{lemma}

\begin{proof}[Proof of \Cref{prop:MM_edge}]
    Since $\specMM\subseteq [0,1]^3$ and every point of $\specMM$ satisfies $\gamma_1+\gamma_2+\gamma_3\geq 2$, the assumption $\gamma_2=1$ already implies
    \[ \gamma_1+\gamma_3 \geq 1. \]

    For the reverse inequality, apply Sch\"onhage's identity with $m=n$:
    \[ \ang{n,1,n} \oplus \ang{1,(n-1)^2,1} \degen \ang{n^2+1}, \]
    and evaluate any $\phi\in\calX$ with $\gamma_2(\phi)=1$. Since degeneration induces the same asymptotic preorder as restriction by \Cref{thm:duality_tensor}, spectral points are monotone for this relation as well. Therefore
    \[ n^{\gamma_1 + \gamma_3} + (n-1)^2 \leq n^2+1, \]
    hence $n^{\gamma_1+\gamma_3}\leq 2n$ for every $n\geq 2$. Letting $n\to\infty$ gives $\gamma_1+\gamma_3\leq 1$, and therefore $\gamma_1+\gamma_3=1$.

    It remains to prove the converse inclusion. The two endpoints $(1,1,0)$ and $(0,1,1)$ are gauge points, hence belong to $\specMM$. By star-convexity (\Cref{thm:star-convex}), the whole line segment between them also lies in $\specMM$, namely
    \[ \{(\rho,1,1-\rho): \rho\in[0,1]\} \subseteq \specMM. \]
    This proves the claimed description of $\gamma(\calZ_{(3)})$.
\end{proof}

The proposition is closely related to rectangular matrix multiplication exponents. For $a,b,c>0$, define
\[ \omega(a,b,c) = \inf \left\{ \tau : \Rk(\ang{\ceil{n^a}, \ceil{n^b}, \ceil{n^c}}) = O(n^\tau) \right\}. \]
By Strassen duality, it is not hard to verify that
\[ \omega(a,b,c) = \max_{\gamma\in\specMM} (a\gamma_1 + b\gamma_2 + c\gamma_3). \]
We record \cite[Proposition~4.5]{Str88AsymSpec} for completeness. In particular, \Cref{prop:MM_edge} is equivalent to the assertion that, among all $\gamma\in\specMM$ with $\gamma_2=1$, the quantity $\gamma_1+\gamma_3$ has supremum $1$. By the displayed formula with $(a,b,c)=(1,k,1)$, this is in turn equivalent to the asymptotic statement
\[ \lim_{k\to\infty } (\omega(1, k, 1) - k - 1) = 0. \]

This consequence of Sch\"onhage's identity was already noticed by Lotti and Romani \cite{LR83Rect}, who proved the quantitative bound
\[ \omega(1,k,1) = k + 1 + O\left(\frac 1{\log k}\right). \]
Later work based on the laser method improved rectangular exponents in many regimes \cite{Cop97Rect,HP98Rect,LeGall12Rect,LU18Rect,LeGall24Rect,VXXZ24Rect,ADWXXZ25MoreAsym}, but did not improve this asymptotic order $O(1/\log k)$.

\subsection{Bounding every functional by the support functional}

\begin{lemma}\label{prop:support_proj}
    For every tensor $T$, we have
    \[\zeta^{(\rho,1-\rho, 0)}(T) = \min_{g \in\GL(U) \times \GL(V)} \max_{P \in \calP(\supp_{U\times V} g \cdot T)} 2^{\rho \eH(P_1) + (1-\rho) \eH(P_2)},\]
    where 
    \[\supp_{U\times V}(g\cdot T) = \pi_{12}(\supp(g \cdot T)) \subseteq [n] \times [m]\] is the projection of $\supp(g\cdot T) \subseteq [n] \times [m] \times [k]$ onto the first two coordinates.
\end{lemma}
\begin{proof}
	By definition,
	\[ \zeta^\theta(T) = \min_{g \in\GL(U) \times \GL(V) \times \GL(W)} \max_{P \in \calP(\supp(g\cdot T))} 2^{\rho \eH(P_1) + (1-\rho) \eH(P_2)}.\]
	Fix $g$; replacing $T$ by $g\cdot T$, we may assume that $g$ is the identity. The inner objective only depends on the first two marginals of $P$. Hence maximizing over distributions on $\supp(T)$ is equivalent to maximizing over distributions on the projected support. Indeed, note that if $P$ is a distribution on $\supp(T)$, then its pushforward under $\pi_{12}$ has the same marginals as $P$. Conversely, if $P'$ is a distribution on $\supp_{U\times V} T \subseteq [n] \times [m]$, then the distribution $P$ on $\supp(T)$ given by choosing for each $(i,j) \in \supp_{U\times V} T$ an index $k_{ij}$ such that $(i,j,k_{ij}) \in \supp(T)$, and setting
	\[
	P_{ijk} =
	\begin{cases}
		P'_{ij}, & \text{if } k = k_{ij},\\
		0, & \text{otherwise}
	\end{cases}
	\]
	has the same first two marginals.
	
	Finally, the projected support is invariant under the action of any $g \in \GL(W)$ in the third mode, since $T_{i,j,\cdot} = 0$ (as a vector in $W$) if and only if $g \cdot T_{i,j,\cdot} = 0$.
\end{proof}

\begin{proposition} \label{prop:bdd_by_supp}
    If $\phi\in \calX$ satisfies $\gamma_2(\phi) = 1$, say
    $\gamma(\phi) = (\rho, 1, 1-\rho)$, then 
    \[ \phi(T) \leq \zeta^{\theta}(T) \]
    for all $T$, where $\theta = (\rho, 1-\rho, 0)$.
\end{proposition}

\begin{proof} 
    By \Cref{prop:support_proj}, there exists a choice of basis of $T$ such that, with $\Phi = \supp_{U\times V} T$,
	\[ \zeta^\theta(T) = \max_{P\in \calP(\Phi)} 2^{\rho \eH(P_1) + (1-\rho) \eH(P_2)}.\]
	
    For a type $\lambda$ of size $n$ over $\Phi$, let $T_{\lambda}$ be the projection of $T^{\otimes n}$ onto its monomials $\ov{x}_{\ov{i}}\ov{y}_{\ov{j}}$ of type $\lambda$. Then 
	\[T^{\otimes n} = \sum_{|\lambda| = n} T_\lambda \le \bigoplus_{|\lambda| = n} T_\lambda, \]
    so by subadditivity of $\phi$,
    \[\phi(T^{\otimes n}) \leq \sum_{|\lambda| = n} \phi(T_\lambda).\]
	
    Let $a_\lambda$ and $b_\lambda$ be the number of basis vectors in the first and second factors appearing in $T_\lambda$. By \cref{multibound}, 
    \[a_\lambda \le 2^{\eH(P(\lambda)_1) n}, \quad b_\lambda \le 2^{\eH(P(\lambda)_2) n}.\]
    Any tensor with first and second dimensions $a_\lambda$ and $b_\lambda$ is a restriction of $\ang{a_\lambda, 1, b_\lambda}$, as is evident from the representation
    \[ T_\lambda = \sum_{\ov{i},\ov{j}} x_{\ov{i}} y_{\ov{j}} \left(\sum_{k} u_{\ov{i}\ov{j}k} z_k \right) \leq  \sum_{\ov{i},\ov{j}} x_{\ov{i}} y_{\ov{j}} z'_{\ov{i}\ov{j}} \cong \ang{a_\lambda, 1, b_\lambda}. \]
    Hence by monotonicity,
    \[\phi(T_\lambda) \le \phi(\ang{ a_\lambda, 1, b_\lambda}) = a_\lambda^\rho b_\lambda^{1-\rho} \le 2^{(\rho \eH(P(\lambda)_1) + (1-\rho)\eH(P(\lambda)_2))n}.\]
	Summing over all $\lambda$ and using the definition of $\zeta^\theta(T)$ we find that
	\[\phi(T^{\otimes n}) \le \sum_\lambda 2^{(\rho \eH(P(\lambda)_1) + (1-\rho)\eH(P(\lambda)_2))n} \le n^{O(1)} \zeta^\theta(T)^n \]
	as the number of types is at most $(n+1)^{(\dim U)(\dim V)}$. Taking $n$th roots and letting $n \to \infty$, we conclude that $\phi(T) \leq \zeta^\theta(T)$.
\end{proof}

\subsection{Extracting matrix multiplication from the Harder--Narasimhan filtration} \label{sec:final_strike}

Let us first summarize the results established so far. For any spectral point $\phi$ satisfying $\gamma_2(\phi) = 1$, \Cref{prop:MM_edge} gives $\gamma(\phi) = (\rho, 1, 1-\rho)$ for some $\rho \in [0, 1]$. Let $\theta = (\rho,1-\rho, 0)$. Then we have
\begin{itemize}
    \item By \Cref{cor:value_leq_func} we have the lower bound $\xi_\gamma \leq \phi$.
    \item By \Cref{prop:bdd_by_supp} we have the upper bound $\phi \leq \zeta^\theta$.
    \item By \Cref{thm:star-convex}, at least one such $\phi$ exists.
    \item Without loss of generality, we may assume that $T$ is concise. Viewing $T$ as a representation $\calV$, it is then also concise. Let
    \[ (U_0, V_0) \subsetneq (U_1, V_1) \subsetneq \cdots \subsetneq (U_r, V_r) = (U, V) \]
    be the HN-filtration of $\calV$. We choose bases of $T$ as follows. Let the first $n_1$ vectors $x_i$ form a basis of $U_1$, then choose $n_2$ further basis vectors extending $U_1$ to $U_2$. In general, we take the vectors $\{x_{n_1+\cdots+n_{u-1}+i}\}_{1\leq i\leq n_u}$ to be a basis of a complement extending $U_{u-1}$ to $U_u$. Similarly, we choose a basis $\{y_i\}$ adapted to the filtration $\{V_u\}$ of $V$.
    Letting $\Phi=\supp_{U\times V}(T)$ be the projected support of $T$ in these bases, we have $\zeta^\theta(T) \leq 2^{\eH_\theta(\Phi)}$ by the definition of the upper support functional.
\end{itemize}

Thus we have
\[ \xi_\gamma(T) \leq \phi(T) \leq \zeta^\theta(T) \leq 2^{\eH_\theta(\Phi)}, \]
and we now prove the reverse inequality.

\begin{theorem}\label{thm:value_geq_hn}
    Under the above assumptions, we have $\xi_\gamma(T) \geq 2^{\eH_\theta(\Phi)}$.
\end{theorem}

Before proving the main result, we need a preparatory lemma.

\begin{lemma} \label{lem:semistable_iff_balanced}
    A representation $\calV$ is semistable if and only if its support in every choice
    of bases is balanced.
\end{lemma}

\begin{proof}
    Let $\calV$ be a representation on vector spaces $U,V$ with
    $\dim U=n$ and $\dim V=m$. Let $J=[n]$ and $K=[m]$. For bases $x_1,\dots,x_n$ of $U$ and
    $y_1,\dots,y_m$ of $V$, let
    \[
        \Phi(\calV)\subseteq [n]\times [m]
    \]
    be the support of $\calV$ in these bases, i.e.,
    $(i,j)\in\Phi(\calV)$ if $\calV(a_\ell)(x_i)$ has nonzero $y_j$-coefficient
    for some arrow $a_\ell$. By \Cref{prop:balance_char}, $\Phi(\calV)$ is balanced if and only if for every $J'\subseteq J$,
    \[
        |N(J')| \ge \frac{|K|}{|J|}|J'| = \frac{m}{n}|J'|.
    \]

    We first prove the contrapositive of the forward implication.
    Suppose that $\Phi(\calV)$ is not balanced. Then there exists a subset
    $J'\subseteq J$ such that
    \[
        |N(J')| < \frac{m}{n}|J'|.
    \]
    Let
    \[
        U^*=\Span\{x_i : i\in J'\}\subseteq U,
        \qquad
        V^*=\Span\{y_j : j\in N(J')\}\subseteq V.
    \]
    By definition of $N(J')$, every arrow map sends $U^*$ into $V^*$, so
    $\calV(U^*)\subseteq V^*$. Therefore
    \[
        \dim \calV(U^*) \le \dim V^* = |N(J')|
        < \frac{m}{n}|J'|
        = \frac{m}{n}\dim U^*.
    \]
    By \Cref{prop:unstable}, this means that $\calV$ is unstable. Hence, if
    $\calV$ is semistable, then $\Phi(\calV)$ must be balanced.

    Conversely, suppose that $\calV$ is unstable. By \Cref{prop:unstable},
    there exists a nonzero subspace $U^*\subseteq U$ such that
    \[
        \dim \calV(U^*) < \frac{m}{n}\dim U^*.
    \]
    Write $a=\dim U^*$ and $b=\dim \calV(U^*)$. Choose a basis
    $x_1,\dots,x_n$ of $U$ such that
    \[
        U^*=\Span\{x_1,\dots,x_a\},
    \]
    and choose a basis $y_1,\dots,y_m$ of $V$ such that
    \[
        \calV(U^*)=\Span\{y_1,\dots,y_b\}.
    \]
    In these bases, no support edge can go from a row $i\le a$ to a column
    $j>b$, because every image of $U^*$ lies in $\calV(U^*)$.
    Thus for
    \[
        J'=\{1,\dots,a\}\subseteq J
    \]
    we have
    \[
        N(J')\subseteq \{1,\dots,b\},
    \]
    and so
    \[ |N(J')|\le b < \frac{m}{n}a = \frac{m}{n}|J'|. \qedhere \]
\end{proof}

\begin{proof}[Proof of \Cref{thm:value_geq_hn}]
    We denote the basis vectors extending $U_{u-1}$ to $U_u$ and $V_{u-1}$ to $V_u$ by the variable blocks
    \begin{center}
        $J_u = \{n_1+\cdots+n_{u-1}+1, \dots, n_1+\cdots+n_u\}$
        
        and
        
        $K_u = \{m_1+\cdots+m_{u-1}+1, \dots, m_1+\cdots+m_u\}$
    \end{center}
    respectively. The support $\Phi$ then fits exactly into the hypotheses of \Cref{prop:argmax_2d}. Since each $\calV_u / \calV_{u-1}$ is semistable, we have that $\Phi^{(u)}=\Phi \cap (J_u\times K_u)$ is balanced. Since $\calV(U_u) \subset V_u$, we have $N(J_1\cup \cdots\cup J_u) \subset K_1\cup \cdots \cup K_u$, i.e., $\Phi$ is blockwise triangular. Finally, the slopes in the HN-filtration imply the convexity condition. We conclude that
    \[ 2^{\eH_\theta(\Phi)} = \sum_{u=1}^r n_u^\rho m_u^{1-\rho} = \sum_{u=1}^r m_u \mu_u^\rho. \]

    We now prove that $\xi_\gamma(T)$ is at least this value. Consider the tensor power $T^{\otimes n}$; viewed as a representation, this is just $\calV^{\otimes n}$. By an iterative application of \Cref{prop:tensor_hn}, every slope $\nu$ in the HN-filtration of $\calV^{\otimes n}$ is a product $\nu=\mu_1^{\lambda_1} \cdots \mu_r^{\lambda_r}$, where $|\lambda| = n$. Let $T_\nu$ be the tensor corresponding to this subquotient, and let $N_\nu, M_\nu$ denote its dimensions in the first two modes. If $N_\nu \leq M_\nu$, then since $T_\nu$ is semistable, we have
    \[ T_\nu \gtrsim \Ang{1, \Ceil{\frac {N_\nu} 2}, \Floor{\frac{M_\nu+N_\nu}{2N_\nu}} }, \]
    because \Cref{thm:mm_from_semistable} gives the restriction $\geq$ after passing to an infinite extension of $\bbF$, and the resulting asymptotic restriction descends to the ground field.
    
    Therefore
    \[ \xi_\gamma(T^{\otimes n}) \geq \xi_\gamma(T_\nu) \geq \frac 1 8 \cdot N_\nu \cdot (M_\nu/N_\nu)^{1-\rho} = \frac 1 8 N_\nu^{\rho} M_\nu^{1-\rho}. \]
    Similarly, when $N_\nu \geq M_\nu$, we can switch the roles of $U$ and $V$, and in either case we conclude that $\xi_\gamma(T^{\otimes n}) \geq \frac 1 8 N_\nu^{\rho} M_\nu^{1-\rho}$.

    The number of distinct slopes is bounded by the number of distinct types, so summing over all slopes gives
    \begin{align*}
        n^{O(1)} \xi_\gamma(T^{\otimes n}) &\geq \sum_{\nu} \frac 1 8 N_\nu^\rho M_\nu^{1-\rho}\\
        &= \frac 1 8 \sum_{\nu} M_\nu \nu^\rho \\
        &= \frac 1 8 \sum_\nu \left(\sum_{\substack{\ov{u} \in [r]^{\times n} \\ \mu_{u_1}\cdots\mu_{u_n} = \nu}} m_{u_1}\cdots m_{u_n} \right) \nu^{\rho}\\
        &= \frac 1 8 \sum_{\ov{u} \in [r]^{\times n}} (m_{u_1}\cdots m_{u_n}) \cdot (\mu_{u_1}\cdots \mu_{u_n})^\rho\\
        &= \frac 1 8 \left(\sum_{u=1}^r m_u \mu_u^\rho\right)^n = \frac 1 8 2^{\eH_\theta(\Phi) n}.
    \end{align*}
    Since $\xi_\gamma(T)^n = \xi_\gamma(T^{\otimes n}) \geq n^{-O(1)} 2^{\eH_\theta(\Phi) n}$, taking $n\to\infty$ proves the claim.
\end{proof}

\begin{corollary} \label{cor:edge-universal}
    Let $\rho\in[0,1]$, $\theta=(\rho,1-\rho,0)$, and $\gamma=(\rho,1,1-\rho)$. Then for every tensor $T$,
    \[ \xi_\gamma(T) = \zeta^\theta(T). \]
    In particular, $\zeta^\theta$ is a universal spectral point, and every $\phi\in\calX$ with $\gamma(\phi)=\gamma$ satisfies
    \[ \phi = \zeta^\theta. \]
\end{corollary}
\begin{proof}
    By \Cref{thm:star-convex}, there exists $\phi\in\calX$ with $\gamma(\phi)=\gamma$. For this $\phi$, the discussion preceding \Cref{thm:value_geq_hn} gives
    \[ \xi_\gamma(T) \leq \phi(T) \leq \zeta^\theta(T) \leq 2^{\eH_\theta(\Phi)}, \]
    while \Cref{thm:value_geq_hn} gives the reverse inequality $\xi_\gamma(T)\geq 2^{\eH_\theta(\Phi)}$. Hence all inequalities are equalities, so $\xi_\gamma(T)=\phi(T)=\zeta^\theta(T)$ for every tensor $T$. Since $\phi$ is a spectral point, so is $\zeta^\theta$. The same chain of equalities applies to any other $\phi\in\calX$ with $\gamma(\phi)=\gamma$, proving uniqueness.
\end{proof}

\begin{corollary} \label{cor:edge-hn-formula}
    Let $T$ be a tensor, and let $\{(n_u,m_u)\}_{1\leq u\leq r}$ be the dimension data of the HN-filtration of the associated Kronecker-quiver representation. Then for every $\rho\in[0,1]$,
    \[ \zeta^{(\rho,1-\rho,0)}(T) = \sum_{u=1}^r n_u^\rho m_u^{1-\rho}. \]
\end{corollary}
\begin{proof}
    In the proof of \Cref{thm:value_geq_hn}, we showed that
    \[ 2^{\eH_\theta(\Phi)} = \sum_{u=1}^r n_u^\rho m_u^{1-\rho}. \]
    Combining this with \Cref{cor:edge-universal}, which identifies $\zeta^\theta(T)$ with $2^{\eH_\theta(\Phi)}$, proves the claim.
\end{proof}

\begin{corollary} \label{cor:edge-mm-char}
    Let $\theta=(\rho,1-\rho,0)\in \Theta_{(3)}$. Then for every tensor $T$,
    \[ \zeta^\theta(T) = \sup \left\{ \zeta^\theta(\ang{E,H,L})^{1/n} : T^{\otimes n} \geq \ang{E,H,L} \right\}. \]
\end{corollary}
\begin{proof}
    By \Cref{cor:edge-universal}, we have $\zeta^\theta = \xi_\gamma$ for $\gamma=(\rho,1,1-\rho)$. By definition of $\xi_\gamma$,
    \[ \xi_\gamma(T) = \sup \left\{ (E^\rho H L^{1-\rho})^{1/n} : T^{\otimes n} \geq \ang{E,H,L} \right\}. \]
    Since $\zeta^\theta(\ang{E,H,L}) = E^\rho H L^{1-\rho}$, the result follows.
\end{proof}

\begin{corollary} \label{cor:edge-unique}
    If $\phi\in\calX$ satisfies $\phi(\ang{1,n,1}) = n$ for all $n\geq 1$, then there exists $\rho\in[0,1]$ such that
    \[ \phi = \zeta^{(\rho,1-\rho,0)}. \]
\end{corollary}
\begin{proof}
    The assumption says exactly that $\gamma_2(\phi)=1$. By \Cref{prop:MM_edge}, we then have $\gamma(\phi)=(\rho,1,1-\rho)$ for some $\rho\in[0,1]$. Now apply the uniqueness statement in \Cref{cor:edge-universal}.
\end{proof}

\subsection{Asymptotic (non-)commutative rank}

\begin{theorem} \label{thm:acr_formula}
    Over any field $\bbF$, the asymptotic (non-)commutative rank is given by
    \[ \ACR(T) = \min_{\rho\in[0, 1]} \zeta^{(\rho,1-\rho, 0)}(T), \]
    more precisely, let $\{(n_i, m_i)\}_{1\leq i\leq r}$ be the dimension data of the HN-filtration of $T$. Then
    \[ \ACR(T) = \min_{\rho\in[0, 1]} \sum_{i=1}^r n_i^\rho m_i^{1-\rho}. \]
\end{theorem}
\begin{proof}
    Since asymptotic non-commutative rank agrees with asymptotic commutative rank, it suffices to prove the formula for $\ACR$.

    By \Cref{cor:acr-minimize-edge}, we have
    \[ \ACR(T) = \min_{\phi : \gamma_2(\phi)=1} \phi(T). \]
    By \Cref{cor:edge-unique}, the spectral points with $\gamma_2=1$ are exactly the edge support functionals
    \[ \zeta^{(\rho,1-\rho,0)}, \qquad \rho\in[0,1]. \]
    Therefore
    \[ \ACR(T) = \min_{\rho\in[0,1]} \zeta^{(\rho,1-\rho,0)}(T). \]
    The rest follows from \Cref{cor:edge-hn-formula}. \qedhere
\end{proof}

The formula above is a consequence of the edge-support-functional theory developed in this paper. On the other hand, there is also a much more direct characterization of $\ACR$, obtained only from Strassen duality together with the restriction-theoretic description of commutative rank.

\begin{proposition} \label{prop:acr_func_pow}
    We have
    \[ \ACR(T) = \inf \left\{ \phi(T)^{1/\gamma_2(\phi)} : \phi \in \calX,\ \gamma_2(\phi)>0 \right\}. \]
\end{proposition}
\begin{proof}
    By invariance of commutative rank under field extension, we may replace $\bbF$ by an infinite extension field. Then \Cref{prop:cr-restr} implies that for every $n$,
    \[ \CR(T^{\otimes n}) = \max\{ q\in \bbN : \ang{1,q,1} \leq T^{\otimes n} \}. \]
    Therefore, for any real number $a\geq 1$, we have
    \[
        a \leq \ACR(T)
        \iff \ang{1,a^{n-o(n)},1} \leq T^{\otimes n}.
    \]
    By Strassen duality, this is equivalent to
    \[
        \forall \phi\in\calX,\qquad
        \phi(\ang{1,a^{n-o(n)},1}) \leq \phi(T)^n.
    \]
    Since
    \[
        \phi(\ang{1,a^{n-o(n)},1}) = a^{\gamma_2(\phi)n-o(n)},
    \]
    the condition above is equivalent to
    \[
        \forall \phi\in\calX \text{ with } \gamma_2(\phi)>0,\qquad
        a \leq \phi(T)^{1/\gamma_2(\phi)}.
    \]
    Equivalently,
    \[
        a \leq \inf \left\{ \phi(T)^{1/\gamma_2(\phi)} : \phi \in \calX,\ \gamma_2(\phi)>0 \right\}.
    \]
    Since this holds for every $a\geq 1$, the claimed formula follows. \qedhere
\end{proof}

The characterization in \Cref{prop:acr_func_pow} is obtained directly from Strassen duality, and \emph{a priori} looks quite different from the one in \Cref{thm:acr_formula}, which comes from the abstract minimization method and the analysis of edge support functionals. Comparing the two formulas suggests that there should be an underlying comparison between general spectral points and edge support functionals. We now make this comparison explicit.

\begin{proposition} \label{prop:phi_vs_edge}
    Let $\phi\in \calX$ satisfy $\gamma_2(\phi)>0$. Then the interval
    \[
        [0, 1] \cap \left[1 - \frac{\gamma_3(\phi)}{\gamma_2(\phi)}, \frac{\gamma_1(\phi)}{\gamma_2(\phi)}\right]
    \]
    is nonempty. Moreover, for every $\rho$ in this interval, we have
    \[ \phi(T) \geq \left(\zeta^{(\rho,1-\rho,0)}(T)\right)^{\gamma_2(\phi)} \]
    for all tensors $T$.
\end{proposition}
\begin{proof}
    Since every point of $\specMM$ lies in $[0,1]^3$ and satisfies $\gamma_1+\gamma_2+\gamma_3\geq 2$, we have
    \[
        \gamma_1(\phi)+\gamma_3(\phi)\geq 2-\gamma_2(\phi)\geq 1\geq \gamma_2(\phi).
    \]
    Hence
    \[
        1 - \frac{\gamma_3(\phi)}{\gamma_2(\phi)} \leq \frac{\gamma_1(\phi)}{\gamma_2(\phi)},
    \]
    while the left endpoint is at most $1$ and the right endpoint is at least $0$. This proves the interval is nonempty.

    Let $\theta=(\rho,1-\rho,0)$. By the assumption on $\rho$, we have
    \[
        \zeta^\theta(T)^{\gamma_2(\phi)} \stackrel{\mathrm{Cor.~\ref{cor:edge-universal}}}{=} \xi_{(\rho, 1, 1-\rho)}(T)^{\gamma_2(\phi)} = \xi_{\gamma_2(\phi)\cdot (\rho, 1, 1-\rho)}(T) \leq \xi_{\gamma(\phi)}(T) \stackrel{\mathrm{Cor.~\ref{cor:value_leq_func}}}{\leq} \phi(T). \qedhere
    \]
\end{proof}

\subsection{Lower support functional}

We briefly discuss Strassen's \emph{lower support functional} \cite[\S3]{Str91SuppFunc}. For
$\theta\in\Theta$, define
\[ \zeta_\theta(T) = \max_{\mathrm{bases~of~}U,V,W} 2^{\eH_\theta(\max \supp T)}, \]
where $\max \supp T$ denotes the set of maximal elements of $\supp T$ in the product partial order on
$[n]\times[m]\times[k]$. Strassen proved that $\zeta_\theta$ is normalized, monotone, superadditive, and supermultiplicative, and moreover that
\[ \zeta_\theta \leq \zeta^\theta. \]
He also showed that equality holds for oblique tensors, namely tensors for which some choice of basis satisfies
\[ \supp T = \max \supp T. \]
This was used to show that, on the class of oblique tensors, the upper support functional $\zeta^\theta$ is itself a spectral point.

The lower and upper support functionals are nevertheless genuinely different in general. B\"urgisser proved \cite{Bur90Thesis} (see also \cite[\S7 (iii)]{Str91SuppFunc}) that for generic tensors they are separated whenever $\theta$ lies in the relative interior of $\Theta$, that is, when $\theta_\varkappa>0$ for all $\varkappa$. Thus the remaining open case is the boundary. In view of the rigidity results proved in this paper for support functionals whose parameter lies on an edge of $\Theta$, it is natural to conjecture that equality should still hold there.

\begin{conjecture}
    Over an algebraically closed field, one has $\zeta_\theta = \zeta^\theta$ for every $\theta\in\Theta_{(\varkappa)}$.
\end{conjecture}

We can verify the asymptotic version of this conjecture.
\begin{proposition}
    For every $\theta\in\Theta_{(\varkappa)}$, the regularization of the lower support functional
    \[ \utilde{\zeta_\theta}(T) \coloneq \lim_{n\to\infty} \zeta_\theta(T^{\otimes n})^{1/n} \]
    exists and satisfies
    \[ \utilde{\zeta_\theta} = \zeta^\theta. \]
\end{proposition}
\begin{proof}
    Since $\zeta_\theta$ is supermultiplicative, the limit exists by Fekete's lemma.

    By Strassen's inequality $\zeta_\theta\leq \zeta^\theta$ \cite[Corollary~4.3]{Str91SuppFunc}, and since $\zeta^\theta$ is spectral on every edge by \Cref{cor:edge-universal}, we immediately get
    \[ \utilde{\zeta_\theta}(T) \leq \zeta^\theta(T). \]

    It remains to prove the reverse inequality. By symmetry it suffices to consider
    \[ \theta=(\rho,1-\rho,0)\in\Theta_{(3)}, \]
    and write
    \[ \gamma=(\rho,1,1-\rho)=\gamma(\zeta^\theta). \]
    Matrix multiplication tensors are oblique, so for every $E,H,L$ we have
    \[ \zeta_\theta(\ang{E,H,L}) = \zeta^\theta(\ang{E,H,L}) = E^{\gamma_1}H^{\gamma_2}L^{\gamma_3}. \]
    Therefore, whenever $T^{\otimes N}\geq \ang{E,H,L}$, monotonicity of $\zeta_\theta$ gives
    \[ \zeta_\theta(T^{\otimes N}) \geq E^{\gamma_1}H^{\gamma_2}L^{\gamma_3}. \]
    Taking $N$-th roots and then the supremum over all such restrictions, we obtain
    \[ \utilde{\zeta_\theta}(T) \geq \xi_\gamma(T). \]
    Finally, \Cref{cor:edge-universal} gives
    \[ \xi_\gamma = \zeta^\theta \]
    on the edge $\Theta_{(3)}$. Hence $\utilde{\zeta_\theta}(T)\geq \zeta^\theta(T)$, completing the proof.
\end{proof}

\section{Efficiently computing the edge functionals} \label{sec:algorithm}

By our characterization of the edge support functionals, for $\theta \in \Theta_{(\varkappa)}$, $\zeta^\theta(T)$ is determined by the HN-filtration of the quiver representation associated to $T$. By combining this with the following algorithm we can efficiently compute $\zeta^\theta(T)$.

\begin{theorem}[Theorem 3.7, \cite{Cheng24HN}, specialized to Kronecker quivers]\label{thm:hn_alg}
	Let $T \in \bbF^{n \times m \times k}$ be a tensor over an infinite field. The HN-filtration of $T$ (with slope $\mu^\star$ as in \Cref{sec:quiver}) can be computed in time polynomial in $n$, $m$, and $k$, assuming unit-cost arithmetic over $\bbF$.
\end{theorem}
\begin{corollary}\label{cor:alg}
Let $\bbF$ be arbitrary, $T \in \bbF^{n \times m \times k}$, and
\[
\theta=(\rho,1-\rho,0), \qquad \rho=\frac{a}{b}\in [0,1]\cap \bbQ.
\]
Then $\zeta^\theta(T)$ can be approximated to absolute error $2^{-\ell}$ in time polynomial in $n,m,k,\ell,\log a$, and $\log b$.
\end{corollary}

\begin{proof}
First, a technical point: if $\bbF$ is a small finite field, we assume access to efficient computation in an extension field of size $(nmk)^{\Omega(1)}$. Although \Cref{thm:hn_alg} is stated for infinite fields, it turns out that the matrix-space algorithms used in its proof extend to fields with sufficiently many elements; see \cite[Remark A.3]{Cheng24HN} and \cite[Theorem 1.5]{IQS18NCR}. In our case $|\bbF| \ge (nmk)^{\Omega(1)}$ suffices. This does not change the output of our algorithm: since semistability is preserved by base change (\Cref{prop:base_change_stability}), the HN-filtration is invariant under base change (cf.~\cite[Proposition 2.4]{HS20Stab}).

Using \Cref{thm:hn_alg} we obtain the dimension data $\{(n_u,m_u)\}_{u=1}^r$ of the HN-filtration in polynomial time. Finally, we compute
\[
\sum_{u=1}^r n_u^\rho m_u^{1-\rho}
\]
to absolute error $2^{-\ell}$. Each term $n_u^\rho m_u^{1-\rho}$ is evaluated as
\[
\exp\!\left(\rho \log n_u + (1-\rho)\log m_u\right).
\]
Standard high-precision algorithms compute $\log$ and $\exp$ of positive rational inputs to $s$ bits in time polynomial in $s$ and the input bitlengths. Since $n_u^\rho m_u^{1-\rho}\le \max\{n_u,m_u\}$, taking
$s=\ell+O(\log r+\log n_u+\log m_u)$ suffices to ensure total absolute error at most $2^{-\ell}$ after summing over all $u$.
\end{proof}

\begin{remark}
	The above algorithms assume unit-cost arithmetic over $\bbF$. For the case of finite fields, this implies efficient algorithms in the Turing model under any standard encoding of field elements. Over $\bbQ$, polynomial-time Turing complexity follows from the bit-complexity analysis of the shrunk subspace subroutine of \cite[Theorem 1.5]{IQS18NCR}. Alternatively, for tensors over $\bbC$ with rational or Gaussian integer entries, we may use \cite{IOS25HN}, which explicitly gives a polynomial-time HN-filtration algorithm in the Turing model.
\end{remark}

\section*{Acknowledgments}

We thank anonymous FOCS reviewers and Jeroen Zuiddam for helpful comments on earlier drafts of this paper.

\bibliographystyle{alphaurl}

\bibliography{main.bib}

\end{document}